\documentclass{cmslatex}
\usepackage[paperwidth=7in, paperheight=10in, margin=.875in]{geometry}
 \usepackage[backref,colorlinks,linkcolor=red,anchorcolor=green,citecolor=blue]{hyperref}
\usepackage{amsfonts,amssymb}
\usepackage{amsmath,bm}
\usepackage{graphicx}
\usepackage{cite}
\usepackage{enumerate}
\renewcommand{\theequation}{\arabic{section}.\arabic{equation}}

\usepackage{mathrsfs}
   \allowdisplaybreaks

\usepackage[utf8x]{inputenc}
\usepackage{esint}

\usepackage{graphicx}
\usepackage{enumitem}[2011/09/28]
\setenumerate{align=left, leftmargin=0pt,labelsep=.5em, labelindent=0\parindent,listparindent=\parindent,itemindent=*}

\usepackage{amssymb,bm,amsmath}
\usepackage{amsfonts}

\usepackage[OT2,T1,T2A]{fontenc}

\usepackage[english,french,german,italian,russian]{babel}
\usepackage{appendix}

\DeclareMathOperator{\Cl}{Cl}

\DeclareMathOperator{\D}{d}
\DeclareMathOperator{\I}{Im}
\DeclareMathOperator{\R}{Re}

\DeclareMathOperator{\Div}{div}

\DeclareMathOperator{\Si}{Si}
\DeclareMathOperator{\Ci}{Ci}
\DeclareMathOperator{\Ei}{Ei}

\def\eor{\hfill$ \square$}

\usepackage{colonequals}
\usepackage{color}

\usepackage{tikz}
\tikzset{>=stealth}
\usepackage{pgfplots}
\begin{document}

\selectlanguage{english}
\title{Quantitative spectral analysis of electromagnetic scattering.\ II: Evolution semigroups and non-perturbative solutions}
\author{Yajun Zhou\thanks{Program in Applied and Computational Mathematics, Princeton University, Princeton, NJ 08544, USA; Academy of Advanced Interdisciplinary Studies (AAIS), Peking University, Beijing 100871, P. R. China, (yajunz@math.princeton.edu, yajun.zhou.1982@pku.edu.cn). } }

\maketitle


\begin{abstract}
We carry out quantitative studies  on the Green operator $ \hat{\mathscr G}$ associated with the Born equation, an  integral equation
that models electromagnetic scattering, building the strong stability of the evolution semigroup $\{\exp(-i\tau\hat{\mathscr G})|\tau\geq0\} $  on polynomial compactness and the Arendt--Batty--Lyubich--V\~u theorem. The strongly-stable evolution semigroup inspires our proposal of  a non-perturbative method  to solve the light scattering problem and improve the Born approximation. \end{abstract}

\section{Introduction}
A great variety of physical applications \cite{HulstMeteo,vandeHulst,Aydin,Bruno2003,RICM} rely on numerical modeling of the interaction between  electromagnetic waves   and dielectric media \cite{Rayleigh,Mie,AsanoYamamoto,ColtonKress1983,ColtonKress,Sh,FDTD,FDTDbook,DDA,Goedecke,Yurkina,Yurkinb,Bruno2009}.  In direct scattering problems \cite{ColtonKress1983}, one predicts the scattering pattern from the shape and physical constituent of the scatterer; in inverse scattering problems \cite{ColtonKress}, one attempts to infer the geometric and physical information of the scattering media from the measured scattering field.

Specific examples of exact solutions to direct light scattering, in terms of an infinite series that involves special functions, do exist for cylindrical \cite{Rayleigh}, spherical \cite{Mie}, spheroidal \cite{AsanoYamamoto} and Chebyshev \cite{Sh} dielectric particles. In practice, these infinite series  do not have a closed-form sum, so their truncations lead to  approximate solutions. The special techniques employed in these series solutions may not generalize well to the analysis of arbitrarily shaped dielectrics.

There are a wealth of numerical recipes for practical solutions to direct light scattering on  arbitrarily shaped dielectric particles. There are two categories of grid-based algorithms  \cite{FDTD,FDTDbook,DDA,Goedecke,Yurkina} for  simulating the propagation of electromagnetic waves in the presence of   dielectric particles with complicated shapes. The first category of methods are based on numerical solutions of partial differential equations for electromagnetic scattering, such as discretization of the Maxwell equations  in the  ``finite difference time domain'' (FDTD) \cite{FDTD,FDTDbook} algorithms. The second category of methods are based on numerical solutions of integral equations for electromagnetic scattering. The  ``discrete dipole approximation'' (DDA), also known as the ``digitized Green's function algorithm'' \cite{DDA, Goedecke}, aim at a volume integral equation. Krylov subspace iterations \cite{Bruno2009} can be used for accelerated solution to boundary integral equations for scattering problems on metallic conductors.
  Typically, for a given (non-conducting) dielectric shape with known refractive index, an efficient   numerical simulation (either following the FDTD or DDA algorithm) requires a  grid  spacing not coarser than one tenth of the wavelength. As a result,  in  grid-based simulations for large particles, the computational cost could become exorbitantly high, and the  risks of numerical instability may ensue \cite{Yurkinb}.

Setting aside the concerns about affordability and stability of numerical solutions, we might still find \textit{ad hoc} simulations of electromagnetic scattering insufficient for certain types of applications that involve inverse scattering problems.  For example, in reflection interference contrast microscopy \cite{RICM}, an approximate analytic expression of the scattering field would be suitable for curve fitting procedures that solve the inverse scattering problem on the fly. It would be less desirable to conduct extensive numerical simulations,  trial after trial, for the information inference of the scattering medium.  Besides optical microscopy~\cite{RICM}, certain applications in meteorological radar~\cite{vandeHulst}, geological prospection~\cite{EIT,ColtonKress1983} also require the solution of an inverse scattering problem, such as deducing the shape and physical constitution from the far-field scattering patterns, and/or deducing the spatial location of the scattering medium from the near-field scattering patterns. Therefore, there are wide applications calling for an approximate analytic understanding of electromagnetic scattering with sufficient numerical accuracy.

In Paper I \cite{QuantEM1Norm} of this series, we performed a quantitative assessment of the error bound in  the perturbative solution (Born approximation) to the light scattering problem. In this work (Paper II), we  propose a non-perturbative alternative that improves the convergence rate of Born approximation. The major idea underlying our proposal is to exploit an evolution semigroup $\{\exp(-i\tau\hat{\mathscr G})|\tau\geq0\} $ generated from the Green operator $ \hat {\mathscr G}$ of electromagnetic scattering. In particular, we shall use the spectral analysis of the Green operator $\hat{\mathscr G} $ to deduce asymptotic behavior of the evolution semigroup, in connection to the asymptotic solutions of the electromagnetic scattering problem. \section{Statement of results}We recapitulate some previous studies (\S\ref{subsec:qual_summary}) on electromagnetic scattering, on which the current developments  (\S\ref{subsec:evolv_semigrp_nonpert_sln_summary}) are based.
\subsection{Discreteness of optical resonance modes and Born approximation\label{subsec:qual_summary}}As in  \cite{QualEM,QuantEM1Norm}, we model   the electromagnetic scattering problem  by the \textit{Born equation}, an exact consequence of the Maxwell equations:\begin{align}(1+\chi)\bm E(\bm r)=\bm E_{\mathrm{inc}} (\bm r)+\chi\nabla\times \nabla\times\iiint_V\frac{\bm E(\bm r') e^{-ik|\bm r-\bm r' |}}{4\pi|\bm r-\bm r' |}\D^3 \bm r',\quad\bm r\in V.\label{eq:Born_eqn}\end{align}Here $\bm E(\bm r),\bm r\in V$ is the dielectric response from a scatterer with dielectric susceptibility  $\chi$,  occupying a bounded open volume $V$ with smooth boundary $ \partial V$, and connected exterior $ \mathbb R^3\smallsetminus(V\cup\partial V)$.   The incident  beam has wavelength  $2\pi/k$, and  is represented by its electric field  $ \bm E_\mathrm{inc}(\bm r)$.
At times, we will abbreviate the Born equation \eqref{eq:Born_eqn} as $ \hat{\mathscr B}\bm E=(\hat I-\chi\hat{\mathscr G})\bm E=\bm E_\mathrm{inc} $, and refer to  $\hat{\mathscr B}$ as the \textit{Born operator},  $\hat{\mathscr G}$ the \textit{Green operator}.

It was pointed out in  \cite{QualEM} that the Green operator $ \hat{\mathscr G}:L^2(V;\mathbb C^3)\longrightarrow L^2(V;\mathbb C^3)$  exhibits certain non-physical spectral behavior: some resonance modes that belong to the Hilbert space $ L^2(V;\mathbb C^3)$ (the totality of square-integrable $ \mathbb C^3$-valued vector fields defined in $V$) fail the transversality condition $\nabla\cdot\bm E=0 $, thus do not mark the actual singularity of the Maxwell equations.    For a correct characterization of the optical resonance modes, one needs to restrict the domain of the Green operator to the Hilbert subspace $\Phi(V;\mathbb C^3)=\Cl(C^\infty(V;\mathbb C^3)\cap\ker(\nabla\cdot)\cap L^2(V;\mathbb C^3))$, which is the smallest Hilbert space that contains all  the smooth, divergence-free and square-integrable complex-valued vector fields (conventional notation:  $ \Phi(V;\mathbb C^3)=H(\Div0 ,V)$ \cite[p.~215]{DautrayLionsVol3}).
It was shown  in  \cite[\S3]{QualEM} that the   quadratic operator polynomial  $ \hat{\mathscr G}(\hat I+2\hat{\mathscr G}):\Phi(V;\mathbb C^3)\longrightarrow \Phi(V;\mathbb C^3)$  is compact
where the cubic  polynomial  $ \hat{\mathscr G}(\hat I+2\hat{\mathscr G})^{2}:\Phi(V;\mathbb C^3)\longrightarrow \Phi(V;\mathbb C^3)$  defines a Hilbert--Schmidt operator.
It has been further shown \cite[\S3.2]{QualEM}  that the spectrum of   $ \hat{\mathscr G}:\Phi(V;\mathbb C^3)\longrightarrow \Phi(V;\mathbb C^3)$  is a countable set. Except for two points $ \{0,-1/2\}$, all the other points in the spectrum are eigenvalues with a strictly negative imaginary part.

As a popular   perturbative solution to the electromagnetic scattering problem, the (first) \textit{Born approximation} is a truncation of the Neumann series $ \bm E=(\hat I-\chi\hat{\mathscr G})^{-1}\bm E_\mathrm{inc} =\bm E_\mathrm{inc}+\chi\hat{\mathscr G}\bm E_\mathrm{inc}+\chi^2\hat{\mathscr G}^2\bm E_\mathrm{inc}+\cdots$ up to the $ O(\chi)$ term.  In the perturbative regime $ |\chi|<1/\Vert \hat{\mathscr G}\Vert_{L^2(V;\mathbb C^3)}$, the error bound for the Born approximation is given by \begin{align}\left\Vert\bm E-(\hat I+\chi\hat{\mathscr G} )\bm E_\mathrm{inc}  \right\Vert_{L^2(V;\mathbb C^3)}\leq\frac{|\chi|\Vert\hat{\mathscr G}\Vert_{L^2(V;\mathbb C^3)}\Vert\bm E_\mathrm{inc}  \Vert_{L^2(V;\mathbb C^3)}}{1-|\chi|\Vert \hat{\mathscr G}\Vert_{L^2(V;\mathbb C^3)}}.\end{align}
In  \cite[Theorem 2.1]{QuantEM1Norm}, we provided  an estimate of the operator norm  $ \Vert \hat{\mathscr G}\Vert_{L^2(V;\mathbb C^3)}$.

\subsection{\label{subsec:evolv_semigrp_nonpert_sln_summary}Evolution semigroups and non-perturbative solutions}
Let  $\{\exp(-i\tau\hat{\mathscr G})\colonequals \hat I+\sum_{s=1}^\infty(-i\tau\hat{\mathscr G})^s/s!|\tau\geq0\}$ be an evolution semigroup  with infinitesimal generator $-i\hat{\mathscr G} $. This is a contractive semigroup for  either     $ \hat{\mathscr G}:L^{2}(V;\mathbb C^3)\longrightarrow L^{2}(V;\mathbb C^3)$  or    $ \hat{\mathscr G}:\Phi(V;\mathbb C^3)\longrightarrow \Phi(V;\mathbb C^3)$, thanks to the energy conservation law in electromagnetic scattering (\S\ref{sec:erg_conserv}). In addition, the spectral properties of the Green operator     $ \hat{\mathscr G}:\Phi(V;\mathbb C^3)\longrightarrow \Phi(V;\mathbb C^3)$ lead to richer structures of the semigroup in question. \begin{theorem}[Strong stability and generalized skin effect]\label{thm:semigrp_prop}\begin{enumerate}[label=\emph{(\roman*)},ref=(\roman*),widest=ii] \item\label{itm:strong_stab} If the smooth  dielectric  has a connected exterior volume   $\mathbb R^3\smallsetminus(V\cup\partial V)$, then we have    $\sigma^\Phi(\hat{\mathscr G})\cap\mathbb R\subset\{0,-1/2\}$  and the related evolution semigroup $\{\exp(-i\tau\hat{\mathscr G})|\tau\geq0\} $ is strongly stable:
\begin{align}\label{eq:strong_stab_statement}\lim_{\tau\to+\infty}\Vert\exp(-i\tau\hat{\mathscr G})\bm F\Vert_{L^2(V;\mathbb C^3)}=0,\quad\forall \bm F\in\Phi(V;\mathbb C^3).\end{align}
\item \label{itm:skin_effect} Suppose that $ \bm E_\mathrm{inc}\in \Phi(V;\mathbb C^3)$ is a transverse incident wave satisfying the Helmholtz equation $ (\nabla^2+k^2)\bm E_\mathrm{inc}=\mathbf 0$ in the sense of distributional derivatives, then the following limit holds \begin{align}\lim_{|\chi|\to+\infty}|\chi|\iiint_Vf(\bm r)((\hat I+i|\chi| \hat{\mathscr G})^{-1}\bm E_\mathrm{inc})(\bm r)\D^3\bm r=\mathbf0\label{eq:GSE}\end{align}for every compactly supported smooth function $f\in C^\infty_0(V;\mathbb C) $.
\end{enumerate}\end{theorem}

As will be revealed in \S\ref{sec:evolv_semigrps}, Theorem~\ref{thm:semigrp_prop}\ref{itm:strong_stab} is a direct consequence of the Arendt--Batty--Lyubich--V\~u theorem (see  \cite{AB88,LV88}, as well as  \cite[pp.~326--327]{erg}); while Theorem~\ref{thm:semigrp_prop}\ref{itm:skin_effect} is a corollary of the strong stability.

From the integral representation of the solution to the Born equation\begin{align}\bm E= (\hat I-\chi\hat{\mathscr G})^{-1}\bm E_\mathrm{inc}=\frac{1}{i\chi}\int_0^{+\infty}e^{i\tau/\chi}\exp(-i\tau\hat{\mathscr G})\bm E_\mathrm{inc}\D \tau,\quad\I\chi<0,\end{align} one can show that the perturbative Born approximation  $ \bm E=(\hat I-\chi\hat{\mathscr G})^{-1}\bm E_\mathrm{inc} \approx\bm E_\mathrm{inc}+\chi\hat{\mathscr G}\bm E_\mathrm{inc}$ respects the short-term behavior $ \exp(-i\tau\hat{\mathscr G})=\hat I-i\tau\hat{\mathscr G}+O(\tau^{2})$, but is incompatible with the long-term strong stability  $ \lim_{\tau\to+\infty}\Vert\exp(-i\tau\hat{\mathscr G})\bm F\Vert_{L^2(V;\mathbb C^3)}=0$ of the evolution semigroup for electromagnetic scattering.

In \S\ref{sec:nonpert_approach}, we propose a non-perturbative alternative to the Born approximation  as follows:\begin{align}&((\hat I-\chi\hat{\mathscr G})^{-1}\bm E_\mathrm{inc})(\bm
r)\notag\\\sim{}&\bm E_\mathrm{inc}(\bm
r)+\chi k^2\iiint_V\frac{e^{-i\sqrt{1+\chi}k|\bm
r-\bm r'|}}{4\pi|\bm
r-\bm r'|}\bm E_\mathrm{inc}(\bm
r')\D^3\bm r'\notag\\{}&-\frac{\chi}{1+\chi}\nabla\varoiint_{\partial V}\frac{\bm n'\cdot\bm E_\mathrm{inc}(\bm
r')e^{-i\sqrt{1+\chi}k|\bm
r-\bm r'|}}{4\pi|\bm
r-\bm r'|}\D^3\bm r',\quad \bm r\in V.\end{align}Such a formula recovers the Born approximation in the perturbative regime, while honoring the  strong stability  $ \lim_{\tau\to+\infty}\Vert\exp(-i\tau\hat{\mathscr G})\bm F\Vert_{L^2(V;\mathbb C^3)}=0$ in the long run. As a concrete application, we evaluate the non-perturbative approximation to the forward scattering amplitude of Mie scattering for $ \I\chi<0,\I n=\I\sqrt{1+\chi}<0$:\begin{footnotesize} \begin{align}&-\chi k\iiint_{|\bm r|<R}\bm E_\mathrm{inc}^{*}(\bm r)\cdot(\hat I-\chi\hat{\mathscr G})^{-1}\bm E_\mathrm{inc}(\bm r)\D^3\bm r\notag\\\sim{}&2\pi inR^{2}+\frac{\pi i(n+1)^2}{(n-1)^2}\frac{1-e^{-2i(n-1)kR}[1+2i(n-1)kR]}{4k^{2}}+\frac{\pi i(n-1)^{2}}{(n+1)^2}\frac{e^{-2i(n+1)kR}[1+2i(n+1)kR]-1}{4k^{2}} \notag\\&+\frac{\pi}{16 k^2 n^2}   \left\{-2 i\chi^{2} [2 (\chi+2)  k^2R^2-1] \left[\Ei(-2 i  (n-1)k R)-\Ei(-2 i  (n+1)k R)+\log\frac{n+1}{n-1}\right]\right.\notag\\&+4 i n (2 \chi ^{2}k^2  R^2-\chi-2)+ e^{-2 i  (n+1)k R}(n-1)^2 \left[2  (n+1) (\chi+2)k R+i (n^{2}+4n+1)\right]\notag\\&\left.- e^{-2 i  (n-1)k R}(n+1)^2  \left[2 (n-1)(\chi+2)k R+i (n^{2}-4n +1)\right]\phantom{\frac12}\hspace{-1em}\right\},\label{eq:nonperturb}\end{align}\end{footnotesize}and compare it to the exact solution (Mie series) along with various well-known approximation formulae in the physical literature. Here in \eqref{eq:nonperturb}, for a special incident beam $\bm E_\mathrm{inc}(\bm r)=\bm e_x\exp(-ikz) $, we use a superscripted asterisk to denote complex conjugation and we have $ \Ei(z):=-\int_{-z}^\infty e^{-t}\frac{\D t}{t}$.

\section{Energy conservation in electromagnetic scattering\label{sec:erg_conserv}}
The Born operator $\hat{\mathscr B}=\hat I-\chi\hat {\mathscr G}:L^2(V;\mathbb C^3)\longrightarrow L^2(V;\mathbb C^3) $ satisfies  an energy conservation law (generalized optical theorem): \begin{align}\sigma_{\mathrm{sc}} \colonequals
\frac{|\chi|^2 k^4}{16\pi^{2}\ } \varoiint_{|\bm n|=1}\left|\bm n\times\iiint_V \bm E(\bm r' ) e^{ik\bm n \cdot \bm r' }  \D^3\bm r' \right|^2\D\Omega=\I\left(\chi k\langle \bm E-\hat{\mathscr  B}\bm E,\bm E\rangle_V \right),\label{eq:GOT1}
\end{align}where $ \D\Omega=\sin\theta\D\theta\D\phi$ stands for the infinitesimal steric angles, and  $ \langle \bm F,\bm G\rangle_V\colonequals \iiint_V\bm F^*(\bm r)\cdot\bm G(\bm r)\D^3\bm r$
denotes the inner product on the Hilbert space $ L^2(V;\mathbb C^3)$. In \cite[Theorem 2.1]{QualEM}, we proved  \eqref{eq:GOT1} using Fourier analysis. In the opening paragraph of \cite[\S4.2]{QuantEM1Norm}, we gave a physical interpretation of   \eqref{eq:GOT1}  as redistribution of the electromagnetic work done by the incident field into dissipated and scattered energies.

In  \cite[\S2.2]{QualEM}, the energy conservation law \eqref{eq:GOT1} played a decisive r\^ole in the proofs of ``uniqueness theorems'' for light scattering. Later in \S\ref{sec:evolv_semigrps} of the current work, we shall use   \eqref{eq:GOT1} again to justify an integral representation for $ \bm E= (\hat I-\chi\hat{\mathscr G})^{-1}\bm E_\mathrm{inc}$ in terms of the evolution semigroup $ \{\exp(-i\tau\hat{\mathscr G})|\tau\geq0\}$. Apart from this, the generalized optical theorem \eqref{eq:GOT1} and its equivalent forms will be  used elsewhere in this article, whenever a discussion on the total scattering cross-section $\sigma_{\mathrm{sc}} $ is needed.

\section{Evolution semigroups\label{sec:evolv_semigrps}}The developments in this section are motivated by  the following integral representation of the solution to the light scattering problem: \begin{align}\bm E= (\hat I-\chi\hat{\mathscr G})^{-1}\bm E_\mathrm{inc}=\frac{1}{i\chi}\int_0^{+\infty}e^{i\tau/\chi}\exp(-i\tau\hat{\mathscr G})\bm E_\mathrm{inc}\D \tau,\quad\I\chi<0.\label{eq:IRS}\end{align}Here, $\{\exp(-i\tau\hat{\mathscr G})\colonequals \hat I+\sum_{s=1}^\infty(-i\tau\hat{\mathscr G})^s/s!|\tau\geq0\}$ is an evolution semigroup  with infinitesimal generator $-i\hat{\mathscr G} $. (In our context, the variable $\tau$ is dimensionless and has nothing to do with the elapse of physical time, so the term ``evolution'' has only a formal meaning.)
If we define $\bm \psi_{\tau}(\bm r)\equiv\bm \psi(\bm r;\tau)\colonequals (\exp(-i\tau\hat{\mathscr G})\bm E_\mathrm{inc})(\bm r) $, then we have the evolution equation \begin{align}i\frac{\partial\bm \psi(\bm r;\tau)}{\partial\tau}=\hat{\mathscr G}\bm \psi(\bm r;\tau)=\nabla\times\nabla\times\iiint_V\frac{\bm \psi(\bm r';\tau)e^{-ik|\bm r-\bm r'|}}{4\pi|\bm r-\bm r'|}\D^3\bm r'-\bm \psi(\bm r;\tau)\label{eq:EqMotion}\end{align} with initial condition $\bm \psi_{0}(\bm r)\equiv\bm \psi(\bm r;0)=\bm E_\mathrm{inc}(\bm r)  $ equal to the incident field.

The right-hand side of \eqref{eq:IRS} is a Bochner--Dunford integral \cite{Bochner,Dunford,Hille1942} (vector-valued integral in infinite-dimensional linear spaces), so it is a non-trivial extension of the scalar-valued integration formula
\begin{align}(1-\chi G)^{-1}=\frac{1}{i\chi}\int_0^{+\infty}e^{i\tau/\chi}\exp(-i\tau G)\D \tau,\quad\I\chi<0.\end{align}

To rigorously justify the integral representation in  \eqref{eq:IRS}, we need to check the Hille--Yosida--Lumer--Phillips criteria \cite{Hille1948,Yosida1948,LP1961,Yosida,LaxPhillips,erg} in the theory of evolution semigroups for infinite-dimensional Hilbert spaces. In simple terms, the Hille--Yosida--Lumer--Phillips criteria  boil down to two parts in the current problem: \begin{enumerate}[label=\arabic*.,widest=2]
\item
The operator $ (\hat I-\chi\hat{\mathscr G})^{-1}$ is non-singular in the entire open lower-half plane $ \I\chi<0$;
\item The energy of $\bm \psi_{\tau}$ does not increase as $\tau$ elapses, \textit{i.e.}
\begin{align}\iiint_V|\bm\psi_{\tau_1}(\bm r)|^{2}\D^3\bm r\geq \iiint_V|\bm\psi_{\tau_2}(\bm r)|^{2}\D^3\bm r,\quad 0\leq\tau_{1}\leq \tau_2<+\infty.\end{align}
\end{enumerate}  Here, Part~1 is guaranteed by the spectral analysis in  \cite{QualEM}. Part~2 follows from the equation of motion \eqref{eq:EqMotion}  as in the following computation \begin{align}\label{eq:moddiff}\frac{\D}{\D\tau}\langle\bm \psi_\tau,\bm \psi_\tau\rangle_{V}={}&-2\I\left\langle i\frac{\partial\bm \psi_\tau}{\partial\tau},\bm \psi_\tau\right\rangle_{V}=-2\I\langle \hat{\mathscr G}\bm \psi_\tau,\bm \psi_\tau\rangle_{V}\notag\\={}&-\frac{k^3}{8\pi^2 }\varoiint_{|\bm n|=1}\left\vert\bm n\times\iiint_V \bm \psi_\tau(\bm r' ) e^{ik\bm n \cdot \bm r' }  \D^3\bm r' \right\vert^2\D\Omega\leq0.\end{align} In the last line of \eqref{eq:moddiff}, we have use the generalized optical theorem \eqref{eq:GOT1}.
\subsection{Strong stability}\begin{proof}[Proof of Theorem~\ref{thm:semigrp_prop}\ref{itm:strong_stab}] According to the qualitative spectral analysis summarized in \S\ref{subsec:qual_summary},  we have $\sigma_p^\Phi(\hat{\mathscr G})\cap\mathbb R =\varnothing $ if the exterior volume   $\mathbb R^3\smallsetminus(V\cup\partial V)$ is  connected. Here, $\sigma_p^\Phi(\hat{\mathscr G}) $ stands for the  point spectrum (totality of eigenvalues) of the bounded linear operator $ \hat{\mathscr G}:\Phi(V;\mathbb C^3)\longrightarrow\Phi(V;\mathbb C^3)$. Using the optical resonance theorem \cite[Theorem 3.1 and Proposition 3.4]{QualEM}, we  obtain     $\sigma^\Phi(\hat{\mathscr G})\cap\mathbb R=\{0,-1/2\}$. Now, the absence of point spectrum on the real axis $\sigma^\Phi_{p}(\hat{\mathscr G})\cap \mathbb R=\varnothing $ and the countability of the spectrum on the real axis $\sigma^\Phi(\hat{\mathscr G})\cap\mathbb R=\{0,-1/2\} $ will allow us to deduce the strong stability \begin{align}\lim_{\tau\to+\infty}\Vert\exp(-i\tau\hat{\mathscr G})\bm F\Vert_{L^2(V;\mathbb C^3)}=0,\forall \bm F\in\Phi(V;\mathbb C^3)\end{align} of the operator semigroup $\{\exp(-i\tau\hat{\mathscr G})|\tau\geq0\}$, according to the Arendt--Batty--Lyubich--V\~u theorem (see  \cite{AB88,LV88}, as well as  \cite[pp.~326--327]{erg}).
 \end{proof}\begin{remark}We note that the evolution semigroup  provides some additional insights regarding the efficacy of the perturbative Born series.
From \eqref{eq:moddiff}, we know that the $ L^2$-norm $\Vert\exp(-i\tau\hat{\mathscr G})\bm E_\mathrm{inc} \Vert_{L^2(V;\mathbb C^3)}\colonequals \sqrt{\langle\exp(-i\tau\hat{\mathscr G})\bm E_\mathrm{inc}, \exp(-i\tau\hat{\mathscr G})\bm E_\mathrm{inc}\rangle_{V}}$ should not increase in $\tau$, but the truncated Taylor series of $ \exp(-i\tau\hat{\mathscr G})\bm E_\mathrm{inc}$ (and accordingly, the truncated Born series) may not necessarily honor this monotone energy decay.

To see this, we first note that a natural bound estimate of the typical amplitude in each term of the Taylor expansion for $ \exp(-i\tau\hat{\mathscr G})\bm E_\mathrm{inc}$ is given by
\begin{align}\left\Vert\frac{(-i\tau\hat{\mathscr G})^s}{s!} \bm E_\mathrm{inc}\right\Vert_{L^2(V;\mathbb C^3)}\leq\frac{\tau^s\Vert\hat{\mathscr G} \Vert_{L^2(V;\mathbb C^3)} ^s}{s!}\Vert\bm E_\mathrm{inc} \Vert_{L^2(V;\mathbb C^3)}\end{align}where $\Vert\hat{\mathscr G} \Vert_{L^2(V;\mathbb C^3)}  $ is the operator norm  of  $ \hat{\mathscr G}$. Only for $s\gtrsim\tau \Vert\hat{\mathscr G} \Vert_{L^2(V;\mathbb C^3)}$ does the upper bound estimates of individual terms decay rapidly in $s$. In other words, as $ \tau$ increases, there are about $\tau \Vert\hat{\mathscr G} \Vert_{L^2(V;\mathbb C^3)} $ terms of the Taylor expansion   that significantly  contribute to the sum $ \exp(-i\tau\hat{\mathscr G})\bm E_\mathrm{inc}$. The net result is the cancellation of many large terms that gives rise to a small quantity consistent with the energy decay of  $ \exp(-i\tau\hat{\mathscr G})\bm E_\mathrm{inc}$. Therefore, a truncated Taylor expansion of  $ \exp(-i\tau\hat{\mathscr G})\bm E_\mathrm{inc}$ may severely misrepresent the behavior of the integrand of \eqref{eq:IRS} beyond the short-term ($\tau\to0^+ $) regime, and the according truncated Born series  for $ (\hat I-\chi\hat{\mathscr G})^{-1}\bm E_\mathrm{inc}  $ may not give accurate enough approximations to the solution.

In the light of this, a more sensible way to improve the accuracy of approximation is  not to incorporate more and more terms in the Born series expansion, but to develop a better estimate of the evolution semigroup $\exp(-i\tau\hat{\mathscr G}) $ beyond the short-term expansion. We shall proceed with this line of thought and develop a non-perturbative alternative in \S\ref{subsec:nonpert_approx}.
\eor\end{remark}\subsection{Generalized skin effect}In Theorem~\ref{thm:semigrp_prop}\ref{itm:skin_effect}, we claimed that \begin{align}\lim_{|\chi|\to+\infty}|\chi|\iiint_Vf(\bm r)((\hat I+i|\chi| \hat{\mathscr G})^{-1}\bm E_\mathrm{inc})(\bm r)\D^3\bm r=\mathbf0,\quad \forall f\in C^\infty_0(V;\mathbb C)\end{align}so long as the incident wave  $ \bm E_\mathrm{inc}\in \Phi(V;\mathbb C^3)$ is a transverse vector field that solves the Helmholtz equation $ (\nabla^2+k^2)\bm E_\mathrm{inc}=\mathbf 0$ in the  distributional sense. We shall give a proof of this statement using the strongly stable semigroup  $\{\exp(-i\tau\hat{\mathscr G})|\tau\geq0\}$  and explain the physical context that leads to the name ``generalized skin effect''.
\begin{proof}[Proof of Theorem~\ref{thm:semigrp_prop}\ref{itm:skin_effect}] By hitting the operator $ (\nabla^2+k^2)$ on both sides of the equation of motion \eqref{eq:EqMotion}, we can justify the following computations where derivatives are taken in the distributional sense:
\begin{align}i(\nabla^2+k^2)\frac{\partial\bm \psi_\tau(\bm r)}{\partial\tau}={}&(\nabla^2+k^2)(\hat  {\mathscr G}\bm \psi_\tau)(\bm r)=(\nabla^2+k^2)k^2\iiint_V\bm \psi_\tau(\bm r')\frac{e^{-ik|\bm r-\bm r'|}}{4\pi|\bm r-\bm r'|}\D^3\bm r'\notag\\={}&-k^{2}\bm \psi_\tau(\bm r),\text{ where } \bm \psi_\tau=\exp(-i\tau\hat{\mathscr G})\bm E_\mathrm{inc}.\label{eq:D2k2Eq}\end{align}In the computation above, we have made use of the fact that
\begin{align}&(\nabla^2+k^2)\nabla\left[\nabla\cdot\iiint_V\bm \psi_\tau(\bm r')\frac{e^{-ik|\bm r-\bm r'|}}{4\pi|\bm r-\bm r'|}\D^3\bm r'\right]\notag\\={}&-(\nabla^2+k^2)\nabla\varoiint_{\partial V}\bm n'\cdot\bm \psi_\tau(\bm r')\frac{e^{-ik|\bm r-\bm r'|}}{4\pi|\bm r-\bm r'|}\D S'=\mathbf0,\, \forall\bm r\in V\end{align}where the surface integral should be interpreted as the canonical pairing between $H^{-1/2} $ and $ H^{1/2}$, as in   \cite[\S3]{QualEM}.

Applying the definition of distributional derivatives to \eqref{eq:D2k2Eq}, we  arrive at
\begin{align}i\iiint_Vf(\bm r)(\nabla^2+k^2)\frac{\partial\bm \psi_\tau(\bm r)}{\partial\tau}\D^3\bm r={}&i\iiint_V\frac{\partial\bm \psi_\tau(\bm r)}{\partial\tau}(\nabla^2+k^2)f(\bm r)\D^3\bm r\notag\\={}&i\frac{\D}{\D\tau}\iiint_V\bm \psi_\tau(\bm r)(\nabla^2+k^2)f(\bm r)\D^3\bm r\notag\\={}&-k^{2}\iiint_Vf(\bm r)\bm \psi_\tau(\bm r)\D^3\bm r,\quad\forall f\in C^\infty_0(V;\mathbb C).\label{eq:D2k2Eq1}\end{align}We integrate \eqref{eq:D2k2Eq1} over $\tau\in[0,+\infty) $, and employ the strong stability condition
\begin{align}\lim_{\tau\to+\infty}\Vert\bm \psi_\tau\Vert_{L^2(V;\mathbb C^3)}=0\end{align} as  well as the Helmholtz equation $(\nabla^2+k^2)\bm E_\mathrm{inc}=\mathbf 0$ for the transverse incident wave, in order to deduce\begin{align}0={}&i\iiint_V\bm E_\mathrm{inc}(\bm r)(\nabla^2+k^2)f(\bm r)\D^3\bm r-i\lim_{T\to+\infty}\iiint_V\bm \psi_T(\bm r)(\nabla^2+k^2)f(\bm r)\D^3\bm r\notag\\={}&k^{2}\lim_{T\to+\infty}\int_0^{T}\left[\iiint_Vf(\bm r)\bm \psi_\tau(\bm r)\D^3\bm r\right]\D\tau,\quad\forall f\in C^\infty_0(V;\mathbb C).\label{eq:preGSE}\end{align}

Now, using the integral representation \eqref{eq:IRS} of the solution to the Born equation, we have
\begin{align}\chi\iiint_Vf(\bm r)((\hat I-\chi \hat{\mathscr G})^{-1}\bm E_\mathrm{inc})(\bm r)\D^3\bm r=\frac{1}{i}\iiint_Vf(\bm r)\left[\int_0^{+\infty}\bm \psi_\tau(\bm r)e^{i\tau/\chi}\D\tau\right]\D^3\bm r.\end{align}Along the negative  $\I\chi$-axis where $\chi=-i|\chi|$, the exponential decay of $e^{i\tau/\chi}=e^{-\tau/|\chi|} ,\tau>0$ and the uniform boundedness of $ \Vert \bm \psi_\tau\Vert_{L^2(V;\mathbb C^3)}\leq\Vert \bm E_\mathrm{inc}\Vert_{L^2(V;\mathbb C^3)},\tau>0$ allow us to interchange the integrations with respect to $\D \tau $ and $\D^3\bm r$ (owing to the Fubini theorem), and derive the following formula\begin{align}|\chi|\iiint_Vf(\bm r)((\hat I+i|\chi| \hat{\mathscr G})^{-1}\bm E_\mathrm{inc})(\bm r)\D^3\bm r=\int_0^{+\infty}\left[\iiint_Vf(\bm r)\bm \psi_\tau(\bm r)\D^3\bm r\right]e^{-\tau/|\chi|} \D\tau.\end{align}In the limit as $\alpha\colonequals 1/|\chi|\to0^+ $, we  prove \eqref{eq:GSE} by showing that
\begin{align}\int_0^{+\infty}\bm \psi_{\tau,f}e^{-\alpha\tau}\D\tau\colonequals{}& \int_0^{+\infty}\left[\iiint_Vf(\bm r)\bm \psi_\tau(\bm r)\D^3\bm r\right]e^{-\alpha\tau}\D\tau\notag\\\to{}&\int_0^{+\infty}\left[\iiint_Vf(\bm r)\bm \psi_\tau(\bm r)\D^3\bm r\right]\D\tau=0.\label{eq:chi_inf}\end{align}Here, the  convergence in \eqref{eq:chi_inf} follows from the Abel criterion for improper integrals: the convergence of the improper integral in \eqref{eq:preGSE} and the bounded monotone factor $e^{-\alpha\tau} $ together ensure the uniform convergence of
\begin{align}\bm \psi_{f}(\alpha)\colonequals \int_0^{+\infty}\bm \psi_{\tau,f}e^{-\alpha\tau}\D\tau=\lim_{T\to+\infty}\int_0^{T}\bm \psi_{\tau,f}e^{-\alpha\tau}\D\tau\end{align}with respect to $\alpha\in[0,+\infty) $, hence  $\bm \psi_{f}(\alpha) $ is continuous with respect to $\alpha\in[0,+\infty)$.
 \end{proof}
\begin{remark}In physical terms, \eqref{eq:GSE}  characterizes the skin effect in metallic conductors (treated as ``dielectrics'' with purely imaginary susceptibilities): as the electric conductivity $\sigma_\mathrm{cond}=|\chi|\omega $ increases to infinity, the  eddy current (the metallic counterpart of dielectric polarization current)  density attributed to the internal field $\bm J(\bm r)=i\omega\epsilon_0\chi \bm E(\bm r)=\omega\epsilon_0|\chi|((\hat I+|\chi| \hat{\mathscr G})^{-1}\bm E_\mathrm{inc})(\bm r) $ converges ``locally'' to zero inside the dielectric volume  $V$.
In other words, only the  current density near the dielectric boundary $\partial V$ is relevant to highly conducting material with  $ i\chi=\sigma_\mathrm{cond}/\omega\to+\infty$.

We call the  limit relation in  \eqref{eq:GSE}   a  ``generalized skin effect'' because it is rigorously established for any transverse incident wave, and for any smooth dielectric  with connected exterior  volume   $\mathbb R^3\smallsetminus(V\cup\partial V)$. \eor\end{remark}

\section{Perturbative and non-perturbative solutions to the Born equation\label{sec:nonpert_approach}}
\subsection{Perturbative theory for spherical scatterers: Born approximation and Rayleigh--Gans scattering}In Mie scattering \cite{Mie},  we have a  plane wave  $\bm E_\mathrm{inc}(\bm r)=\bm e_x\exp(-ikz) $ incident upon a dielectric sphere $V=O(\mathbf0,R)=\{(x,y,z)\in\mathbb R^3|x^2+y^2+z^2<R^2\} $.
We choose to elaborate on the perturbative solutions to light scattering on spherical particles, for  two aesthetic reasons:  (i)~The exact solution to Mie scattering is known in the form of infinite series~(see \cite{Mie} or \cite[\S9.22]{vandeHulst})---for  given values of  $ \xi=kR$ and $ n=\sqrt{1+\chi}>1$, the total scattering cross-section $ \sigma_{\mathrm{sc}}$ satisfies\footnote{Such an infinite series does not have a closed form. In practice,  the first $ kR+O((kR)^{1/3})$ terms \cite{AlvaroRanhaNeves2012} amount to satisfactory numerical accuracy. }\begin{footnotesize}\begin{align}&
\frac{\sigma_{\mathrm{sc}}}{\pi R^{2}}\notag\\={}&\frac{2}{\xi^{2}}\R\sum_{\ell=1}^\infty(2\ell+1)\left[ \frac{\det\begin{pmatrix}n^2j_\ell(n\xi) & \left.\frac{\D}{\D z}\right|_{z=n\xi}[zj_\ell(z)] \\
j_\ell(\xi) & \left.\frac{\D}{\D z}\right|_{z=\xi}[zj_\ell(z)] \\
\end{pmatrix}}{\det\begin{pmatrix}n^2j_\ell(n\xi) & \left.\frac{\D}{\D z}\right|_{z=n\xi}[zj_\ell(z)] \\
h^{(2)}_\ell(\xi) & \left.\frac{\D}{\D z}\right|_{z=\xi}[zh^{(2)}_\ell(z)] \\
\end{pmatrix}} +\frac{\det\begin{pmatrix}j_\ell(n\xi) & \left.\frac{\D}{\D z}\right|_{z=n\xi}[zj_\ell(z)] \\
j_\ell(\xi) & \left.\frac{\D}{\D z}\right|_{z=\xi}[zj_\ell(z)] \\
\end{pmatrix}}{\det\begin{pmatrix}j_\ell(n\xi) & \left.\frac{\D}{\D z}\right|_{z=n\xi}[zj_\ell(z)] \\
h^{(2)}_\ell(\xi) & \left.\frac{\D}{\D z}\right|_{z=\xi}[zh^{(2)}_\ell(z)] \\
\end{pmatrix}} \right],\label{eq:Mie_sc_series}
\end{align}\end{footnotesize}where $ j_\ell(z)\colonequals(-z)^\ell\left( \frac{\D}{z\D z} \right)^\ell\frac{\sin z}{z} $ and $ h^{(2)}_\ell(z)\colonequals(-z)^\ell\left( \frac{\D}{z\D z} \right)^\ell\frac{ie^{-i z}}{z} $ are spherical Bessel functions; (ii)~The Born approximation of the forward scattering amplitude $ \langle\bm E_\mathrm{inc},-\chi k(\hat I-\chi\hat{\mathscr G})^{-1}\bm E_\mathrm{inc} \rangle_V$, which is $  \langle\bm E_\mathrm{inc},-\chi k(\hat I+\chi\hat{\mathscr G})\bm E_\mathrm{inc} \rangle_V$,   can be evaluated in closed functional form. Therefore, we may  compare (later in \S\ref{subsec:approx_comp})  the exact benchmark  and perturbative solutions in a relatively neat fashion.

 For real-valued $\chi$, the  Born approximation to the total scattering cross-section  $\I \langle\bm E_\mathrm{inc},-\chi k(\hat I-\chi\hat{\mathscr G})^{-1}\bm E_\mathrm{inc} \rangle_{V=O(\mathbf0,R)}$ is given by the Rayleigh--Gans formula \cite{Rayleigh,Gans,vandeHulst}:
\begin{align}\label{eq:RGB}-\chi^2k\I\langle\bm E_\mathrm{inc},\hat{\mathscr G}\bm E_\mathrm{inc} \rangle_{V}={}&\frac{\pi R^2\chi^2}{4}\left\{\frac52+2k^2R^2-\frac{\sin(4kR)}{4kR}-\frac{7[1-\cos(4kR)]}{16k^{2}R^2}+\right.\notag\\{}&\left.+ \left( \frac{1}{2k^{2}R^2}-2 \right) [\gamma _{0}+\log (4k R)- \Ci(4 k R)]\right\},\end{align} where $ \Ci(x)\colonequals-\int_x^{+\infty}\frac{\cos t}{t}\D t$ is the cosine integral, and $\gamma_0\colonequals \lim_{M\to\infty}(\sum_{m=1}^M\frac1m-\log M)=0.577215+ $ is the Euler--Mascheroni constant.

As the  $ L^2$-norm of the Green operator $ \hat{\mathscr G}\colon{\Phi(V;\mathbb C^3)}\longrightarrow {\Phi(V;\mathbb C^3)}$ satisfies
\cite[Theorem 2.1 and Lemma 4.6]{QuantEM1Norm}\begin{align}&\Vert\hat {\mathscr G}\Vert_{\Phi(V;\mathbb C^3)}\leq G(kR)\notag\\\colonequals{}&\min \left\{\frac12+\frac{3}{5 \pi } \left(\frac{4 \pi }{3}\right)^{2/3} (kR)^{2},2+\frac{4}{5} k R+\frac{1}{2\pi} \left(  \sqrt[3]{k R+\frac{1}{2}}+\frac{4}{9 \sqrt[3]{k R_{}+\frac{1}{2}}} \right )\right\}\notag\\{}&+\min \left\{\frac{11}{45}(kR)^3,\frac{9}{16}kR\right\},\end{align}and the positive definite operator $\hat\gamma_S\colonequals -\I\hat{\mathscr G}=-(\hat{\mathscr G}-\hat{\mathscr G}^*)/(2i)\colon{L^{2}(V;\mathbb C^3)}\longrightarrow {L^{2}(V;\mathbb C^3)} $ satisfies
\begin{align}\Vert\hat {\gamma}_S\Vert_{L^2(V;\mathbb C^3)}\leq g(kR)\colonequals \min \left\{\frac{11}{45}(kR)^3,\frac{9}{16}kR\right\},\end{align}we may proceed with the error bound estimate  for the Born approximation of total scattering cross-section
(cf.~\cite[(4.79)]{QuantEM1Norm})\begin{align}\chi^{2}k|\I\langle\bm E_\mathrm{inc},\hat{\mathscr G}\bm E_\mathrm{inc} \rangle_{V}-\I\langle\bm E,\hat{\mathscr G}\bm E \rangle_{V}|\leq{}&\frac{4\pi R^2}{3}\frac{|\chi|^{3}kRg(kR)G(kR)}{[1-\chi G(kR)]^{2}}[|\chi|G(kR)+2].\end{align}for $ |\chi|G(kR)<1$ and $ \I \chi=0$.

It has been well recognized that the Rayleigh--Gans formula gives a good approximation to the light scattering problem only when the susceptibility is ``very small'' $|\chi|\ll1 $ \cite{vandeHulst}. Here, we may quantify the smallness by investigating the limit behavior of the Rayleigh--Gans formula and its error bound estimate. Expanding in a neighborhood of $R=0$, we recover the Rayleigh quartic law
\begin{align}-\chi^2k\I\langle\bm E_\mathrm{inc},\hat{\mathscr G}\bm E_\mathrm{inc} \rangle_{V}\sim\frac{8\pi k^{4}R^6\chi^2}{27}\end{align} which shows a dependence of total scattering cross-section proportional to the inverse fourth power of the wavelength. Meanwhile, expanding the error bound estimate in the limit of $R\to0^+$, we obtain a conservative guess of the deviation caused by the Born approximation as
\begin{align}\sim\frac{44|\chi|^{3}\pi\l(|\chi|+4)k^{4}R^6}{135(2-|\chi|)^{2}}. \end{align}Thus, we see that the relative error is bounded by a factor of
$\frac{11}{10}\frac{|\chi|(|\chi|+4)}{(2-|\chi|)^{2}}$
in such a limit scenario. This gives a criterion for the susceptibility range where the Born approximation is accurate to our desired level. On the other hand, it is also known that the Rayleigh--Gans formula works well  when the relative phase shift $2(n-1)kR \sim \chi kR$ is ``very small'' $|\chi |kR\ll1 $ \cite[\S\S6--7]{vandeHulst}. To see what this means, we may pick a conservative estimate based on a sufficient condition for the convergence of the Born series (and hence the finiteness of the error bound estimate), which is $  |\chi|G(kR)<1$. For small values of $|\chi|$, we only need to consider relatively large values of $kR$, so that the identity $G(kR)= 2+\frac{109}{80} k R+\frac{1}{2\pi} \left(  \sqrt[3]{k R+\frac{1}{2}}+\frac{4}{9 \sqrt[3]{k R_{}+\frac{1}{2}}} \right)$ may be used. This leaves us another quantitative criterion for the efficacy of Born approximation:
\begin{align}|\chi|\left[ 2+\frac{109}{80} k R+\frac{1}{2\pi} \left(  \sqrt[3]{k R+\frac{1}{2}}+\frac{4}{9 \sqrt[3]{k R_{}+\frac{1}{2}}} \right ) \right]<1.\end{align}

From these discussions, we can clearly see that the Born approximation works fine only when the dielectric medium perturbs the incident field minimally. Some practical applications of light scattering may have sphere sizes or susceptibility values that lie outside this perturbative regime, which calls for a non-trivial effort (see \S\S\ref{subsec:nonpert_approx}--\ref{subsec:nonpert_Mie} below) to complement the approximations based on Born series.

\subsection{Heuristics for the non-perturbative approximation\label{subsec:nonpert_approx}}This subsection is more  experimental than the rest of the current article.
Instead of deriving a non-perturbative formula with full rigor, our modest goal is to present an approximation scheme that is compatible with semigroup asymptotics, consistent with physical picture, and amenable to  computation.

As mentioned before, an improvement to lower-order Born approximation would be a good estimate of the evolution  semigroup $\exp(-i\tau\hat{\mathscr G})$ for large values of $ \tau$, instead of incorporation of more  terms in a Taylor series expansion.

In this section, we will exploit the long-term ($ \tau\to+\infty$) behavior of the evolution  semigroup $\exp(-i\tau\hat{\mathscr G})$ to derive a non-perturbative approximation with higher accuracy than the  Born series.

We will first perform asymptotic analysis of the ``bulk contribution'' in
the long-term limit $ \tau\to+\infty$. To begin, we note the following integral identity
\begin{align}1-\exp(-i\tau G)=\int_0^{+\infty}\exp\left(-\frac{s^2}{4i\tau
G}\right)J_1(s)\D s,\quad \I G<0,\label{eq:J1int}\end{align}which holds for complex numbes $G\in \mathbb C$ in the lower half-plane, and the first-order Bessel function  $ J_1$. From \cite[Theorem 1.1]{QualEM}, we see that all the eigenvalues of the Green operator $\hat{\mathscr G} $  lie in the lower half-plane as well, \textit{i.e.}~the ``physical point spectrum'' satisfies $ \sigma_p^\Phi(\hat{\mathscr G})\subset \{\lambda\in \mathbb C|\I\lambda<0\}$. Moreover, the operator $\hat {\mathscr G }:\Phi(V;\mathbb C^3)\longrightarrow\Phi(V;\mathbb C^3)$ has a densely-defined\footnote{Here, the domain of definition $ \hat{\mathscr G}\Phi(V;\mathbb C^3)$ is a dense subset of $ \Phi(V;\mathbb C^3)$, because the continuous spectrum $ \sigma_c^\Phi(\hat{\mathscr G})$ contains the origin $0\in\mathbb C$ \cite[Proposition 3.4]{QualEM}.} unbounded inverse \begin{align}-\left(\hat I+\frac{\nabla^2}{k^2}\right):\hat{\mathscr G}\Phi(V;\mathbb C^3)\longrightarrow\Phi(V;\mathbb C^3) \text{ such that }-\left(\hat I+\frac{\nabla^2}{k^2}\right)\hat{\mathscr G}\bm E=\bm E,\forall\bm E\in\Phi(V;\mathbb C^3).\end{align} Thus, heuristically speaking, we may generalize the integral identity \eqref{eq:J1int} into the operator form
\begin{align}\text{``}\hat I-\exp(-i\tau \hat{\mathscr G})=\int_0^{+\infty}\exp\left(-\frac{is^2}{4\tau }\right)\exp\left(-\frac{is^2\nabla^{2}}{4\tau k^{2}}\right)J_1(s)\D s\text{''}.\label{eq:opJ1int}\end{align}
Here, $ \exp(-\frac{is^2\nabla^2}{4\tau k^2})$  hearkens back to the Schr\"odinger semigroup in quantum mechanics that governs the evolution of wave functions.

In the formula above, we have added quotation marks  as a  caveat
for two possible weaknesses of the heuristic generalization:

\begin{enumerate}[label=(\alph*), widest=a]
\item
In passage from \eqref{eq:J1int} to \eqref{eq:opJ1int}, we have literally
pretended that $ \hat{\mathscr G}$ is a finite-dimensional square matrix:
its left inverse equals its right inverse and the evolution semigroups are
identified with matrix exponentials;  \item Upon  writing $\exp(-\frac{is^2\nabla^2}{4\tau k^2}) $ in  \eqref{eq:opJ1int}, we have not explicitly specified the boundary
condition for the Laplace operator $\nabla^2 $, thus adding ambiguity to
the notation.

\end{enumerate}  Here, point (a) might not amount to a serious physical flaw
in practice. This is because the Born equation is, after all, a continuum idealization
of the interaction between light and a physical medium built upon finite-sized
atoms and
molecules. Essentially,  macroscopic light scattering  may involve very large but still
finite degrees of freedom, so it is physically acceptable to ``discretize''  $ \hat{\mathscr G}$ as a finite-dimensional square matrix. Point (b) might
 not pose as a significant numerical obstacle either, if our main interest is deducing $(\exp(-i\tau \hat{\mathscr G})\bm E_\mathrm{inc})(\bm r) $ in the long-term limit $ \tau\to+\infty$ and investigating the ``bulk region'' points  $\bm r$  that are not too close to the boundary $ \partial V$. When
 $\tau $ is large and the ``propagation time'' $\propto s^2/\tau$ is small, the points in the ``bulk region''
 may not have enough chance to  feel the effect of  the boundary, hence the
 boundary conditions for the Laplace operator $ \nabla^2$ do not matter. We
 momentarily ignore the physical boundary $ \partial V$ and write down the asymptotic expression
\begin{align}\label{eq:freeprop}&
\left(\exp\left(\frac{s^2{\nabla^2}}{4i\tau k^{2}}\right)\bm E_\mathrm{inc}\right)(\bm
r)\sim\left(\exp\left(\frac{s^2\overset{_\circ\text{\hspace{0.5em}}}{\nabla^2}}{4i\tau k^{2}}\right)\bm E_\mathrm{inc}\right)(\bm
r)\notag\\={}&\iiint_V \overset{_\circ}{\kappa}_{s^2/(4i\tau k^2)}(\bm r,\bm r')\bm E_\mathrm{inc}(\bm
r')\D^3\bm r',\quad \text{where }\overset{_\circ}{\kappa}_{\beta}(\bm r,\bm r')=\frac{e^{-|\bm
r-\bm r'|^2/(4\beta)}}{(4\pi\beta)^{3/2}}
\end{align}
 for ``bulk contributions'' with $ \tau\to+\infty$. Here, $ \overset{_\circ\text{\hspace{0.5em}}}{\nabla^2}$ refers to the usual Laplace operator defined in the free space $\mathbb R^3 $, and $\overset{_\circ}{\kappa}_{\beta}(\bm r,\bm r') $ is its associated free-space ``heat kernel'' or  ``Schr\"odinger propagator''.
Here, we have regarded $ s^2/\tau$ as a small quantity in the expression of $\exp(-\frac{is^2\nabla^2}{4\tau k^2}) $ for $\tau\to+\infty$, irrespective of the value of $s$. This can be heuristically justified by the identity:
\begin{align}\bm E_\mathrm{inc}=\lim_{\tau\to+\infty}[\hat I-\exp(-i\tau \hat{\mathscr G})]\bm E_\mathrm{inc}=\lim_{\tau\to+\infty}\int _0^{+\infty}\exp\left(-\frac{is^2}{4(\tau+i0^{+}) }\right)\bm E_\mathrm{inc}J_1(s)\D s,\end{align}
which means that treating $  s^2/\tau$  as ``uniformly small'' in the limit of $\tau\to+\infty$ is quantitatively consistent with the strong stability of the evolution semigroup $ \lim_{\tau\to+\infty}\Vert\exp(-i\tau \hat{\mathscr G})\bm E_\mathrm{inc}\Vert_{L^2(V;\mathbb C^3)}$ $=0$ in light scattering.

Now, if we take the liberty of interchanging the integrations over $\D^3\bm r' $, $\D s$ and $\D \tau$, we will be able to formally derive an asymptotic formula for $\I\chi<0$ and $|\chi|\to+\infty$:\begin{footnotesize}{
\begin{align}&i\chi[\bm E_\mathrm{inc}(\bm
r)-((\hat I-\chi\hat{\mathscr G})^{-1}\bm E_\mathrm{inc})(\bm
r)]+\text{``boundary corrections''}\notag\\\sim{}&\int_0^{+\infty}\left\{\int_0^{+\infty}\exp\left(-\frac{is^2}{4\tau }\right)\left[\iiint_V \overset{_\circ}{\kappa}_{s^2/(4i\tau k^2)}(\bm r,\bm r')\bm E_\mathrm{inc}(\bm
r')\D^3\bm r'\right]J_1(s)\D s\right\}e^{i\tau/\chi}\D\tau\notag\\={}&\iiint_V\left\{\int_0^{+\infty}\left[ \int_0^{+\infty}\left( \frac{i\tau k^{2}}{\pi s^{2}} \right)^{3/2}\exp\left( -\frac{is^2}{4\tau }-\frac{i\tau k^{2}|\bm
r-\bm r'|^2}{s^{2}} +\frac{i\tau}{\chi}\right)\D\tau\right]J_1(s)\D s\right\}\bm E_\mathrm{inc}(\bm
r')\D^3\bm r'\notag\\={}&\iiint_V\left\{\int_0^{+\infty}\left[ \int_0^{+\infty}\left( \frac{i\tau' k^{2}}{\pi } \right)^{3/2}\exp\left( -\frac{i}{4\tau' }-i\tau' k^{2}|\bm
r-\bm r'|^2 +\frac{i\tau's^{2}}{\chi}\right)s^{2}\D\tau'\right]J_1(s)\D s\right\}\bm E_\mathrm{inc}(\bm
r')\D^3\bm r'\notag\\={}&\iiint_V\left\{\int_0^{+\infty}\left[ \int_0^{+\infty}e^{i\tau's^{2}/\chi}J_1(s)s^{2}\D s\right]\left( \frac{i\tau' k^{2}}{\pi } \right)^{3/2}\exp\left( -\frac{i}{4\tau' }-i\tau' k^{2}|\bm
r-\bm r'|^2 \right)\D\tau'\right\}\bm E_\mathrm{inc}(\bm
r')\D^3\bm r'\notag\\={}&-\frac{\chi^2}{4}\iiint_V\left[\int_0^{+\infty}\left( \frac{ik^{2}}{\pi  } \right)^{3/2}\exp\left( -\frac{i(1+\chi)}{4\tau' }-i\tau' k^{2}|\bm
r-\bm r'|^2 \right)\frac{\D\tau'}{\sqrt{\tau'}}\right]\bm E_\mathrm{inc}(\bm
r')\D^3\bm r'\notag\\={}&-i\chi^2 k^2\iiint_V\frac{e^{-i\sqrt{1+\chi}k|\bm
r-\bm r'|}}{4\pi|\bm
r-\bm r'|}\bm E_\mathrm{inc}(\bm
r')\D^3\bm r'.\label{eq:BulkAsym}\end{align}}\end{footnotesize}In the derivation of  \eqref{eq:BulkAsym}, we have made use of the substitution $\tau'=\tau/s^2 $ and exactly evaluated individual integrals with respect to $ \D s$ and $\D\tau'$. In the last step of \eqref{eq:BulkAsym}, the branch cut of the refractive index $n=\sqrt{1+\chi} $ is chosen such that $\R(i\sqrt{1+\chi})>0 $ for $ \I\chi<0$.

Now we point out that the ``boundary corrections'' we previously ignored can be asymptotically recovered by the following formula \begin{align}((\hat I-\chi\hat{\mathscr G})^{-1}\bm E_\mathrm{inc})(\bm
r)\sim{}&\bm E_\mathrm{inc}(\bm
r)+\chi k^2\iiint_V\frac{e^{-i\sqrt{1+\chi}k|\bm
r-\bm r'|}}{4\pi|\bm
r-\bm r'|}\bm E_\mathrm{inc}(\bm
r')\D^3\bm r'\notag\\{}&+\frac{\chi}{1+\chi}\nabla\left[ \nabla\cdot\iiint_V\frac{e^{-i\sqrt{1+\chi}k|\bm
r-\bm r'|}}{4\pi|\bm
r-\bm r'|}\bm E_\mathrm{inc}(\bm
r')\D^3\bm r' \right]\notag\\={}&\bm E_\mathrm{inc}(\bm
r)+\chi k^2\iiint_V\frac{e^{-i\sqrt{1+\chi}k|\bm
r-\bm r'|}}{4\pi|\bm
r-\bm r'|}\bm E_\mathrm{inc}(\bm
r')\D^3\bm r'\notag\\{}&-\frac{\chi}{1+\chi}\nabla\varoiint_{\partial V}\frac{\bm n'\cdot\bm E_\mathrm{inc}(\bm
r')e^{-i\sqrt{1+\chi}k|\bm
r-\bm r'|}}{4\pi|\bm
r-\bm r'|}\D^3\bm r'.\label{eq:bulkapprox}\end{align}
Here in \eqref{eq:bulkapprox}, we have introduced  a surface integral term with  prefactor  $\frac{\chi}{1+\chi} $ to correct the boundary effects. This correction term serves two purposes: (i)~It ensures that the asymptotic formula is consistent with both the transversality constraint $\nabla\cdot((\hat I-\chi\hat{\mathscr G})^{-1}\bm E_\mathrm{inc})(\bm r)=0,\bm r\in V $ and the Helmholtz equation $[\nabla^2+(1+\chi)k^2]((\hat I-\chi\hat{\mathscr G})^{-1}\bm E_\mathrm{inc})(\bm r)=\mathbf 0,\bm r\in V $; (ii)~It recovers the correct leading order Born series $(\hat I-\chi\hat{\mathscr G})^{-1}\bm E_\mathrm{inc}\sim \bm E_\mathrm{inc}+\chi\hat{\mathscr G}\bm E_\mathrm{inc} +\cdots$ in the limit of $|\chi|\to0$, up to order $O(\chi) $. Thus, the integral formula   \eqref{eq:bulkapprox} respects both the short-term ($\tau\to0^+$) and the long-term ($\tau\to+\infty$) asymptotic behavior of the evolution semigroup $ \exp(-i\tau\hat{\mathscr G})$.

Admittedly, neat as it may seem, the asymptotic formula in \eqref{eq:bulkapprox}  still has its limitations, in at least three respects.

First, we note that    \eqref{eq:bulkapprox} only takes care of both extremes of the asymptotic behavior of the evolution semigroup, but ignores whatever happens ``in between''. It is  definitely possible that a dominant contribution to the integral representation \eqref{eq:IRS} comes from ``intermediate regimes'' in the $\tau$-domain, rather than the short-term ($\tau\to0^+$) and long-term ($\tau\to+\infty$) extremes.

Second,
we can semi-quantitatively estimate the domain of validity for  \eqref{eq:bulkapprox}, by a singularity argument.
The right-hand side of \eqref{eq:bulkapprox} has a pole at $ \chi=-1$ and  a branch cut for $\chi<-1 $. Neither the pole nor the branch cut should be present in the solution to electromagnetic scattering by transverse incident fields, according to the optical resonance theorem \cite[Theorem 3.1]{QualEM}. Therefore, loosely speaking, the ``radius of convergence'' for   \eqref{eq:bulkapprox} does not exceed $1$, and one may not deduce reliable information therefrom if the refractive index $ n=\sqrt{1+\chi}$ is greater than $\sqrt2 $. Nonetheless, numerical experiments (in \S\ref{subsec:approx_comp}) suggest that the approximate formula  \eqref{eq:bulkapprox} might still work properly when the refractive index is a little larger than $ \sqrt2$.

Third, we need to physically quantify the applicability domain   for \eqref{eq:bulkapprox}, or equivalently, understand the conditions that justify our strategy of ``fixing two extreme ends''. We  assume, without loss of generality, that  the incident wave $\bm E_\mathrm{inc}\in C^\infty(V\cup\partial V;\mathbb C^3) $ is smooth up to the dielectric boundary, and convert the right-hand side of \eqref{eq:bulkapprox} into a surface integral form \begin{footnotesize}{\begin{align}&\bm E_\mathrm{inc}(\bm
r)+\chi k^2\iiint_V\frac{e^{-i\sqrt{1+\chi}k|\bm
r-\bm r'|}}{4\pi|\bm
r-\bm r'|}\bm E_\mathrm{inc}(\bm
r')\D^3\bm r'+\frac{\chi}{1+\chi}\nabla\left[ \nabla\cdot\iiint_V\frac{e^{-i\sqrt{1+\chi}k|\bm
r-\bm r'|}}{4\pi|\bm
r-\bm r'|}\bm E_\mathrm{inc}(\bm
r')\D^3\bm r' \right]\notag\\={}&\iiint_V\frac{e^{-i\sqrt{1+\chi}k|\bm
r-\bm r'|}}{4\pi|\bm
r-\bm r'|}(\nabla'^{2}+k^{2})\bm E_\mathrm{inc}(\bm
r')\D^3\bm r'-\iiint_V\bm E_\mathrm{inc}(\bm
r')(\nabla'^{2}+k^{2})\frac{e^{-i\sqrt{1+\chi}k|\bm
r-\bm r'|}}{4\pi|\bm
r-\bm r'|}\D^3\bm r'\notag\\&-\frac{\chi}{1+\chi}\nabla\left[ \varoiint_{\partial V}\frac{e^{-i\sqrt{1+\chi}k|\bm
r-\bm r'|}}{4\pi|\bm
r-\bm r'|}\bm n'\cdot\bm E_\mathrm{inc}(\bm
r')\D S' \right]\notag\\={}&\varoiint_{\partial V}\frac{e^{-i\sqrt{1+\chi}k|\bm
r-\bm r'|}}{4\pi|\bm
r-\bm r'|}(\bm n'\cdot\nabla')\bm E_\mathrm{inc}(\bm
r')\D S'-\varoiint_{\partial V}\bm E_\mathrm{inc}(\bm
r')(\bm n'\cdot\nabla')\frac{e^{-i\sqrt{1+\chi}k|\bm
r-\bm r'|}}{4\pi|\bm
r-\bm r'|}\D S'\notag\\&-\frac{\chi}{1+\chi}\nabla\left[ \varoiint_{\partial V}\frac{e^{-i\sqrt{1+\chi}k|\bm
r-\bm r'|}}{4\pi|\bm
r-\bm r'|}\bm n'\cdot\bm E_\mathrm{inc}(\bm
r')\D S' \right]\notag\\={}&-\varoiint_{\partial V}\bm n'\times[\nabla'\times\bm E_\mathrm{inc}(\bm r')]\frac{e^{-i\sqrt{1+\chi}k|\bm
r-\bm r'|}}{4\pi|\bm
r-\bm r'|}\D S'-\varoiint_{\partial V}[\bm n'\times\bm E_\mathrm{inc}(\bm r')]\times\nabla'\frac{e^{-i\sqrt{1+\chi}k|\bm
r-\bm r'|}}{4\pi|\bm
r-\bm r'|}\D S'\notag\\&-\frac{1}{1+\chi}\varoiint_{\partial V}[\bm n'\cdot\bm E_\mathrm{inc}(\bm
r')]\nabla'\frac{e^{-i\sqrt{1+\chi}k|\bm
r-\bm r'|}}{4\pi|\bm
r-\bm r'|}\D S'.\label{eq:surfintbulk}\end{align}}\end{footnotesize}Recalling the familiar Kirchhoff integral for vector wave diffraction   \cite[p.~483]{Jackson:EM}, we see from \eqref{eq:surfintbulk} that the asymptotic formula in \eqref{eq:bulkapprox} describes a transverse electric field $\bm E $ with boundary conditions
\begin{align}&\bm n\times[\nabla\times\bm E(\bm r)]=\bm n\times[\nabla\times\bm E_\mathrm{inc}(\bm r)],\quad \bm n\times\bm E(\bm r)=\bm n\times\bm E_\mathrm{inc}(\bm r),\notag\\{}&\lim_{\varepsilon\to0^+}\bm n\cdot\bm E(\bm r-\varepsilon\bm n)=\frac{\bm n\cdot\bm E_\mathrm{inc}(\bm
r)}{1+\chi},\,\forall\bm r\in\partial V.\label{eq:quasistat_resp}\end{align}As the continuity of the normal component of the electric displacement leads us to $ \lim_{\varepsilon\to0^+}\bm n\cdot\bm E(\bm r+\varepsilon\bm n)=(1+\chi)\lim_{\varepsilon\to0^+}\bm n\cdot\bm E(\bm r-\varepsilon\bm n)={\bm n\cdot\bm E_\mathrm{inc}(\bm
r)},\forall\bm r\in\partial V$, this means that  \eqref{eq:bulkapprox} approximates the dielectric response by a vector wave diffraction problem, with boundary conditions specified by the local quasistatic response of the molecular dipoles residing on the dielectric interface.

According to  this diffraction approximation  \eqref{eq:bulkapprox}, the total electric field immediately outside the dielectric boundary appears  identical to the unperturbed incident wave $\bm E_\mathrm{inc} $. This approximation is fair only if the scattered electric field is indeed negligible as compared to the incident field, a condition that is not necessarily met for very large values of $|\chi| $. While we are not making further analytic attempts in this article to amend the boundary corrections to the Schr\"odinger semigroup that accommodates  very large values of $|\chi| $, we will numerically illustrate (in \S\ref{subsec:approx_comp}) that  \eqref{eq:bulkapprox} indeed works well for moderately large values of $ |\chi|$ that are encountered in practical problems of direct and inverse light scattering.

In short, we improve the Born approximation [short-term asymptotic expansion of the semigroup $\exp(-i\tau \hat{\mathscr G}),\tau\geq0 $] by taking care of the long-term asymptotic behavior of  $\exp(-i\tau \hat{\mathscr G}),\tau\to+\infty$. As will be shown in the next subsection, the net effect of this improvement is to lift the constraints $|\chi|\ll1 $ and $|\chi|kR\ll1$ in the application of Born approximation, and provide sufficiently accurate approximations when $|\chi| $ and $|\chi|kR$ are of order unity.
 \subsection{Non-perturbative approximation to Mie scattering\label{subsec:nonpert_Mie}}In the following two propositions, we will evaluate two integrals arising from \eqref{eq:bulkapprox} of \S\ref{subsec:nonpert_approx}, in order to derive an asymptotic formula for the forward scattering amplitude $ \langle\bm E_\mathrm{inc},-\chi k(\hat I-\chi\hat{\mathscr G})^{-1}\bm E_\mathrm{inc} \rangle_V$ of  Mie scattering.
\begin{proposition}\label{prop:F1}Let  the vector field $\bm F_1(\bm r),\bm r\in O(\mathbf0,R)$ be defined as
\begin{align}\bm F_1(\bm r)\colonequals \bm E_\mathrm{inc}(\bm
r)+\chi k^2\iiint_{O(\mathbf0,R)}\frac{e^{-i\sqrt{1+\chi}k|\bm
r-\bm r'|}}{4\pi|\bm
r-\bm r'|}{\bm E}_\mathrm{inc}(\bm
r')\D^3\bm r' \end{align}for $ \bm E_\mathrm{inc}(\bm r)=\bm e_x\exp(-ikz)$, then we have the exact identity
\begin{align}\langle\bm E_\mathrm{inc} ,-\chi k\bm F_1\rangle_{V=O(\mathbf0,R)}  ={}&2\pi inR^{2}+\frac{\pi i(n+1)^2}{(n-1)^2}\frac{1-e^{-2i(n-1)kR}[1+2i(n-1)kR]}{4k^{2}}\notag\\&+\frac{\pi i(n-1)^{2}}{(n+1)^2}\frac{e^{-2i(n+1)kR}[1+2i(n+1)kR]-1}{4k^{2}},\label{eq:F1div}\end{align}where $\I\chi<0,\I n=\I\sqrt{1+\chi}<0$.
  \end{proposition}

\begin{proof}We will evaluate the inner product $\langle\bm E_\mathrm{inc} ,\chi k\bm F_1\rangle_{V=O(\mathbf0,R)}  $ with the help of Fourier transforms. Direct computation shows that \begin{align}\widetilde{\bm E}_\mathrm{inc}(\bm q) =\bm e_x \iiint_{|\bm r|<R}e^{i(\bm q-k\bm e_z)\cdot\bm r}\D^3\bm r=4\pi R^{3}\bm e_x\frac{ j_1(|\bm q-k\bm e_z|R)}{|\bm q-k\bm e_z|R}\label{eq:E_inc_q}\end{align}where $ j_1(u)=(\sin u-u\cos u)/u^3$ is the first order spherical Bessel function. Akin to \cite[(2.18)]{QualEM}, we may write down
\begin{align}\langle\bm E_\mathrm{inc} ,\chi k\bm F_1\rangle_{V=O(\mathbf0,R)}=\frac{4\pi\chi k R^3}{3}+\frac{\chi^{2}k^3}{(2\pi)^{3}}\iiint_{\mathbb R^3}\frac{|\widetilde{\bm E}_\mathrm{inc}(\bm q)|^2}{|\bm q|^2-(1+\chi)k^{2}}\D^3\bm q.\end{align} Using the coordinate transformation $\bm q'=\bm q-k\bm e_z $, the triple integral  above can be converted into\begin{align}&\iiint_{\mathbb R^3}\frac{|\widetilde{\bm E}_\mathrm{inc}(\bm q)|^2}{|\bm q|^2-(1+\chi)k^{2}}\D^3\bm q=\iiint_{\mathbb R^3}\frac{|\widetilde{\bm E}_\mathrm{inc}(\bm q'+k\bm e_z )|^2}{|\bm q'+k\bm e_z |^2-(1+\chi)k^{2}}\D^3\bm q'\notag\\={}&\iiint_{\mathbb R^3}\frac{|\widetilde{\bm E}_\mathrm{inc}(\bm q'+k\bm e_z )|^2}{|\bm q'|^{2}+2k\bm e_z\cdot\bm q'-\chi k^{2}}\D^3\bm q'\notag\\={}&2\pi(4\pi R^{3})^2\int_{0}^{+\infty}\left[\frac{ j_1(qR)}{R}\right]^2\dfrac{1}{2kq}\log\dfrac{q^{2}+2kq-\chi k^{2}}{q^{2}-2kq-\chi k^{2}}\D q\notag\\={}&-32\pi^{3}\int_{-\infty}^{+\infty}\frac{ \left[2 q R \sin (2 q R)+\cos (2 q R)-2 q^2 R^2-1\right](k+q)}{8k q^4 (q^2+2 k q-\chi k^{2})}\D q.\label{eq:FTqxFq}\end{align}Here, in the last step of \eqref{eq:FTqxFq}, we have performed integration by parts and exploited the symmetry between $q $ and $-q$. Denoting the last term in \eqref{eq:FTqxFq} by $(2\pi)^{3}A_{1}(\chi,k,R) $, we have the differential equation
\begin{align}\frac{\partial}{\partial R}\left[\frac1R\frac{\partial A_{1}(\chi,k,R)}{\partial R}\right]={}&\int_{-\infty}^{+\infty}\frac{ 4\sin (2 q R)}{ kq}\frac{ k+ q}{q^2+2 k q-\chi k^{2} }\D q\notag\\={}&\int_{-\infty}^{+\infty}\frac{ 2\sin (2 q R)}{ kq}\left[\frac{1}{q+(1-n)k}+\frac{1}{q+(1+n)k}\right]\D q\notag\\={}&\frac{2\pi}{k^{2}}\left[\frac{1-e^{-2i(n-1)kR}}{1-n} +\frac{1-e^{-2i(n+1)kR}}{1+n} \right].\label{eq:diffEqR}\end{align}Here, we have applied the identity
\begin{align}\label{eq:SinIntQ}\int_{-\infty }^{+\infty } \frac{\sin (2 q R)}{ q (q+Q)} \D q=\frac{ \pi}{Q}(  1-e^{2 i Q R\I Q/|\I Q|}),\quad \I Q\neq0\end{align}
 to the last step of \eqref{eq:diffEqR}. After evaluating the integral expressed in  \eqref{eq:FTqxFq} by integrating the relation in \eqref{eq:diffEqR}, we may arrive at the conclusion
\begin{align}&\langle\bm E_\mathrm{inc} ,\chi k\bm F_1\rangle_{V=O(\mathbf0,R)}\notag\\={}&\frac{4\pi\chi k R^3}{3}-\frac{\pi i\chi^{2}}{4k^{2}}\frac{1-e^{-2i(n-1)kR}[1+2i(n-1)kR]}{(n-1)^4}\notag\\{}&-\frac{\pi i\chi^{2}}{4k^{2}}\frac{e^{-2i(n+1)kR}[1+2i(n+1)kR]-1}{(n+1)^4}-\frac{\pi\chi^2k}{3}\left (\frac{6 i nR^2}{k \chi^2}+\frac{4R^3}{\chi} \right)\notag\\={}&- \frac{\pi i(n+1)^2}{(n-1)^2}\frac{1-e^{-2i(n-1)kR}[1+2i(n-1)kR]}{4k^{2}}\notag\\{}&-\frac{\pi i(n-1)^{2}}{(n+1)^2}\frac{e^{-2i(n+1)kR}[1+2i(n+1)kR]-1}{4k^{2}}-2\pi inR^{2} \end{align}as claimed.
 \end{proof}

\begin{proposition}Let  the vector field $\bm F_2(\bm r),\bm r\in O(\mathbf0,R)$ be defined as
\begin{align}\bm F_2(\bm r)\colonequals \frac{\chi}{1+\chi}\nabla\left[ \nabla\cdot\iiint_{O(\mathbf0,R)}\frac{e^{-i\sqrt{1+\chi}k|\bm
r-\bm r'|}}{4\pi|\bm
r-\bm r'|}{\bm E}_\mathrm{inc}(\bm
r')\D^3\bm r'\right] \end{align}for $ \bm E_\mathrm{inc}(\bm r)=\bm e_x\exp(-ikz)$, then we have the exact identity
\begin{align}&\langle\bm E_\mathrm{inc} ,-\chi k\bm F_2\rangle_{V=O(\mathbf0,R)} \notag\\ =&\frac{\pi}{16 k^2 n^2}   \left\{-2 i\chi^{2} [2 (\chi+2)  k^2R^2-1] \left[\Ei(-2 i  (n-1)k R)-\Ei(-2 i  (n+1)k R)+\log\frac{n+1}{n-1}\right]\right.\notag\\&+4 i n (2 \chi ^{2}k^2  R^2-\chi-2)+ e^{-2 i  (n+1)k R}(n-1)^2 \left[2  (n+1) (\chi+2)k R+i (n^{2}+4n+1)\right]\notag\\&\left.- e^{-2 i  (n-1)k R}(n+1)^2  \left[2 (n-1)(\chi+2)k R+i (n^{2}-4n +1)\right]\phantom{\frac12}\hspace{-1em}\right\},\label{eq:F2div}\end{align} where $ \I\chi<0,\I n=\I\sqrt{1+\chi}<0$.  \end{proposition}
\begin{proof}By Fourier transform, we convert the target object into \begin{align}\langle\bm E_\mathrm{inc} ,-\chi k\bm F_2\rangle_{V=O(\mathbf0,R)}=\frac{\chi^2}{1+\chi}\frac{k}{(2\pi)^{3}}\iiint_{\mathbb R^3}\frac{|\bm q\cdot\widetilde{\bm E}_\mathrm{inc}(\bm q)|^2}{|\bm q|^2-(1+\chi)k^{2}}\D^3\bm q,\end{align}where $ \widetilde{\bm E}_\mathrm{inc}(\bm q)$ has been given by \eqref{eq:E_inc_q}. As before, we introduce the coordinate transformation  $\bm q'=\bm q-k\bm e_z $ to evaluate the triple integral in question:
\begin{align}&\iiint_{\mathbb R^3}\frac{|\bm q\cdot\widetilde{\bm E}_\mathrm{inc}(\bm q)|^2}{|\bm q|^2-(1+\chi)k^{2}}\D^3\bm q\notag\\={}&\iiint_{\mathbb R^3}\frac{(\bm e_x\cdot\bm q')^{2}|\widetilde{\bm E}_\mathrm{inc}(\bm q'+k\bm e_z )|^2}{|\bm q'+k\bm e_z |^2-(1+\chi)k^{2}}\D^3\bm q'=\iiint_{\mathbb R^3}\frac{(\bm e_x\cdot\bm q')^{2}|\widetilde{\bm E}_\mathrm{inc}(\bm q'+k\bm e_z )|^2}{|\bm q'|^{2}+2k\bm e_z\cdot\bm q'-\chi k^{2}}\D^3\bm q'\notag\\={}&2\pi(4\pi R^{3})^2\int_{0}^{+\infty}\left[\frac{ j_1(qR)}{R}\right]^2\times\notag\\&\times\left[\dfrac{q^{2}-\chi k^2}{4k^{2}}-\dfrac{(q^{2}+2kq-\chi k^{2})(q^{2}-2kq-\chi k^{2})}{16k^{3}q}\log\dfrac{q^{2}+2kq-\chi k^{2}}{q^{2}-2kq-\chi k^{2}}\right]\D q.\label{eq:FTqFq}\end{align}Denoting the last term in \eqref{eq:FTqFq} by $(2\pi)^{3}A_{2}(\chi,k,R) $, we may deduce the differential equation\begin{align}\frac{\partial A_{2}(\chi,k,R)}{\partial \chi}={}&\frac{1}{2n}\frac{\partial A_{2}(n^{2}-1,k,R)}{\partial n}\notag\\={}&4R^{4}\int_{0}^{+\infty}j_1^{2}(qR)\left(-\dfrac{1}{2}+\frac{q^2-k^2 \chi}{8 k q}\log\dfrac{q^{2}+2kq-\chi k^{2}}{q^{2}-2kq-\chi k^{2}}\right)\D q.\end{align} We may apply the dominated convergence theorem to the triple integral to derive the limit value
\begin{align}A_{2}(-i\infty,k,R)\colonequals \lim_{|\chi|\to+\infty}A_{2}(-i|\chi|,k,R)=0.\end{align} Direct computation shows that $\int_{0}^{+\infty}j_1^{2}(qR)\D q=\pi/(6R)$, and the result in Proposition~\ref{prop:F1} implies that
\begin{align}&4R^{4}\int_{0}^{+\infty}\frac{k^{2}j_1^{2}(qR)}{2  q}\log\dfrac{q^{2}+2kq-\chi k^{2}}{q^{2}-2kq-\chi k^{2}}\D q\notag\\={}&-\frac{\pi i}{4k^{2}}\frac{1-e^{-2i(n-1)kR}[1+2i(n-1)kR]}{(n-1)^4}-\frac{\pi i}{4k^{2}}\frac{e^{-2i(n+1)kR}[1+2i(n+1)kR]-1}{(n+1)^4}\notag\\{}&-\frac{\pi k}{3}\left (\frac{6 i nR^2}{k (n^{2}-1)^2}+\frac{4R^3}{n^{2}-1} \right),
\end{align}thus
\begin{align}&4R^{4}\int_{0}^{+\infty}j_1^{2}(qR)\left(-\dfrac{1}{2}-\frac{k \chi}{8 q}\log\dfrac{q^{2}+2kq-\chi k^{2}}{q^{2}-2kq-\chi k^{2}}\right)\D q\notag\\={}&\frac{\pi i(n+1)}{16k^{3}}\frac{1-e^{-2i(n-1)kR}[1+2i(n-1)kR]}{(n-1)^3}\notag\\{}&+\frac{\pi i(n-1)}{16k^{3}}\frac{e^{-2i(n+1)kR}[1+2i(n+1)kR]-1}{(n+1)^3}+\frac{\pi}{2}\frac{ i nR^2}{k (n^{2}-1)}.\label{eq:A2_1}\end{align}In a similar vein as Proposition~\ref{prop:F1}, we may evaluate the integral
\begin{align}\frac{\partial}{\partial\chi}\int_{0}^{+\infty}\frac{qR^{4}j_1^{2}(qR)}{8 k}\log\dfrac{q^{2}+2kq-\chi k^{2}}{q^{2}-2kq-\chi k^{2}}\D q=-\int_{-\infty}^{+\infty}\frac{qkR^{4}j_1^{2}(qR)\D q}{8(q^{2}+2kq-\chi k^{2}) }\end{align}by considering the differential equation
\begin{align}-\frac{\partial}{\partial R}\int_{-\infty}^{+\infty}\frac{qkR^{4}j_1^{2}(qR)\D q}{8(q^{2}+2kq-\chi k^{2}) }={}&R^4\frac{\partial}{\partial R} \int_{-\infty}^{+\infty}\frac{k\sin^{2}(qR)\D q}{8qR^{2}(q^{2}+2kq-\chi k^{2}) }\notag\\={}&R^4\frac{\partial}{\partial R}\left\{\frac{i\pi }{32nkR^{2}}\left[\frac{1-e^{-2i(n-1)kR}}{1-n} +\frac{1-e^{-2i(n+1)kR}}{1+n} \right]\right\},\end{align}which  implies
\begin{align}&\frac{\partial}{\partial\chi}\int_{0}^{+\infty}\frac{qR^{4}j_1^{2}(qR)}{8 k}\log\dfrac{q^{2}+2kq-\chi k^{2}}{q^{2}-2kq-\chi k^{2}}\D q\notag\\={}&\frac{i \pi  }{32 k^3 n} \left\{2 \left[\frac{k^2 R^2}{n^{2}-1}+\frac{3 n^2+1}{ (n^2-1)^3}\right]-e^{-2 i  (n+1)k R} \frac{[ (n+1)k R-i]^2}{(n+1)^{3}}+\right.\notag\\{}&\left.+e^{-2 i  (n-1) kR}\frac{[ (n-1)k R-i]^2}{(n-1)^{3}}\right\}.\end{align}Integrating from $\infty e^{-i\pi/4} $ to $n$  in the complex $n$-plane, we obtain
\begin{align}&4R^{4}\int_{0}^{+\infty}\frac{qj_1^{2}(qR)}{8 k}\log\dfrac{q^{2}+2kq-\chi k^{2}}{q^{2}-2kq-\chi k^{2}}\D q\notag\\={}&\frac{i \pi }{4 k^3} \left\{k^2 R^2 \left[-\text{Ei}(-2 i  (n-1) kR)+\text{Ei}(-2 i  (n+1)k R)+\log \frac{n-1}{n+1}\right]-\frac{2n}{(n^{2}-1)^{2}}\right.\notag\\&\left.+\frac{e^{-2 i (n-1) kR}[1+2i(n+1)kR]}{2(n-1)^{2}}-\frac{e^{-2 i  (n+1)k R} [1+2i(n-1)kR]}{2  (n+1)^2}\right\}.\label{eq:A2_2}\end{align}Now that \eqref{eq:A2_1} and \eqref{eq:A2_2} elucidate all the contributions to $\partial A_{2}(\chi,k,R) /\partial\chi $, we may integrate in the complex $n$-plane to obtain
\begin{align}&\frac{k}{(2\pi)^{3}}\iiint_{\mathbb R^3}\frac{|\bm q\cdot\widetilde{\bm E}_\mathrm{inc}(\bm q)|^2}{|\bm q|^2-(1+\chi)k^{2}}\D^3\bm q\notag\\={}&\frac{\pi}{16 k^2 \chi^2}   \left\{-2 i\chi^{2} [2 (\chi+2)  k^2R^2-1] \left[\text{Ei}(-2 i  (n-1)k R)-\text{Ei}(-2 i  (n+1)k R)+\log\frac{n+1}{n-1}\right]\right.\notag\\&+4 i n (2 \chi ^{2}k^2  R^2-\chi-2)+ e^{-2 i  (n+1)k R}(n-1)^2 \left[2  (n+1) (\chi+2)k R+i (n^{2}+4n+1)\right]\notag\\&\left.- e^{-2 i  (n-1)k R}(n+1)^2  \left[2  (n-1)(\chi+2)k R+i (n^{2}-4n +1)\right]\phantom{\frac12}\hspace{-1em}\right\},\end{align}thereby leading to the claimed conclusion.
  \end{proof}\begin{remark}We may check the reasonability of our computations for $ \langle\bm E_\mathrm{inc} ,-\chi k\bm F_1\rangle_{V=O(\mathbf0,R)}+\langle\bm E_\mathrm{inc} ,-\chi k\bm F_2\rangle_{V=O(\mathbf0,R)}$ in the case of optically soft materials with small $ |n-1|$.

 In the  limit $n\to1$, we may derive  the asymptotic expansion
\begin{align}& \langle\bm E_\mathrm{inc} ,-\chi k\bm F_1\rangle_{V=O(\mathbf0,R)}+\langle\bm E_\mathrm{inc} ,-\chi k\bm F_2\rangle_{V=O(\mathbf0,R)}\notag\\
={}&-(n-1)\frac{8\pi  k R^3}{3} +(n-1)^2 \frac{\pi}{k^{2}}\left\{ \frac{4 k^2 R^2-1}{2} [\Si(4 k R)+i \Ci(4 k R)-i\log (4k R)-i\gamma _{0}]\right.\notag\\&\left.+i\left(2 k^4 R^4+\frac{8}{3} i k^3 R^3+\frac{5}{2} k^2 R^2-\frac{7}{16}\right)+ e^{-4 i k R} \left(\frac{ k R}{4}+\frac{7 i}{16}\right)\right\}\label{eq:F1_F2_n1}+o((n-1)^{2}),\end{align}where $\gamma_0\colonequals \lim_{M\to\infty}(\sum_{m=1}^M\frac1m-\log M)=0.577215+$ is the Euler--Mascheroni constant, and $ \Si(x)\colonequals\int^x_{0}\frac{\sin t}{t}\D t$,  $\Ci(x)\colonequals-\int_x^{+\infty}\frac{\cos t}{t}\D t$. For a transparent insulator standing in vacuum, the refractive index $n$ is a positive real number, so \eqref{eq:F1_F2_n1} leads to an asymptotic expansion for the total scattering cross-section as
\begin{align}&\I[\langle\bm E_\mathrm{inc} ,-\chi k\bm F_1\rangle_{V=O(\mathbf0,R)}+\langle\bm E_\mathrm{inc} ,-\chi k\bm F_2\rangle_{V=O(\mathbf0,R)}]\notag\\\sim{}& \pi R^2(n-1)^2\left\{\frac52+2k^2R^2-\frac{\sin(4kR)}{4kR}-\frac{7[1-\cos(4kR)]}{16k^{2}R^2}+\right.\notag\\{}&\left.+ \left( \frac{1}{2k^{2}R^2}-2 \right) [\gamma _{0}+\log (4k R)- \Ci(4 k R)]\right\},\end{align}which is exactly the Rayleigh--Gans formula (cf.~\eqref{eq:RGB}, see also  \cite[p.~90]{vandeHulst} or \cite{Rayleigh,Gans}) for $ (n-1)^2= \chi^2/4+O(\chi^3)$. It might be noted that the original derivation of the Rayleigh--Gans formula \eqref{eq:RGB} for total scattering cross-section is based on a different line of thought: obtain an angular distribution of scattering fields by a phase-shift argument, then integrate over all solid angles.
\eor
\end{remark}

\subsection{Numerical comparisons of approximation schemes\label{subsec:approx_comp}}Now, let us  consider a model problem of a glass (refractive index $n=3/2 $) sphere   immersed in water (refractive index $n=4/3 $). The relative susceptibility is given by $\chi=(9/8)^{2}-1$. This model problem may arise from a practical context in
optical microscopy~\cite{RICM}, where one wishes to use the near-field scattering pattern  to localize the center of the glass sphere with high precision (sub-wavelength resolution) in all the three spatial dimensions.

According to the computational details in \S\ref{subsec:nonpert_Mie},  the forward scattering amplitude $ \langle\bm E_\mathrm{inc},-\chi k(\hat I-\chi\hat{\mathscr G})^{-1}\bm E_\mathrm{inc} \rangle_V$ for $ V=O(\mathbf 0,R)$ and $ \I \chi<0$, as approximated by integrals in \eqref{eq:bulkapprox}, can be evaluated in closed form \eqref{eq:nonperturb}. The forward scattering amplitude $ \langle\bm E_\mathrm{inc},-\chi k(\hat I-\chi\hat{\mathscr G})^{-1}\bm E_\mathrm{inc} \rangle_{O(\mathbf 0,R)}$ remains finite when $ \I\chi=0$, and we will test our approximation  \eqref{eq:nonperturb} by analytically continuing it to the $\R\chi$-axis. \begin{figure}[!t]\begin{minipage}{.75\textwidth}\begin{tikzpicture}\pgfplotsset{xlabel style={yshift=.3cm}, ylabel style={yshift=-.8cm},width=10.5cm,height=5cm,xmajorgrids,ymajorgrids, tick label style={font=\tiny},label style={font=\tiny}}
\begin{axis}[xlabel={$kR$},ylabel={$ \sigma_N=\frac{\sigma_{\text{sc}}}{\pi R^2}$},ymin=0,ymax=4,xmin=0,xmax=160,enlargelimits=false,  minor y tick num =1, minor x tick num =1,ytick={0,1,...,5},legend style={
legend columns=1,
cells={anchor=west},at={(.98,.95)},
font=\tiny,legend style={row sep=-2.5pt},
},
   restrict y to domain=0:6]\addplot [only marks,
draw=green!50!black,mark=o,thin,mark size=1.5pt
] coordinates{(0.01,0.000000000176426)(0.843333,0.0072602)(1.67667,0.0591146)(2.51,0.16348)(3.34333,0.317605)(4.17667,0.51471)(5.01,0.743057)(5.84333,1.00075)(6.67667,1.27619)(7.51,1.56168)(8.34333,1.85223)(9.17667,2.13685)(10.01,2.40986)(10.8433,2.6643)(11.6767,2.89208)(12.51,3.09023)(13.3433,3.25305)(14.1767,3.3767)(15.01,3.46079)(15.8433,3.50344)(16.6767,3.50679)(17.51,3.47356)(18.3433,3.40521)(19.1767,3.30661)(20.01,3.1817)(20.8433,3.03527)(21.6767,2.87435)(22.51,2.70398)(23.3433,2.53091)(24.1767,2.36226)(25.01,2.20314)(25.8433,2.05948)(26.6767,1.93479)(27.51,1.83089)(28.3433,1.75039)(29.1767,1.6938)(30.01,1.66188)(30.8433,1.65548)(31.6767,1.67336)(32.51,1.7149)(33.3433,1.77777)(34.1767,1.85741)(35.01,1.94922)(35.8433,2.04672)(36.6767,2.14403)(37.51,2.23732)(38.3433,2.32315)(39.1767,2.40044)(40.01,2.46843)(40.8433,2.52524)(41.6767,2.56945)(42.51,2.5982)(43.3433,2.60824)(44.1767,2.59761)(45.01,2.56499)(45.8433,2.51275)(46.6767,2.44662)(47.51,2.37291)(48.3433,2.29784)(49.1767,2.22504)(50.01,2.15606)(50.8433,2.09188)(51.6767,2.03142)(52.51,1.97341)(53.3433,1.91802)(54.1767,1.86775)(55.01,1.82962)(55.8433,1.8104)(56.6767,1.81238)(57.51,1.8349)(58.3433,1.87461)(59.1767,1.92663)(60.01,1.98489)(60.8433,2.04211)(61.6767,2.09306)(62.51,2.13607)(63.3433,2.17604)(64.1767,2.21951)(65.01,2.26629)(65.8433,2.31281)(66.6767,2.35598)(67.51,2.39136)(68.3433,2.41257)(69.1767,2.41461)(70.01,2.39697)(70.8433,2.35752)(71.6767,2.30425)(72.51,2.25199)(73.3433,2.20447)(74.1767,2.16068)(75.01,2.12126)(75.8433,2.0875)(76.6767,2.05487)(77.51,2.01523)(78.3433,1.9779)(79.1767,1.94258)(80.01,1.89987)(80.8433,1.87376)(81.6767,1.86859)(82.51,1.8795)(83.3433,1.90375)(84.1767,1.94147)(85.01,1.9914)(85.8433,2.03412)(86.6767,2.06565)(87.51,2.11369)(88.3433,2.14465)(89.1767,2.16174)(90.01,2.19084)(90.8433,2.2219)(91.6767,2.24973)(92.51,2.2729)(93.3433,2.2943)(94.1767,2.31524)(95.01,2.29729)(95.8433,2.27232)(96.6767,2.26796)(97.51,2.22368)(98.3433,2.17557)(99.1767,2.138)(100.01,2.10168)(100.843,2.06593)(101.677,2.03258)(102.51,2.00809)(103.343,2.00211)(104.177,1.9555)(105.01,1.9322)(105.843,1.94438)(106.677,1.92882)(107.51,1.92287)(108.343,1.93414)(109.177,1.95216)(110.01,1.97475)(110.843,2.00138)(111.677,2.03474)(112.51,2.09272)(113.343,2.10071)(114.177,2.12832)(115.01,2.16883)(115.843,2.20255)(116.677,2.2104)(117.51,2.22454)(118.343,2.23137)(119.177,2.23023)(120.01,2.22268)(120.843,2.21162)(121.677,2.20809)(122.51,2.18375)(123.343,2.1557)(124.177,2.13296)(125.01,2.13638)(125.843,2.07721)(126.677,2.04383)(127.51,2.01137)(128.343,1.98211)(129.177,1.95936)(130.01,1.94497)(130.843,1.94005)(131.677,1.96749)(132.51,1.94852)(133.343,1.96064)(134.177,1.97686)(135.01,1.99806)(135.843,2.00584)(136.677,2.0233)(137.51,2.0455)(138.343,2.07198)(139.177,2.10105)(140.01,2.13071)(140.843,2.16017)(141.677,2.18536)(142.51,2.19912)(143.343,2.20498)(144.177,2.20083)(145.01,2.19048)(145.843,2.17713)(146.677,2.16413)(147.51,2.15017)(148.343,2.13518)(149.177,2.11932)(150.01,2.10219)(150.843,2.08223)(151.677,2.0559)(152.51,2.02403)(153.343,1.99334)(154.177,1.9708)(155.01,1.9568)(155.843,1.94953)(156.677,1.94986)(157.51,1.95741)(158.343,1.97121)(159.177,1.98714)(160.01,2.0019)(160.843,2.01332)(161.677,2.02226)(162.51,2.03764)(163.343,2.05884)(164.177,2.08231)(165.01,2.10661)(165.843,2.13025)(166.677,2.15414)
};
\addplot [
draw=blue,thin
] coordinates{(0.01,0.000000000209054)(0.843333,0.00802717)(1.67667,0.0635396)(2.51,0.165793)(3.34333,0.329635)(4.17667,0.540061)(5.01,0.805293)(5.84333,1.1171)(6.67667,1.48172)(7.51,1.89334)(8.34333,2.35661)(9.17667,2.86744)(10.01,3.42903)(10.8433,4.03876)(11.6767,4.69848)(12.51,5.40692)(13.3433,6.1647)(14.1767,6.97169)(15.01,7.8275)(15.8433,8.73291)(16.6767,9.68678)(17.51,10.6905)(18.3433,11.7425)(19.1767,12.8444)(20.01,13.9944)(20.8433,15.1945)(21.6767,16.4427)(22.51,17.7407)(23.3433,19.0872)(24.1767,20.4832)(25.01,21.9278)(25.8433,23.4218)(26.6767,24.9647)(27.51,26.5566)(28.3433,28.1976)(29.1767,29.8874)(30.01,31.6266)(30.8433,33.4145)(31.6767,35.2517)(32.51,37.1376)(33.3433,39.0729)(34.1767,41.0568)(35.01,43.0902)(35.8433,45.1721)(36.6767,47.3035)(37.51,49.4835)(38.3433,51.7129)(39.1767,53.991)(40.01,56.3183)(40.8433,58.6945)(41.6767,61.1197)(42.51,63.5941)(43.3433,66.1172)(44.1767,68.6897)(45.01,71.3108)(45.8433,73.9813)(46.6767,76.7004)(47.51,79.469)(48.3433,82.2861)(49.1767,85.1526)(50.01,88.0678)(50.8433,91.0323)(51.6767,94.0456)(52.51,97.1081)(53.3433,100.219)(54.1767,103.38)(55.01,106.589)(55.8433,109.848)(56.6767,113.155)(57.51,116.511)(58.3433,119.917)(59.1767,123.371)(60.01,126.875)(60.8433,130.427)(61.6767,134.029)(62.51,137.679)(63.3433,141.378)(64.1767,145.127)(65.01,148.924)(65.8433,152.771)(66.6767,156.666)(67.51,160.611)(68.3433,164.604)(69.1767,168.647)(70.01,172.738)(70.8433,176.879)(71.6767,181.068)(72.51,185.307)(73.3433,189.594)(74.1767,193.931)(75.01,198.316)(75.8433,202.751)(76.6767,207.234)(77.51,211.767)(78.3433,216.348)(79.1767,220.978)(80.01,225.658)(80.8433,230.386)(81.6767,235.164)(82.51,239.99)(83.3433,244.866)(84.1767,249.79)(85.01,254.764)(85.8433,259.786)(86.6767,264.858)(87.51,269.978)(88.3433,275.148)(89.1767,280.366)(90.01,285.634)(90.8433,290.95)(91.6767,296.316)(92.51,301.73)(93.3433,307.194)(94.1767,312.707)(95.01,318.268)(95.8433,323.879)(96.6767,329.538)(97.51,335.247)(98.3433,341.004)(99.1767,346.811)(100.01,352.666)(100.843,358.571)(101.677,364.524)(102.51,370.527)(103.343,376.578)(104.177,382.679)(105.01,388.828)(105.843,395.027)(106.677,401.274)(107.51,407.571)(108.343,413.916)(109.177,420.311)(110.01,426.754)(110.843,433.247)(111.677,439.788)(112.51,446.379)(113.343,453.018)(114.177,459.707)(115.01,466.444)(115.843,473.231)(116.677,480.066)(117.51,486.951)(118.343,493.884)(119.177,500.867)(120.01,507.898)(120.843,514.979)(121.677,522.108)(122.51,529.287)(123.343,536.514)(124.177,543.791)(125.01,551.116)(125.843,558.491)(126.677,565.914)(127.51,573.387)(128.343,580.908)(129.177,588.479)(130.01,596.098)(130.843,603.767)(131.677,611.484)(132.51,619.251)(133.343,627.066)(134.177,634.931)(135.01,642.844)(135.843,650.807)(136.677,658.818)(137.51,666.879)(138.343,674.988)(139.177,683.147)(140.01,691.354)(140.843,699.611)(141.677,707.916)(142.51,716.271)(143.343,724.674)(144.177,733.127)(145.01,741.628)(145.843,750.179)(146.677,758.778)(147.51,767.427)(148.343,776.124)(149.177,784.871)(150.01,793.666)(150.843,802.511)(151.677,811.404)(152.51,820.347)(153.343,829.338)(154.177,838.378)(155.01,847.468)(155.843,856.606)(156.677,865.794)(157.51,875.03)(158.343,884.316)(159.177,893.65)(160.01,903.034)(160.843,912.466)(161.677,921.948)(162.51,931.478)(163.343,941.057)(164.177,950.686)(165.01,960.363)(165.843,970.09)(166.677,979.865)
 };
\addplot [
draw=red,ultra thick
] coordinates{(0.01,0.000000000285687)(0.843333,0.00871643)(1.67667,0.064612)(2.51,0.166261)(3.34333,0.321883)(4.17667,0.51408)(5.01,0.742923)(5.84333,0.997381)(6.67667,1.27156)(7.51,1.55786)(8.34333,1.84658)(9.17667,2.13241)(10.01,2.40477)(10.8433,2.65952)(11.6767,2.88791)(12.51,3.08627)(13.3433,3.24948)(14.1767,3.37436)(15.01,3.4599)(15.8433,3.50413)(16.6767,3.50943)(17.51,3.47602)(18.3433,3.40839)(19.1767,3.30951)(20.01,3.18472)(20.8433,3.03956)(21.6767,2.87935)(22.51,2.71121)(23.3433,2.54007)(24.1767,2.37309)(25.01,2.21477)(25.8433,2.07081)(26.6767,1.94513)(27.51,1.84103)(28.3433,1.76147)(29.1767,1.70715)(30.01,1.67948)(30.8433,1.67704)(31.6767,1.6992)(32.51,1.74317)(33.3433,1.80617)(34.1767,1.88479)(35.01,1.9746)(35.8433,2.07201)(36.6767,2.17192)(37.51,2.27082)(38.3433,2.36406)(39.1767,2.4484)(40.01,2.52058)(40.8433,2.57792)(41.6767,2.61899)(42.51,2.64211)(43.3433,2.64759)(44.1767,2.63518)(45.01,2.60649)(45.8433,2.56293)(46.6767,2.50682)(47.51,2.44102)(48.3433,2.36819)(49.1767,2.292)(50.01,2.21514)(50.8433,2.14128)(51.6767,2.07295)(52.51,2.013)(53.3433,1.96359)(54.1767,1.92627)(55.01,1.90251)(55.8433,1.89242)(56.6767,1.89649)(57.51,1.91366)(58.3433,1.94314)(59.1767,1.98315)(60.01,2.03171)(60.8433,2.08668)(61.6767,2.14526)(62.51,2.20527)(63.3433,2.26371)(64.1767,2.31853)(65.01,2.36721)(65.8433,2.40799)(66.6767,2.43934)(67.51,2.45999)(68.3433,2.46961)(69.1767,2.46768)(70.01,2.45491)(70.8433,2.43175)(71.6767,2.39964)(72.51,2.36002)(73.3433,2.31469)(74.1767,2.26584)(75.01,2.21533)(75.8433,2.16566)(76.6767,2.11855)(77.51,2.07623)(78.3433,2.04016)(79.1767,2.01183)(80.01,1.9923)(80.8433,1.98203)(81.6767,1.98155)(82.51,1.99032)(83.3433,2.00812)(84.1767,2.03372)(85.01,2.06606)(85.8433,2.10357)(86.6767,2.14447)(87.51,2.18714)(88.3433,2.22947)(89.1767,2.26995)(90.01,2.30659)(90.8433,2.33815)(91.6767,2.36322)(92.51,2.38087)(93.3433,2.39057)(94.1767,2.39185)(95.01,2.38506)(95.8433,2.37032)(96.6767,2.34868)(97.51,2.32095)(98.3433,2.28853)(99.1767,2.25287)(100.01,2.21543)(100.843,2.17803)(101.677,2.142)(102.51,2.10913)(103.343,2.0805)(104.177,2.05746)(105.01,2.04079)(105.843,2.03111)(106.677,2.0288)(107.51,2.03367)(108.343,2.04563)(109.177,2.06381)(110.01,2.08755)(110.843,2.11559)(111.677,2.14676)(112.51,2.17967)(113.343,2.21278)(114.177,2.24487)(115.01,2.27434)(115.843,2.3002)(116.677,2.32119)(117.51,2.33663)(118.343,2.34586)(119.177,2.34855)(120.01,2.34478)(120.843,2.33462)(121.677,2.31881)(122.51,2.29788)(123.343,2.27297)(124.177,2.24508)(125.01,2.21549)(125.843,2.18551)(126.677,2.15632)(127.51,2.12933)(128.343,2.10545)(129.177,2.08588)(130.01,2.07121)(130.843,2.06213)(131.677,2.05891)(132.51,2.06157)(133.343,2.07004)(134.177,2.08371)(135.01,2.10213)(135.843,2.12425)(136.677,2.14923)(137.51,2.17588)(138.343,2.20305)(139.177,2.22962)(140.01,2.25434)(140.843,2.27634)(141.677,2.29451)(142.51,2.30831)(143.343,2.31704)(144.177,2.32047)(145.01,2.3185)(145.843,2.31122)(146.677,2.29912)(147.51,2.2826)(148.343,2.26261)(149.177,2.23987)(150.01,2.21551)(150.843,2.19052)(151.677,2.16599)(152.51,2.14302)(153.343,2.12245)(154.177,2.10532)(155.01,2.09213)(155.843,2.08356)(156.677,2.07983)(157.51,2.08108)(158.343,2.08723)(159.177,2.09788)(160.01,2.11266)(160.843,2.13073)(161.677,2.15143)(162.51,2.17372)(163.343,2.19672)(164.177,2.21938)(165.01,2.24073)(165.843,2.25993)(166.677,2.27605)
 };

\legend{Truncated Mie series,Born approximation,Semigroup asymptotics};

\node[fill=white,anchor=south,draw=black] at
(axis cs:146.75,.25) {\tiny$ n=9/8$};
\end{axis}
\end{tikzpicture}\end{minipage}\begin{minipage}{.25\textwidth}\caption{Comparison of various  approximate solutions to the total scattering cross-section of a glass bead  $n=3/2 $ immersed in water  $n=4/3 $, a system with relative refractive index $ n=9/8$. For numerical computations, the Mie series \eqref{eq:Mie_sc_series} are truncated after the first  200 terms.   \label{fig:5-1}}\end{minipage}

\end{figure}
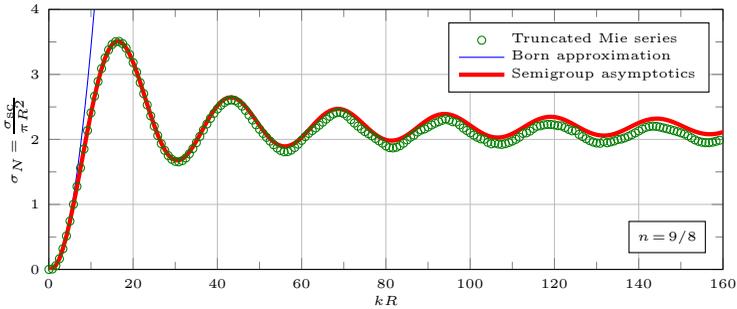

In Figure~\ref{fig:5-1}, we  normalize the total scattering cross-section $\sigma_\text{sc} $ by the geometric cross-section $\pi R^2 $ and plot the normalized cross-section $\sigma_N=\sigma_\text{sc}/(\pi R^2) $ against the normalized radius $kR$ and compare the exact solution (Mie theory) with the Born approximation \eqref{eq:RGB} and the semigroup approach [\textit{i.e.}~the total scattering cross-section  $\sigma_\text{sc} $ read off from the  imaginary part of \eqref{eq:nonperturb}]. It is graphically evident that the failure of Born approximation in the non-perturbative regime is remedied by a semigroup approach.

  \begin{figure}[!b]
\begin{minipage}{.75\textwidth}\begin{tikzpicture}\pgfplotsset{xlabel style={yshift=.3cm}, ylabel style={yshift=-.8cm},width=10.5cm,height=5cm,xmajorgrids,ymajorgrids, tick label style={font=\tiny},label style={font=\tiny}}
\begin{axis}[xlabel={$kR$},ylabel={$ \sigma_N=\frac{\sigma_{\text{sc}}}{\pi R^2}$},ymin=0,ymax=6,xmin=0,xmax=20,enlargelimits=false,  minor y tick num =1, minor x tick num =4,ytick={0,1,...,6},xtick={0,5,...,35},legend style={
legend columns=1,
cells={anchor=west},at={(.98,.95)},
font=\tiny,legend style={row sep=-2.5pt},
},
   restrict y to domain=0:6]\addplot [only marks,
draw=green!50!black,mark=o,thin,mark size=1.5pt
] coordinates{(0.01,0.00000000230682)(0.343333,0.00322711)(0.676667,0.0486471)(1.01,0.222812)(1.34333,0.550947)(1.67667,1.06397)(2.01,1.81641)(2.34333,2.24966)(2.67667,2.99393)(3.01,3.42184)(3.34333,3.75702)(3.67667,4.19165)(4.01,4.05502)(4.34333,4.36475)(4.67667,3.92882)(5.01,3.93502)(5.34333,3.47404)(5.67667,3.09257)(6.01,2.88156)(6.34333,2.28813)(6.67667,2.35549)(7.01,1.83495)(7.34333,2.02727)(7.67667,1.784)(8.01,1.80577)(8.34333,1.98156)(8.67667,1.93477)(9.01,2.24373)(9.34333,2.39625)(9.67667,2.6585)(10.01,2.87189)(10.3433,3.11075)(10.6767,3.094)(11.01,2.89086)(11.3433,2.79736)(11.6767,2.68763)(12.01,2.47039)(12.3433,2.43772)(12.6767,2.19289)(13.01,2.24962)(13.3433,2.09533)(13.6767,2.13213)(14.01,1.97719)(14.3433,1.85361)(14.6767,2.03353)(15.01,1.98631)(15.3433,2.24742)(15.6767,2.25545)(16.01,2.42866)(16.3433,2.51659)(16.6767,2.59241)(17.01,2.66854)(17.3433,2.56431)(17.6767,2.82632)(18.01,2.45552)(18.3433,2.49244)(18.6767,2.2926)(19.01,2.18756)(19.3433,2.1362)(19.6767,1.97622)(20.01,2.03192)(20.3433,1.9253)(20.6767,2.0601)(21.01,2.0266)(21.3433,2.1568)(21.6767,2.21056)(22.01,2.22918)(22.3433,2.34921)(22.6767,2.3854)(23.01,2.41534)(23.3433,2.46638)(23.6767,2.5018)(24.01,2.4565)(24.3433,2.46121)(24.6767,2.38959)(25.01,2.24716)(25.3433,2.11985)(25.6767,2.11407)(26.01,1.98239)(26.3433,2.01456)(26.6767,1.97439)(27.01,2.01118)(27.3433,2.2014)(27.6767,2.13704)(28.01,2.20515)(28.3433,2.16735)(28.6767,2.33908)(29.01,2.29805)(29.3433,2.37458)(29.6767,2.38202)(30.01,2.35195)(30.3433,2.38544)(30.6767,2.42125)(31.01,2.32017)(31.3433,2.15963)(31.6767,2.20465)(32.01,2.06235)(32.3433,2.03745)(32.6767,2.03117)(33.01,1.9857)(33.3433,2.03367)
};
\addplot [
draw=blue,thin
] coordinates{(0.01,0.0000499997)(0.0933333,0.00435345)(0.176667,0.0155785)(0.26,0.0336733)(0.343333,0.058554)(0.426667,0.0901058)(0.51,0.128183)(0.593333,0.17261)(0.676667,0.223181)(0.76,0.279665)(0.843333,0.341802)(0.926667,0.409307)(1.01,0.481871)(1.09333,0.559163)(1.17667,0.64083)(1.26,0.726504)(1.34333,0.815794)(1.42667,0.908301)(1.51,1.00361)(1.59333,1.10129)(1.67667,1.2009)(1.76,1.30202)(1.84333,1.40419)(1.92667,1.50696)(2.01,1.6099)(2.09333,1.71255)(2.17667,1.81447)(2.26,1.91523)(2.34333,2.01442)(2.42667,2.11161)(2.51,2.20641)(2.59333,2.29844)(2.67667,2.38733)(2.76,2.47272)(2.84333,2.5543)(2.92667,2.63176)(3.01,2.70481)(3.09333,2.77319)(3.17667,2.83668)(3.26,2.89507)(3.34333,2.94817)(3.42667,2.99585)(3.51,3.03797)(3.59333,3.07444)(3.67667,3.1052)(3.76,3.13021)(3.84333,3.14947)(3.92667,3.163)(4.01,3.17086)(4.09333,3.17311)(4.17667,3.16987)(4.26,3.16127)(4.34333,3.14747)(4.42667,3.12864)(4.51,3.105)(4.59333,3.07677)(4.67667,3.04419)(4.76,3.00752)(4.84333,2.96706)(4.92667,2.9231)(5.01,2.87594)(5.09333,2.8259)(5.17667,2.77333)(5.26,2.71856)(5.34333,2.66192)(5.42667,2.60378)(5.51,2.54448)(5.59333,2.48437)(5.67667,2.42379)(5.76,2.3631)(5.84333,2.30263)(5.92667,2.24272)(6.01,2.18367)(6.09333,2.12582)(6.17667,2.06945)(6.26,2.01484)(6.34333,1.96227)(6.42667,1.912)(6.51,1.86425)(6.59333,1.81923)(6.67667,1.77716)(6.76,1.73819)(6.84333,1.7025)(6.92667,1.6702)(7.01,1.6414)(7.09333,1.6162)(7.17667,1.59466)(7.26,1.57682)(7.34333,1.5627)(7.42667,1.5523)(7.51,1.54558)(7.59333,1.54251)(7.67667,1.54301)(7.76,1.547)(7.84333,1.55437)(7.92667,1.56499)(8.01,1.57872)(8.09333,1.5954)(8.17667,1.61485)(8.26,1.63689)(8.34333,1.66131)(8.42667,1.68791)(8.51,1.71646)(8.59333,1.74672)(8.67667,1.77847)(8.76,1.81147)(8.84333,1.84546)(8.92667,1.88021)(9.01,1.91546)(9.09333,1.95097)(9.17667,1.9865)(9.26,2.02181)(9.34333,2.05666)(9.42667,2.09083)(9.51,2.1241)(9.59333,2.15626)(9.67667,2.18711)(9.76,2.21647)(9.84333,2.24417)(9.92667,2.27003)(10.01,2.29393)(10.0933,2.31572)(10.1767,2.3353)(10.26,2.35256)(10.3433,2.36742)(10.4267,2.37982)(10.51,2.38971)(10.5933,2.39706)(10.6767,2.40185)(10.76,2.40409)(10.8433,2.4038)(10.9267,2.40102)(11.01,2.39579)(11.0933,2.38819)(11.1767,2.37829)(11.26,2.3662)(11.3433,2.35201)(11.4267,2.33587)(11.51,2.31788)(11.5933,2.29821)(11.6767,2.277)(11.76,2.25441)(11.8433,2.23061)(11.9267,2.20577)(12.01,2.18007)(12.0933,2.1537)(12.1767,2.12682)(12.26,2.09964)(12.3433,2.07233)(12.4267,2.04508)(12.51,2.01806)(12.5933,1.99145)(12.6767,1.96542)(12.76,1.94014)(12.8433,1.91576)(12.9267,1.89244)(13.01,1.87032)(13.0933,1.84953)(13.1767,1.83018)(13.26,1.8124)(13.3433,1.79627)(13.4267,1.78189)(13.51,1.76932)(13.5933,1.75864)(13.6767,1.74988)(13.76,1.74308)(13.8433,1.73826)(13.9267,1.73544)(14.01,1.73459)(14.0933,1.73571)(14.1767,1.73875)(14.26,1.74369)(14.3433,1.75045)(14.4267,1.75898)(14.51,1.76919)(14.5933,1.78099)(14.6767,1.79428)(14.76,1.80895)(14.8433,1.8249)(14.9267,1.84199)(15.01,1.8601)(15.0933,1.87908)(15.1767,1.8988)(15.26,1.91913)(15.3433,1.9399)(15.4267,1.96098)(15.51,1.98221)(15.5933,2.00345)(15.6767,2.02456)(15.76,2.04539)(15.8433,2.0658)(15.9267,2.08565)(16.01,2.10482)(16.0933,2.12319)(16.1767,2.14062)(16.26,2.15702)(16.3433,2.17228)(16.4267,2.18631)(16.51,2.19902)(16.5933,2.21034)(16.6767,2.22021)(16.76,2.22856)(16.8433,2.23537)(16.9267,2.24059)(17.01,2.24421)(17.0933,2.24621)(17.1767,2.2466)(17.26,2.24539)(17.3433,2.2426)(17.4267,2.23826)(17.51,2.23242)(17.5933,2.22513)(17.6767,2.21645)(17.76,2.20646)(17.8433,2.19523)(17.9267,2.18285)(18.01,2.16942)(18.0933,2.15503)(18.1767,2.13979)(18.26,2.12382)(18.3433,2.10722)(18.4267,2.09013)(18.51,2.07264)(18.5933,2.0549)(18.6767,2.03701)(18.76,2.01912)(18.8433,2.00132)(18.9267,1.98375)(19.01,1.96653)(19.0933,1.94976)(19.1767,1.93356)(19.26,1.91803)(19.3433,1.90327)(19.4267,1.88938)(19.51,1.87643)(19.5933,1.86453)(19.6767,1.85372)(19.76,1.84409)(19.8433,1.83567)(19.9267,1.82853)(20.01,1.8227)(20.0933,1.8182)(20.1767,1.81506)(20.26,1.81329)(20.3433,1.81288)(20.4267,1.81383)(20.51,1.81612)(20.5933,1.81972)(20.6767,1.82459)(20.76,1.8307)(20.8433,1.83799)
};
\addplot [
draw=blue,thin,dashed
] coordinates{(0.01,0.0000999994)(0.0933333,0.00867316)(0.176667,0.0305067)(0.26,0.0644547)(0.343333,0.109491)(0.426667,0.164778)(0.51,0.229607)(0.593333,0.30334)(0.676667,0.385382)(0.76,0.475162)(0.843333,0.572115)(0.926667,0.67568)(1.01,0.785295)(1.09333,0.900391)(1.17667,1.02039)(1.26,1.14472)(1.34333,1.27278)(1.42667,1.40399)(1.51,1.53775)(1.59333,1.67346)(1.67667,1.81053)(1.76,1.94836)(1.84333,2.08636)(1.92667,2.22395)(2.01,2.36055)(2.09333,2.49559)(2.17667,2.62854)(2.26,2.75885)(2.34333,2.886)(2.42667,3.00951)(2.51,3.12889)(2.59333,3.24371)(2.67667,3.35353)(2.76,3.45796)(2.84333,3.55663)(2.92667,3.64921)(3.01,3.73539)(3.09333,3.81492)(3.17667,3.88754)(3.26,3.95306)(3.34333,4.01132)(3.42667,4.06219)(3.51,4.10558)(3.59333,4.14142)(3.67667,4.1697)(3.76,4.19044)(3.84333,4.2037)(3.92667,4.20955)(4.01,4.20811)(4.09333,4.19955)(4.17667,4.18405)(4.26,4.16183)(4.34333,4.13313)(4.42667,4.09822)(4.51,4.0574)(4.59333,4.01101)(4.67667,3.95937)(4.76,3.90287)(4.84333,3.84187)(4.92667,3.77678)(5.01,3.70802)(5.09333,3.63599)(5.17667,3.56115)(5.26,3.48391)(5.34333,3.40473)(5.42667,3.32405)(5.51,3.24231)(5.59333,3.15994)(5.67667,3.07739)(5.76,2.99509)(5.84333,2.91344)(5.92667,2.83285)(6.01,2.75371)(6.09333,2.6764)(6.17667,2.60128)(6.26,2.52868)(6.34333,2.45893)(6.42667,2.39231)(6.51,2.3291)(6.59333,2.26956)(6.67667,2.21389)(6.76,2.1623)(6.84333,2.11496)(6.92667,2.07201)(7.01,2.03356)(7.09333,1.99971)(7.17667,1.9705)(7.26,1.94597)(7.34333,1.92612)(7.42667,1.91093)(7.51,1.90035)(7.59333,1.8943)(7.67667,1.89268)(7.76,1.89536)(7.84333,1.90221)(7.92667,1.91304)(8.01,1.92768)(8.09333,1.94591)(8.17667,1.9675)(8.26,1.99223)(8.34333,2.01983)(8.42667,2.05004)(8.51,2.08258)(8.59333,2.11716)(8.67667,2.1535)(8.76,2.19131)(8.84333,2.23027)(8.92667,2.27009)(9.01,2.31047)(9.09333,2.35111)(9.17667,2.39173)(9.26,2.43203)(9.34333,2.47174)(9.42667,2.51058)(9.51,2.5483)(9.59333,2.58465)(9.67667,2.61941)(9.76,2.65234)(9.84333,2.68324)(9.92667,2.71194)(10.01,2.73827)(10.0933,2.76206)(10.1767,2.7832)(10.26,2.80158)(10.3433,2.81709)(10.4267,2.82968)(10.51,2.83929)(10.5933,2.84589)(10.6767,2.84947)(10.76,2.85004)(10.8433,2.84764)(10.9267,2.8423)(11.01,2.8341)(11.0933,2.82313)(11.1767,2.80949)(11.26,2.79328)(11.3433,2.77466)(11.4267,2.75376)(11.51,2.73075)(11.5933,2.70579)(11.6767,2.67908)(11.76,2.65079)(11.8433,2.62113)(11.9267,2.59031)(12.01,2.55854)(12.0933,2.52602)(12.1767,2.49298)(12.26,2.45963)(12.3433,2.42618)(12.4267,2.39285)(12.51,2.35984)(12.5933,2.32737)(12.6767,2.29563)(12.76,2.26481)(12.8433,2.23509)(12.9267,2.20664)(13.01,2.17964)(13.0933,2.15422)(13.1767,2.13052)(13.26,2.10868)(13.3433,2.0888)(13.4267,2.07097)(13.51,2.05529)(13.5933,2.04182)(13.6767,2.0306)(13.76,2.02167)(13.8433,2.01507)(13.9267,2.01078)(14.01,2.00879)(14.0933,2.0091)(14.1767,2.01164)(14.26,2.01637)(14.3433,2.02322)(14.4267,2.0321)(14.51,2.04293)(14.5933,2.05559)(14.6767,2.06997)(14.76,2.08593)(14.8433,2.10336)(14.9267,2.12209)(15.01,2.14199)(15.0933,2.16288)(15.1767,2.18462)(15.26,2.20703)(15.3433,2.22995)(15.4267,2.2532)(15.51,2.27663)(15.5933,2.30005)(15.6767,2.3233)(15.76,2.34622)(15.8433,2.36864)(15.9267,2.39042)(16.01,2.41141)(16.0933,2.43146)(16.1767,2.45044)(16.26,2.46823)(16.3433,2.4847)(16.4267,2.49976)(16.51,2.51331)(16.5933,2.52527)(16.6767,2.53557)(16.76,2.54414)(16.8433,2.55094)(16.9267,2.55594)(17.01,2.55911)(17.0933,2.56045)(17.1767,2.55995)(17.26,2.55763)(17.3433,2.55352)(17.4267,2.54766)(17.51,2.54011)(17.5933,2.53091)(17.6767,2.52015)(17.76,2.50791)(17.8433,2.49428)(17.9267,2.47936)(18.01,2.46326)(18.0933,2.44609)(18.1767,2.42797)(18.26,2.40904)(18.3433,2.38942)(18.4267,2.36924)(18.51,2.34865)(18.5933,2.32778)(18.6767,2.30678)(18.76,2.28577)(18.8433,2.2649)(18.9267,2.2443)(19.01,2.22411)(19.0933,2.20445)(19.1767,2.18545)(19.26,2.16722)(19.3433,2.14989)(19.4267,2.13354)(19.51,2.11828)(19.5933,2.1042)(19.6767,2.09138)(19.76,2.07989)(19.8433,2.06979)(19.9267,2.06113)(20.01,2.05395)(20.0933,2.04829)(20.1767,2.04416)(20.26,2.04158)(20.3433,2.04053)(20.4267,2.04102)(20.51,2.04302)(20.5933,2.0465)(20.6767,2.05141)(20.76,2.05771)(20.8433,2.06533)(20.9267,2.07421)(21.01,2.08428)(21.0933,2.09544)(21.1767,2.10761)(21.26,2.12069)(21.3433,2.13458)(21.4267,2.14917)(21.51,2.16435)(21.5933,2.18002)(21.6767,2.19604)(21.76,2.21232)(21.8433,2.22872)(21.9267,2.24513)(22.01,2.26143)(22.0933,2.27751)(22.1767,2.29326)(22.26,2.30856)(22.3433,2.32331)(22.4267,2.33741)(22.51,2.35077)(22.5933,2.36328)(22.6767,2.37488)(22.76,2.38548)(22.8433,2.39501)(22.9267,2.40342)(23.01,2.41064)(23.0933,2.41664)(23.1767,2.42138)(23.26,2.42483)(23.3433,2.42698)(23.4267,2.42781)(23.51,2.42733)(23.5933,2.42554)(23.6767,2.42247)(23.76,2.41813)(23.8433,2.41257)(23.9267,2.40583)(24.01,2.39796)(24.0933,2.38902)(24.1767,2.37907)(24.26,2.36819)(24.3433,2.35645)(24.4267,2.34394)(24.51,2.33074)(24.5933,2.31695)(24.6767,2.30266)(24.76,2.28796)(24.8433,2.27297)(24.9267,2.25779)(25.01,2.2425)(25.0933,2.22722)(25.1767,2.21205)(25.26,2.19708)(25.3433,2.18242)(25.4267,2.16816)(25.51,2.1544)(25.5933,2.14121)(25.6767,2.12869)(25.76,2.1169)(25.8433,2.10593)(25.9267,2.09585)(26.01,2.0867)(26.0933,2.07854)(26.1767,2.07143)(26.26,2.06539)(26.3433,2.06047)(26.4267,2.05667)(26.51,2.05403)(26.5933,2.05254)(26.6767,2.0522)(26.76,2.05301)(26.8433,2.05495)(26.9267,2.058)(27.01,2.06212)(27.0933,2.06728)(27.1767,2.07343)(27.26,2.08052)(27.3433,2.08849)(27.4267,2.09728)(27.51,2.10682)(27.5933,2.11703)(27.6767,2.12784)(27.76,2.13917)(27.8433,2.15093)(27.9267,2.16304)(28.01,2.17541)(28.0933,2.18795)(28.1767,2.20056)(28.26,2.21316)(28.3433,2.22566)(28.4267,2.23797)(28.51,2.25)(28.5933,2.26168)(28.6767,2.27291)(28.76,2.28362)(28.8433,2.29374)(28.9267,2.30321)(29.01,2.31194)(29.0933,2.31989)(29.1767,2.32701)(29.26,2.33325)(29.3433,2.33856)(29.4267,2.34291)(29.51,2.34629)(29.5933,2.34865)(29.6767,2.35001)(29.76,2.35034)(29.8433,2.34964)(29.9267,2.34794)(30.01,2.34523)
};
\addplot [
draw=red,ultra thick
] coordinates{(0.01,0.00000000694673)(0.0933333,0.0000523994)(0.176667,0.000662916)(0.26,0.00303747)(0.343333,0.00894172)(0.426667,0.020468)(0.51,0.0397629)(0.593333,0.0687533)(0.676667,0.10891)(0.76,0.161079)(0.843333,0.225407)(0.926667,0.301358)(1.01,0.38783)(1.09333,0.483336)(1.17667,0.586231)(1.26,0.694944)(1.34333,0.808177)(1.42667,0.925053)(1.51,1.04518)(1.59333,1.16862)(1.67667,1.2958)(1.76,1.42734)(1.84333,1.56386)(1.92667,1.70577)(2.01,1.85311)(2.09333,2.00545)(2.17667,2.16181)(2.26,2.32074)(2.34333,2.48042)(2.42667,2.63883)(2.51,2.79394)(2.59333,2.9439)(2.67667,3.08724)(2.76,3.22294)(2.84333,3.35048)(2.92667,3.46983)(3.01,3.58134)(3.09333,3.68558)(3.17667,3.78322)(3.26,3.87478)(3.34333,3.96058)(3.42667,4.04059)(3.51,4.11445)(3.59333,4.18147)(3.67667,4.24074)(3.76,4.29127)(3.84333,4.33211)(3.92667,4.36253)(4.01,4.38211)(4.09333,4.39079)(4.17667,4.38892)(4.26,4.37718)(4.34333,4.35649)(4.42667,4.32791)(4.51,4.29247)(4.59333,4.2511)(4.67667,4.20447)(4.76,4.15303)(4.84333,4.09692)(4.92667,4.03609)(5.01,3.97037)(5.09333,3.89956)(5.17667,3.82356)(5.26,3.74245)(5.34333,3.65658)(5.42667,3.56657)(5.51,3.47329)(5.59333,3.37779)(5.67667,3.28125)(5.76,3.18479)(5.84333,3.08947)(5.92667,2.99611)(6.01,2.90532)(6.09333,2.81743)(6.17667,2.73256)(6.26,2.65066)(6.34333,2.5716)(6.42667,2.49528)(6.51,2.42167)(6.59333,2.35092)(6.67667,2.28336)(6.76,2.2195)(6.84333,2.15998)(6.92667,2.10551)(7.01,2.05678)(7.09333,2.01436)(7.17667,1.97868)(7.26,1.94991)(7.34333,1.92801)(7.42667,1.91272)(7.51,1.90359)(7.59333,1.90011)(7.67667,1.90172)(7.76,1.90796)(7.84333,1.91845)(7.92667,1.93299)(8.01,1.95151)(8.09333,1.97409)(8.17667,2.00085)(8.26,2.03194)(8.34333,2.06743)(8.42667,2.10728)(8.51,2.15126)(8.59333,2.19896)(8.67667,2.24981)(8.76,2.30311)(8.84333,2.35807)(8.92667,2.41392)(9.01,2.46996)(9.09333,2.5256)(9.17667,2.58042)(9.26,2.63416)(9.34333,2.6867)(9.42667,2.73803)(9.51,2.78817)(9.59333,2.83714)(9.67667,2.88488)(9.76,2.93123)(9.84333,2.97587)(9.92667,3.0184)(10.01,3.0583)(10.0933,3.09502)(10.1767,3.12803)(10.26,3.15689)(10.3433,3.18126)(10.4267,3.20098)(10.51,3.21606)(10.5933,3.22664)(10.6767,3.233)(10.76,3.23546)(10.8433,3.23436)(10.9267,3.22999)(11.01,3.22256)(11.0933,3.21215)(11.1767,3.19876)(11.26,3.18231)(11.3433,3.16265)(11.4267,3.13967)(11.51,3.11331)(11.5933,3.08359)(11.6767,3.05069)(11.76,3.01492)(11.8433,2.97671)(11.9267,2.93658)(12.01,2.8951)(12.0933,2.85284)(12.1767,2.81031)(12.26,2.76792)(12.3433,2.72597)(12.4267,2.68464)(12.51,2.644)(12.5933,2.60408)(12.6767,2.56485)(12.76,2.52633)(12.8433,2.4886)(12.9267,2.45181)(13.01,2.41623)(13.0933,2.38221)(13.1767,2.35018)(13.26,2.32056)(13.3433,2.29379)(13.4267,2.27019)(13.51,2.25002)(13.5933,2.2334)(13.6767,2.22032)(13.76,2.21066)(13.8433,2.20422)(13.9267,2.20075)(14.01,2.20003)(14.0933,2.20185)(14.1767,2.20609)(14.26,2.21271)(14.3433,2.22175)(14.4267,2.2333)(14.51,2.24747)(14.5933,2.26438)(14.6767,2.28407)(14.76,2.30649)(14.8433,2.33149)(14.9267,2.35882)(15.01,2.38811)(15.0933,2.41895)(15.1767,2.45088)(15.26,2.48346)(15.3433,2.51631)(15.4267,2.54912)(15.51,2.58166)(15.5933,2.61382)(15.6767,2.64554)(15.76,2.67682)(15.8433,2.70767)(15.9267,2.73804)(16.01,2.76787)(16.0933,2.79699)(16.1767,2.82517)(16.26,2.8521)(16.3433,2.87742)(16.4267,2.90076)(16.51,2.9218)(16.5933,2.94026)(16.6767,2.95596)(16.76,2.96881)(16.8433,2.97883)(16.9267,2.98613)(17.01,2.99087)(17.0933,2.99325)(17.1767,2.99344)(17.26,2.99158)(17.3433,2.98774)(17.4267,2.98193)(17.51,2.9741)(17.5933,2.96416)(17.6767,2.952)(17.76,2.93753)(17.8433,2.92074)(17.9267,2.90168)(18.01,2.88048)(18.0933,2.85738)(18.1767,2.83268)(18.26,2.80675)(18.3433,2.77995)(18.4267,2.75263)(18.51,2.72512)(18.5933,2.69763)(18.6767,2.67033)(18.76,2.64331)(18.8433,2.6166)(18.9267,2.59021)(19.01,2.56412)(19.0933,2.53839)(19.1767,2.51309)(19.26,2.48839)(19.3433,2.46449)(19.4267,2.44169)(19.51,2.42027)(19.5933,2.40056)(19.6767,2.38283)(19.76,2.36729)(19.8433,2.35408)(19.9267,2.34326)(20.01,2.3348)(20.0933,2.32858)(20.1767,2.32448)(20.26,2.32234)(20.3433,2.32201)(20.4267,2.32341)(20.51,2.32649)(20.5933,2.33126)(20.6767,2.33778)(20.76,2.34615)(20.8433,2.35644)
 };

\legend{Truncated Mie series,Van de Hulst approximation,Evans--Fournier approximation,Semigroup asymptotics};

\node[fill=white,anchor=south,draw=black] at
(axis cs:18.4,.25) {\tiny$ n=3/2$};
\end{axis}
\end{tikzpicture}
\vspace{-.75em}
\begin{center}\begin{tiny}\quad\,(a)\end{tiny}\end{center}\begin{tikzpicture}\pgfplotsset{xlabel style={yshift=.3cm}, ylabel style={yshift=-.8cm},width=10.5cm,height=5cm,xmajorgrids,ymajorgrids, tick label style={font=\tiny},label style={font=\tiny}}
\begin{axis}[xlabel={$kR$},ylabel={$ \sigma_N=\frac{\sigma_{\text{sc}}}{\pi R^2}$},ymin=0,ymax=6,xmin=0,xmax=30,enlargelimits=false,  minor y tick num =1, minor x tick num =4,ytick={0,1,...,6},xtick={0,5,...,35},legend style={
legend columns=1,
cells={anchor=west},at={(.98,.95)},
font=\tiny,legend style={row sep=-2.5pt},
},
   restrict y to domain=0:6]\addplot [only marks,
draw=green!50!black,mark=o,thin,mark size=1.5pt
] coordinates{(0.01,0.00000000113033)(0.343333,0.00155574)(0.676667,0.0225093)(1.01,0.0992157)(1.34333,0.243916)(1.67667,0.439388)(2.01,0.740065)(2.34333,1.09484)(2.67667,1.40118)(3.01,1.80054)(3.34333,2.16727)(3.67667,2.49425)(4.01,2.87131)(4.34333,3.12348)(4.67667,3.42204)(5.01,3.62494)(5.34333,3.74581)(5.67667,3.94028)(6.01,3.89306)(6.34333,3.9496)(6.67667,3.9137)(7.01,3.71844)(7.34333,3.72076)(7.67667,3.43257)(8.01,3.28306)(8.34333,3.11513)(8.67667,2.77573)(9.01,2.73265)(9.34333,2.38561)(9.67667,2.25269)(10.01,2.14183)(10.3433,1.8517)(10.6767,1.96326)(11.01,1.73268)(11.3433,1.74382)(11.6767,1.81325)(12.01,1.67218)(12.3433,1.95305)(12.6767,1.87917)(13.01,2.02803)(13.3433,2.19725)(13.6767,2.14273)(14.01,2.47867)(14.3433,2.44334)(14.6767,2.61101)(15.01,2.74386)(15.3433,2.6266)(15.6767,2.85715)(16.01,2.76918)(16.3433,2.86242)(16.6767,2.86133)(17.01,2.63127)(17.3433,2.70368)(17.6767,2.54427)(18.01,2.60176)(18.3433,2.50437)(18.6767,2.21513)(19.01,2.22661)(19.3433,2.06734)(19.6767,2.12232)(20.01,2.08169)(20.3433,1.86611)(20.6767,1.91977)(21.01,1.8294)(21.3433,1.8957)(21.6767,1.98034)(22.01,1.93663)(22.3433,2.06127)(22.6767,2.01353)(23.01,2.11004)(23.3433,2.23003)(23.6767,2.38475)(24.01,2.4644)(24.3433,2.36226)(24.6767,2.43588)(25.01,2.49436)(25.3433,2.50686)(25.6767,2.69883)(26.01,2.48317)(26.3433,2.4979)(26.6767,2.46633)(27.01,2.37489)(27.3433,2.42308)(27.6767,2.31952)(28.01,2.2642)(28.3433,2.18441)(28.6767,2.07833)(29.01,2.10294)(29.3433,1.99092)(29.6767,2.10721)(30.01,1.95941)(30.3433,1.90433)(30.6767,1.98236)(31.01,1.91366)(31.3433,2.0204)(31.6767,2.0605)(32.01,2.02011)(32.3433,2.1389)(32.6767,2.10914)(33.01,2.24687)(33.3433,2.25181)(33.6767,2.37286)(34.01,2.37833)(34.3433,2.32421)(34.6767,2.43735)(35.01,2.3988)(35.3433,2.38659)(35.6767,2.43471)(36.01,2.32761)(36.3433,2.37991)(36.6767,2.29687)(37.01,2.26289)(37.3433,2.258)(37.6767,2.12529)(38.01,2.15733)(38.3433,2.06572)(38.6767,2.06256)(39.01,2.05472)(39.3433,1.94507)(39.6767,1.99403)(40.01,1.96035)(40.3433,2.01685)(40.6767,2.02902)(41.01,1.97574)(41.3433,2.06584)(41.6767,2.05697)(42.01,2.14739)(42.3433,2.18603)(42.6767,2.16379)(43.01,2.26954)(43.3433,2.2455)(43.6767,2.28826)(44.01,2.32964)(44.3433,2.30406)(44.6767,2.38732)(45.01,2.30208)(45.3433,2.29061)(45.6767,2.29602)(46.01,2.26584)(46.3433,2.31776)(46.6767,2.17226)(47.01,2.13592)(47.3433,2.12321)(47.6767,2.13062)(48.01,2.12331)(48.3433,2.00468)(48.6767,1.99197)(49.01,1.99367)(49.3433,1.96221)(49.6767,2.02476)(50.01,1.97647)(50.3433,2.01367)(50.6767,2.04068)(51.01,2.02965)(51.3433,2.12397)(51.6767,2.11557)(52.01,2.18349)(52.3433,2.19633)(52.6767,2.18631)(53.01,2.27258)(53.3433,2.33491)(53.6767,2.35223)(54.01,2.28469)(54.3433,2.24437)(54.6767,2.29933)(55.01,2.22775)(55.3433,2.26019)(55.6767,2.21272)(56.01,2.14741)(56.3433,2.17787)(56.6767,2.08818)(57.01,2.11672)(57.3433,2.13165)(57.6767,2.01289)(58.01,2.05019)(58.3433,1.97253)(58.6767,2.02266)(59.01,1.99649)(59.3433,1.97992)(59.6767,2.05268)(60.01,2.00221)(60.3433,2.0692)(60.6767,2.07664)(61.01,2.09161)(61.3433,2.25007)(61.6767,2.13196)(62.01,2.19341)(62.3433,2.2038)(62.6767,2.22468)(63.01,2.27934)(63.3433,2.21585)(63.6767,2.24785)(64.01,2.23345)(64.3433,2.24848)(64.6767,2.24476)(65.01,2.16229)(65.3433,2.176)(65.6767,2.13383)(66.01,2.12351)(66.3433,2.11356)(66.6767,2.03569)(67.01,2.06282)(67.3433,2.01864)(67.6767,1.99077)(68.01,2.0236)(68.3433,1.97801)(68.6767,2.04891)(69.01,2.00907)(69.3433,1.99978)(69.6767,2.06426)(70.01,2.05157)(70.3433,2.155)(70.6767,2.11071)(71.01,2.10928)(71.3433,2.17919)(71.6767,2.21196)(72.01,2.23484)(72.3433,2.20569)(72.6767,2.19063)(73.01,2.23576)(73.3433,2.19422)(73.6767,2.23202)(74.01,2.18921)(74.3433,2.15607)(74.6767,2.17491)(75.01,2.10663)(75.3433,2.13161)(75.6767,2.08476)(76.01,2.05943)(76.3433,2.06433)(76.6767,2.0024)(77.01,2.03482)(77.3433,2.05705)(77.6767,2.02811)(78.01,2.02613)(78.3433,1.98723)(78.6767,2.04122)(79.01,2.03449)(79.3433,2.05972)(79.6767,2.09588)(80.01,2.11378)(80.3433,2.13643)(80.6767,2.13398)(81.01,2.17086)(81.3433,2.19012)(81.6767,2.16453)(82.01,2.21484)(82.3433,2.1893)(82.6767,2.20647)(83.01,2.20262)(83.3433,2.16226)(83.6767,2.19199)(84.01,2.14026)(84.3433,2.12541)(84.6767,2.11651)(85.01,2.06916)(85.3433,2.10093)(85.6767,2.04092)(86.01,2.01975)(86.3433,2.0281)(86.6767,1.99995)(87.01,2.04362)(87.3433,2.00403)(87.6767,1.99684)(88.01,2.02824)(88.3433,2.0381)(88.6767,2.07525)(89.01,2.06029)(89.3433,2.0734)(89.6767,2.11042)(90.01,2.14557)(90.3433,2.15975)(90.6767,2.14815)(91.01,2.16151)(91.3433,2.18042)(91.6767,2.1687)(92.01,2.19367)(92.3433,2.16752)(92.6767,2.17615)(93.01,2.15875)(93.3433,2.12179)(93.6767,2.13652)(94.01,2.09778)(94.3433,2.12622)(94.6767,2.07191)(95.01,2.02945)(95.3433,2.04774)(95.6767,2.01354)(96.01,2.01921)(96.3433,2.01343)(96.6767,1.99023)(97.01,2.02551)(97.3433,2.00837)(97.6767,2.04235)(98.01,2.04468)(98.3433,2.02942)(98.6767,2.0334)(99.01,2.05872)(99.3433,2.05586)(99.6767,2.08738)(100.01,2.09116)(100.343,2.101)(100.677,2.11693)(101.01,2.10639)(101.343,2.12724)(101.677,2.1135)(102.01,2.06598)(102.343,2.0363)(102.677,2.00493)(103.01,1.96068)(103.343,1.93929)(103.677,1.89397)(104.01,1.86999)(104.343,1.84339)(104.677,1.81444)(105.01,1.82018)(105.343,1.80921)(105.677,1.78793)(106.01,1.77486)(106.343,1.75056)(106.677,1.73828)(107.01,1.72557)(107.343,1.7085)(107.677,1.70421)(108.01,1.68675)(108.343,1.68213)(108.677,1.68845)(109.01,1.70492)(109.343,1.72259)(109.677,1.73502)(110.01,1.75104)(110.343,1.75964)(110.677,1.76447)(111.01,1.76781)(111.343,1.76058)(111.677,1.74512)(112.01,1.70963)(112.343,1.67238)(112.677,1.63773)(113.01,1.59287)(113.343,1.56049)(113.677,1.52154)(114.01,1.48755)(114.343,1.46499)(114.677,1.44087)(115.01,1.44135)(115.343,1.44533)(115.677,1.44679)(116.01,1.46258)(116.343,1.46789)(116.677,1.48373)(117.01,1.49734)(117.343,1.50396)(117.677,1.51832)(118.01,1.51069)(118.343,1.50369)(118.677,1.49842)(119.01,1.481)(119.343,1.47497)(119.677,1.45793)(120.01,1.44218)(120.343,1.43128)(120.677,1.41096)(121.01,1.40685)(121.343,1.40025)(121.677,1.39291)(122.01,1.39043)(122.343,1.37705)(122.677,1.37245)(123.01,1.36165)(123.343,1.34759)(123.677,1.34043)(124.01,1.3183)(124.343,1.29843)(124.677,1.27763)(125.01,1.25523)(125.343,1.24158)(125.677,1.22376)(126.01,1.21489)(126.343,1.20815)(126.677,1.20192)(127.01,1.20992)(127.343,1.21819)(127.677,1.2338)(128.01,1.2511)(128.343,1.26429)(128.677,1.28239)(129.01,1.29309)(129.343,1.30317)(129.677,1.31033)(130.01,1.30669)(130.343,1.29843)(130.677,1.27978)(131.01,1.25911)(131.343,1.23456)(131.677,1.20351)(132.01,1.17589)(132.343,1.14304)(132.677,1.11126)(133.01,1.08518)(133.343,1.06281)(133.677,1.04857)(134.01,1.03728)(134.343,1.03329)(134.677,1.03485)(135.01,1.03956)(135.343,1.05367)(135.677,1.06959)(136.01,1.08815)(136.343,1.10654)(136.677,1.12172)(137.01,1.13741)(137.343,1.14802)(137.677,1.15531)(138.01,1.15807)(138.343,1.15286)(138.677,1.14236)(139.01,1.12554)(139.343,1.10767)(139.677,1.08648)(140.01,1.0621)(140.343,1.03973)(140.677,1.01564)(141.01,0.993785)(141.343,0.976807)(141.677,0.964201)(142.01,0.955406)(142.343,0.948714)(142.677,0.947069)(143.01,0.947136)(143.343,0.949452)(143.677,0.956126)(144.01,0.962128)(144.343,0.968015)(144.677,0.972451)(145.01,0.976101)(145.343,0.979279)(145.677,0.980442)(146.01,0.981099)(146.343,0.978832)(146.677,0.974308)(147.01,0.968312)(147.343,0.961369)(147.677,0.955538)(148.01,0.947937)(148.343,0.940253)(148.677,0.933263)(149.01,0.925631)(149.343,0.919567)(149.677,0.915281)(150.01,0.912368)(150.343,0.908023)(150.677,0.903845)(151.01,0.900524)(151.343,0.895082)(151.677,0.890075)(152.01,0.884931)(152.343,0.877338)(152.677,0.867915)(153.01,0.858354)(153.343,0.849366)(153.677,0.839744)(154.01,0.831969)(154.343,0.82547)(154.677,0.81882)(155.01,0.814678)(155.343,0.813039)(155.677,0.814054)(156.01,0.816848)(156.343,0.821285)(156.677,0.826979)(157.01,0.832954)(157.343,0.840219)(157.677,0.84763)(158.01,0.854702)(158.343,0.859792)(158.677,0.861861)(159.01,0.862051)(159.343,0.858984)(159.677,0.852671)(160.01,0.843591)(160.343,0.83114)(160.677,0.814875)(161.01,0.796376)(161.343,0.778222)(161.677,0.759186)(162.01,0.740849)(162.343,0.725211)(162.677,0.711041)(163.01,0.699604)(163.343,0.693011)(163.677,0.691309)(164.01,0.692521)(164.343,0.697854)(164.677,0.707168)(165.01,0.717937)(165.343,0.731456)(165.677,0.747281)(166.01,0.76286)(166.343,0.776831)(166.677,0.789227)
};
\addplot [
draw=blue,thin
] coordinates{(0.01,0.0000222222)(0.0933333,0.00193539)(0.176667,0.00693046)(0.26,0.0149972)(0.343333,0.0261189)(0.426667,0.0402728)(0.51,0.0574299)(0.593333,0.0775547)(0.676667,0.100606)(0.76,0.126537)(0.843333,0.155293)(0.926667,0.186817)(1.01,0.221043)(1.09333,0.257902)(1.17667,0.297319)(1.26,0.339212)(1.34333,0.383496)(1.42667,0.430082)(1.51,0.478875)(1.59333,0.529776)(1.67667,0.582681)(1.76,0.637484)(1.84333,0.694075)(1.92667,0.752338)(2.01,0.812158)(2.09333,0.873415)(2.17667,0.935986)(2.26,0.999746)(2.34333,1.06457)(2.42667,1.13033)(2.51,1.19689)(2.59333,1.26412)(2.67667,1.3319)(2.76,1.40009)(2.84333,1.46855)(2.92667,1.53716)(3.01,1.60578)(3.09333,1.67428)(3.17667,1.74254)(3.26,1.81041)(3.34333,1.87778)(3.42667,1.94451)(3.51,2.01049)(3.59333,2.07559)(3.67667,2.13969)(3.76,2.20267)(3.84333,2.26443)(3.92667,2.32486)(4.01,2.38384)(4.09333,2.44127)(4.17667,2.49706)(4.26,2.55111)(4.34333,2.60334)(4.42667,2.65365)(4.51,2.70197)(4.59333,2.74822)(4.67667,2.79233)(4.76,2.83424)(4.84333,2.87388)(4.92667,2.9112)(5.01,2.94615)(5.09333,2.97869)(5.17667,3.00878)(5.26,3.03639)(5.34333,3.06149)(5.42667,3.08405)(5.51,3.10408)(5.59333,3.12155)(5.67667,3.13646)(5.76,3.14881)(5.84333,3.15862)(5.92667,3.16589)(6.01,3.17065)(6.09333,3.17292)(6.17667,3.17273)(6.26,3.17011)(6.34333,3.1651)(6.42667,3.15776)(6.51,3.14812)(6.59333,3.13625)(6.67667,3.1222)(6.76,3.10604)(6.84333,3.08783)(6.92667,3.06765)(7.01,3.04557)(7.09333,3.02167)(7.17667,2.99603)(7.26,2.96875)(7.34333,2.9399)(7.42667,2.90958)(7.51,2.87788)(7.59333,2.8449)(7.67667,2.81073)(7.76,2.77548)(7.84333,2.73924)(7.92667,2.70212)(8.01,2.66422)(8.09333,2.62564)(8.17667,2.58649)(8.26,2.54687)(8.34333,2.50688)(8.42667,2.46663)(8.51,2.42622)(8.59333,2.38575)(8.67667,2.34533)(8.76,2.30504)(8.84333,2.265)(8.92667,2.22529)(9.01,2.18602)(9.09333,2.14726)(9.17667,2.10911)(9.26,2.07167)(9.34333,2.035)(9.42667,1.9992)(9.51,1.96433)(9.59333,1.93048)(9.67667,1.89772)(9.76,1.8661)(9.84333,1.8357)(9.92667,1.80658)(10.01,1.77878)(10.0933,1.75237)(10.1767,1.72738)(10.26,1.70386)(10.3433,1.68185)(10.4267,1.66138)(10.51,1.64248)(10.5933,1.62519)(10.6767,1.6095)(10.76,1.59545)(10.8433,1.58305)(10.9267,1.57229)(11.01,1.5632)(11.0933,1.55575)(11.1767,1.54995)(11.26,1.54578)(11.3433,1.54324)(11.4267,1.54229)(11.51,1.54293)(11.5933,1.54511)(11.6767,1.54882)(11.76,1.55401)(11.8433,1.56065)(11.9267,1.5687)(12.01,1.57811)(12.0933,1.58884)(12.1767,1.60083)(12.26,1.61402)(12.3433,1.62837)(12.4267,1.64382)(12.51,1.66029)(12.5933,1.67774)(12.6767,1.69609)(12.76,1.71528)(12.8433,1.73524)(12.9267,1.75589)(13.01,1.77718)(13.0933,1.79902)(13.1767,1.82135)(13.26,1.84409)(13.3433,1.86717)(13.4267,1.89051)(13.51,1.91405)(13.5933,1.9377)(13.6767,1.9614)(13.76,1.98508)(13.8433,2.00867)(13.9267,2.03209)(14.01,2.05528)(14.0933,2.07817)(14.1767,2.10069)(14.26,2.12279)(14.3433,2.14439)(14.4267,2.16545)(14.51,2.18591)(14.5933,2.2057)(14.6767,2.22478)(14.76,2.24309)(14.8433,2.2606)(14.9267,2.27725)(15.01,2.29301)(15.0933,2.30784)(15.1767,2.3217)(15.26,2.33456)(15.3433,2.34639)(15.4267,2.35717)(15.51,2.36687)(15.5933,2.37548)(15.6767,2.38298)(15.76,2.38936)(15.8433,2.39461)(15.9267,2.39873)(16.01,2.40171)(16.0933,2.40355)(16.1767,2.40427)(16.26,2.40386)(16.3433,2.40235)(16.4267,2.39973)(16.51,2.39605)(16.5933,2.3913)(16.6767,2.38552)(16.76,2.37873)(16.8433,2.37096)(16.9267,2.36225)(17.01,2.35262)(17.0933,2.34212)(17.1767,2.33077)(17.26,2.31864)(17.3433,2.30574)(17.4267,2.29214)(17.51,2.27787)(17.5933,2.26299)(17.6767,2.24754)(17.76,2.23158)(17.8433,2.21515)(17.9267,2.19831)(18.01,2.18111)(18.0933,2.16361)(18.1767,2.14586)(18.26,2.12791)(18.3433,2.10981)(18.4267,2.09163)(18.51,2.07342)(18.5933,2.05523)(18.6767,2.03712)(18.76,2.01913)(18.8433,2.00132)(18.9267,1.98374)(19.01,1.96645)(19.0933,1.94948)(19.1767,1.93289)(19.26,1.91672)(19.3433,1.90102)(19.4267,1.88582)(19.51,1.87118)(19.5933,1.85713)(19.6767,1.8437)(19.76,1.83093)(19.8433,1.81885)(19.9267,1.80749)(20.01,1.79688)(20.0933,1.78705)(20.1767,1.77801)(20.26,1.76979)(20.3433,1.7624)(20.4267,1.75587)(20.51,1.75019)(20.5933,1.74539)(20.6767,1.74146)(20.76,1.73842)(20.8433,1.73626)(20.9267,1.73498)(21.01,1.73459)(21.0933,1.73506)(21.1767,1.7364)(21.26,1.7386)(21.3433,1.74163)(21.4267,1.74549)(21.51,1.75015)(21.5933,1.75559)(21.6767,1.7618)(21.76,1.76875)(21.8433,1.7764)(21.9267,1.78473)(22.01,1.79372)(22.0933,1.80332)(22.1767,1.81351)(22.26,1.82424)(22.3433,1.83548)(22.4267,1.8472)(22.51,1.85935)(22.5933,1.8719)(22.6767,1.8848)(22.76,1.898)(22.8433,1.91148)(22.9267,1.92518)(23.01,1.93906)(23.0933,1.95308)(23.1767,1.96719)(23.26,1.98136)(23.3433,1.99553)(23.4267,2.00966)(23.51,2.02372)(23.5933,2.03765)(23.6767,2.05142)(23.76,2.06499)(23.8433,2.07831)(23.9267,2.09135)(24.01,2.10407)(24.0933,2.11643)(24.1767,2.1284)(24.26,2.13994)(24.3433,2.15103)(24.4267,2.16162)(24.51,2.17169)(24.5933,2.18122)(24.6767,2.19018)(24.76,2.19854)(24.8433,2.20628)(24.9267,2.21339)(25.01,2.21984)(25.0933,2.22562)(25.1767,2.23072)(25.26,2.23513)(25.3433,2.23883)(25.4267,2.24182)(25.51,2.2441)(25.5933,2.24565)(25.6767,2.24649)(25.76,2.24662)(25.8433,2.24603)(25.9267,2.24473)(26.01,2.24274)(26.0933,2.24006)(26.1767,2.2367)(26.26,2.23268)(26.3433,2.22801)(26.4267,2.22272)(26.51,2.21682)(26.5933,2.21034)(26.6767,2.20329)(26.76,2.1957)(26.8433,2.1876)(26.9267,2.17901)(27.01,2.16997)(27.0933,2.1605)(27.1767,2.15064)(27.26,2.14042)(27.3433,2.12986)(27.4267,2.11901)(27.51,2.1079)(27.5933,2.09656)(27.6767,2.08503)(27.76,2.07335)(27.8433,2.06155)(27.9267,2.04966)(28.01,2.03773)(28.0933,2.02579)(28.1767,2.01388)(28.26,2.00203)(28.3433,1.99028)(28.4267,1.97866)(28.51,1.96721)(28.5933,1.95596)(28.6767,1.94494)(28.76,1.93419)(28.8433,1.92374)(28.9267,1.91361)(29.01,1.90384)(29.0933,1.89445)(29.1767,1.88548)(29.26,1.87693)(29.3433,1.86885)(29.4267,1.86124)(29.51,1.85413)(29.5933,1.84754)(29.6767,1.84149)(29.76,1.83599)(29.8433,1.83105)(29.9267,1.82669)(30.01,1.82291)
};
\addplot [
draw=blue,thin,dashed
] coordinates{(0.01,0.0000444443)(0.0933333,0.00385577)(0.176667,0.0135716)(0.26,0.0287064)(0.343333,0.0488399)(0.426667,0.0736477)(0.51,0.102871)(0.593333,0.136293)(0.676667,0.173723)(0.76,0.21499)(0.843333,0.259933)(0.926667,0.308396)(1.01,0.36023)(1.09333,0.415287)(1.17667,0.47342)(1.26,0.534481)(1.34333,0.598322)(1.42667,0.664794)(1.51,0.733745)(1.59333,0.805023)(1.67667,0.878474)(1.76,0.95394)(1.84333,1.03126)(1.92667,1.11029)(2.01,1.19085)(2.09333,1.27278)(2.17667,1.35592)(2.26,1.44011)(2.34333,1.52518)(2.42667,1.61096)(2.51,1.6973)(2.59333,1.78401)(2.67667,1.87095)(2.76,1.95794)(2.84333,2.04482)(2.92667,2.13144)(3.01,2.21762)(3.09333,2.30322)(3.17667,2.38807)(3.26,2.47202)(3.34333,2.55493)(3.42667,2.63664)(3.51,2.71702)(3.59333,2.79592)(3.67667,2.8732)(3.76,2.94874)(3.84333,3.02241)(3.92667,3.09408)(4.01,3.16364)(4.09333,3.23097)(4.17667,3.29598)(4.26,3.35855)(4.34333,3.4186)(4.42667,3.47602)(4.51,3.53075)(4.59333,3.5827)(4.67667,3.6318)(4.76,3.67799)(4.84333,3.72121)(4.92667,3.76141)(5.01,3.79855)(5.09333,3.83258)(5.17667,3.86348)(5.26,3.89122)(5.34333,3.91579)(5.42667,3.93718)(5.51,3.95538)(5.59333,3.97039)(5.67667,3.98224)(5.76,3.99093)(5.84333,3.99649)(5.92667,3.99894)(6.01,3.99833)(6.09333,3.9947)(6.17667,3.9881)(6.26,3.97857)(6.34333,3.96619)(6.42667,3.95102)(6.51,3.93312)(6.59333,3.91258)(6.67667,3.88947)(6.76,3.86389)(6.84333,3.83592)(6.92667,3.80566)(7.01,3.77321)(7.09333,3.73867)(7.17667,3.70215)(7.26,3.66375)(7.34333,3.62359)(7.42667,3.58179)(7.51,3.53845)(7.59333,3.49371)(7.67667,3.44767)(7.76,3.40047)(7.84333,3.35223)(7.92667,3.30307)(8.01,3.25311)(8.09333,3.20249)(8.17667,3.15133)(8.26,3.09974)(8.34333,3.04787)(8.42667,2.99582)(8.51,2.94373)(8.59333,2.89172)(8.67667,2.83989)(8.76,2.78837)(8.84333,2.73728)(8.92667,2.68673)(9.01,2.63682)(9.09333,2.58766)(9.17667,2.53935)(9.26,2.49201)(9.34333,2.44571)(9.42667,2.40056)(9.51,2.35663)(9.59333,2.31403)(9.67667,2.27281)(9.76,2.23307)(9.84333,2.19486)(9.92667,2.15827)(10.01,2.12333)(10.0933,2.09012)(10.1767,2.05868)(10.26,2.02906)(10.3433,2.0013)(10.4267,1.97543)(10.51,1.95149)(10.5933,1.92949)(10.6767,1.90946)(10.76,1.8914)(10.8433,1.87534)(10.9267,1.86127)(11.01,1.84919)(11.0933,1.83909)(11.1767,1.83096)(11.26,1.82479)(11.3433,1.82055)(11.4267,1.81821)(11.51,1.81775)(11.5933,1.81913)(11.6767,1.82231)(11.76,1.82725)(11.8433,1.83389)(11.9267,1.84218)(12.01,1.85208)(12.0933,1.86351)(12.1767,1.87643)(12.26,1.89075)(12.3433,1.90642)(12.4267,1.92335)(12.51,1.94149)(12.5933,1.96075)(12.6767,1.98105)(12.76,2.00232)(12.8433,2.02447)(12.9267,2.04742)(13.01,2.07109)(13.0933,2.09539)(13.1767,2.12024)(13.26,2.14555)(13.3433,2.17124)(13.4267,2.19722)(13.51,2.2234)(13.5933,2.24971)(13.6767,2.27606)(13.76,2.30236)(13.8433,2.32853)(13.9267,2.3545)(14.01,2.38018)(14.0933,2.4055)(14.1767,2.43038)(14.26,2.45475)(14.3433,2.47854)(14.4267,2.50169)(14.51,2.52412)(14.5933,2.54578)(14.6767,2.56661)(14.76,2.58655)(14.8433,2.60554)(14.9267,2.62355)(15.01,2.64051)(15.0933,2.6564)(15.1767,2.67117)(15.26,2.68479)(15.3433,2.69722)(15.4267,2.70844)(15.51,2.71842)(15.5933,2.72715)(15.6767,2.73461)(15.76,2.74078)(15.8433,2.74566)(15.9267,2.74925)(16.01,2.75154)(16.0933,2.75253)(16.1767,2.75225)(16.26,2.75068)(16.3433,2.74786)(16.4267,2.74379)(16.51,2.7385)(16.5933,2.73202)(16.6767,2.72436)(16.76,2.71557)(16.8433,2.70568)(16.9267,2.69472)(17.01,2.68273)(17.0933,2.66977)(17.1767,2.65586)(17.26,2.64106)(17.3433,2.62542)(17.4267,2.60899)(17.51,2.59183)(17.5933,2.57398)(17.6767,2.55551)(17.76,2.53647)(17.8433,2.51692)(17.9267,2.49692)(18.01,2.47654)(18.0933,2.45583)(18.1767,2.43485)(18.26,2.41367)(18.3433,2.39235)(18.4267,2.37095)(18.51,2.34954)(18.5933,2.32816)(18.6767,2.30689)(18.76,2.28578)(18.8433,2.2649)(18.9267,2.24429)(19.01,2.22402)(19.0933,2.20414)(19.1767,2.1847)(19.26,2.16575)(19.3433,2.14734)(19.4267,2.12953)(19.51,2.11235)(19.5933,2.09585)(19.6767,2.08007)(19.76,2.06505)(19.8433,2.05082)(19.9267,2.03741)(20.01,2.02486)(20.0933,2.01319)(20.1767,2.00243)(20.26,1.9926)(20.3433,1.98372)(20.4267,1.9758)(20.51,1.96886)(20.5933,1.9629)(20.6767,1.95794)(20.76,1.95398)(20.8433,1.95102)(20.9267,1.94905)(21.01,1.94807)(21.0933,1.94808)(21.1767,1.94906)(21.26,1.951)(21.3433,1.95389)(21.4267,1.9577)(21.51,1.96241)(21.5933,1.96801)(21.6767,1.97445)(21.76,1.98172)(21.8433,1.98979)(21.9267,1.99861)(22.01,2.00817)(22.0933,2.01841)(22.1767,2.0293)(22.26,2.0408)(22.3433,2.05287)(22.4267,2.06547)(22.51,2.07855)(22.5933,2.09206)(22.6767,2.10596)(22.76,2.12021)(22.8433,2.13475)(22.9267,2.14954)(23.01,2.16453)(23.0933,2.17966)(23.1767,2.1949)(23.26,2.21019)(23.3433,2.22548)(23.4267,2.24073)(23.51,2.25588)(23.5933,2.2709)(23.6767,2.28573)(23.76,2.30032)(23.8433,2.31464)(23.9267,2.32865)(24.01,2.34229)(24.0933,2.35553)(24.1767,2.36833)(24.26,2.38065)(24.3433,2.39246)(24.4267,2.40372)(24.51,2.4144)(24.5933,2.42448)(24.6767,2.43391)(24.76,2.44269)(24.8433,2.45077)(24.9267,2.45815)(25.01,2.4648)(25.0933,2.47071)(25.1767,2.47586)(25.26,2.48024)(25.3433,2.48384)(25.4267,2.48665)(25.51,2.48867)(25.5933,2.4899)(25.6767,2.49033)(25.76,2.48997)(25.8433,2.48882)(25.9267,2.4869)(26.01,2.4842)(26.0933,2.48074)(26.1767,2.47654)(26.26,2.47162)(26.3433,2.46598)(26.4267,2.45965)(26.51,2.45265)(26.5933,2.44501)(26.6767,2.43675)(26.76,2.4279)(26.8433,2.41849)(26.9267,2.40855)(27.01,2.39811)(27.0933,2.38721)(27.1767,2.37588)(27.26,2.36415)(27.3433,2.35206)(27.4267,2.33966)(27.51,2.32697)(27.5933,2.31403)(27.6767,2.3009)(27.76,2.2876)(27.8433,2.27417)(27.9267,2.26067)(28.01,2.24711)(28.0933,2.23356)(28.1767,2.22004)(28.26,2.20659)(28.3433,2.19326)(28.4267,2.18009)(28.51,2.1671)(28.5933,2.15434)(28.6767,2.14184)(28.76,2.12965)(28.8433,2.11778)(28.9267,2.10628)(29.01,2.09518)(29.0933,2.08451)(29.1767,2.07428)(29.26,2.06455)(29.3433,2.05532)(29.4267,2.04662)(29.51,2.03847)(29.5933,2.0309)(29.6767,2.02392)(29.76,2.01755)(29.8433,2.0118)(29.9267,2.00669)(30.01,2.00223)
};
\addplot [
draw=red,ultra thick
] coordinates{(0.01,0.00000000239483)(0.0933333,0.0000180477)(0.176667,0.00022881)(0.26,0.00105199)(0.343333,0.00311146)(0.426667,0.00716484)(0.51,0.0140193)(0.593333,0.0244438)(0.676667,0.0390886)(0.76,0.058422)(0.843333,0.0826904)(0.926667,0.111908)(1.01,0.145874)(1.09333,0.184216)(1.17667,0.226454)(1.26,0.272069)(1.34333,0.320579)(1.42667,0.371602)(1.51,0.4249)(1.59333,0.480405)(1.67667,0.53821)(1.76,0.598548)(1.84333,0.661738)(1.92667,0.728128)(2.01,0.798026)(2.09333,0.871641)(2.17667,0.949033)(2.26,1.03009)(2.34333,1.1145)(2.42667,1.20181)(2.51,1.29141)(2.59333,1.38264)(2.67667,1.47483)(2.76,1.56733)(2.84333,1.65961)(2.92667,1.7513)(3.01,1.84217)(3.09333,1.93214)(3.17667,2.02129)(3.26,2.10977)(3.34333,2.19778)(3.42667,2.28553)(3.51,2.37313)(3.59333,2.46061)(3.67667,2.54788)(3.76,2.6347)(3.84333,2.72073)(3.92667,2.80555)(4.01,2.88865)(4.09333,2.96959)(4.17667,3.04791)(4.26,3.12328)(4.34333,3.19547)(4.42667,3.26438)(4.51,3.33001)(4.59333,3.39247)(4.67667,3.45192)(4.76,3.50857)(4.84333,3.56257)(4.92667,3.61405)(5.01,3.66304)(5.09333,3.70948)(5.17667,3.75322)(5.26,3.79404)(5.34333,3.83166)(5.42667,3.8658)(5.51,3.89619)(5.59333,3.92264)(5.67667,3.94502)(5.76,3.9633)(5.84333,3.97757)(5.92667,3.98796)(6.01,3.99472)(6.09333,3.99808)(6.17667,3.99831)(6.26,3.99561)(6.34333,3.99017)(6.42667,3.98207)(6.51,3.97132)(6.59333,3.95788)(6.67667,3.94167)(6.76,3.92255)(6.84333,3.90043)(6.92667,3.87523)(7.01,3.84694)(7.09333,3.81563)(7.17667,3.78144)(7.26,3.74459)(7.34333,3.70534)(7.42667,3.66399)(7.51,3.62086)(7.59333,3.57624)(7.67667,3.53036)(7.76,3.4834)(7.84333,3.43549)(7.92667,3.38668)(8.01,3.33697)(8.09333,3.28634)(8.17667,3.23478)(8.26,3.18226)(8.34333,3.12882)(8.42667,3.07456)(8.51,3.01962)(8.59333,2.96422)(8.67667,2.90861)(8.76,2.85307)(8.84333,2.7979)(8.92667,2.74337)(9.01,2.68971)(9.09333,2.6371)(9.17667,2.58566)(9.26,2.53546)(9.34333,2.4865)(9.42667,2.43875)(9.51,2.39219)(9.59333,2.34678)(9.67667,2.30251)(9.76,2.25942)(9.84333,2.2176)(9.92667,2.17718)(10.01,2.13832)(10.0933,2.10123)(10.1767,2.0661)(10.26,2.03313)(10.3433,2.00246)(10.4267,1.97421)(10.51,1.94841)(10.5933,1.92507)(10.6767,1.90413)(10.76,1.8855)(10.8433,1.86907)(10.9267,1.85474)(11.01,1.8424)(11.0933,1.832)(11.1767,1.82351)(11.26,1.81697)(11.3433,1.81243)(11.4267,1.80996)(11.51,1.80966)(11.5933,1.8116)(11.6767,1.81585)(11.76,1.82241)(11.8433,1.83124)(11.9267,1.84228)(12.01,1.85538)(12.0933,1.87039)(12.1767,1.88713)(12.26,1.90541)(12.3433,1.92507)(12.4267,1.94596)(12.51,1.96799)(12.5933,1.99109)(12.6767,2.01524)(12.76,2.04043)(12.8433,2.06667)(12.9267,2.09398)(13.01,2.12235)(13.0933,2.15174)(13.1767,2.18208)(13.26,2.21325)(13.3433,2.2451)(13.4267,2.27744)(13.51,2.31008)(13.5933,2.34281)(13.6767,2.37546)(13.76,2.40785)(13.8433,2.43988)(13.9267,2.47145)(14.01,2.50253)(14.0933,2.5331)(14.1767,2.56316)(14.26,2.59273)(14.3433,2.6218)(14.4267,2.65036)(14.51,2.67835)(14.5933,2.70568)(14.6767,2.73225)(14.76,2.75792)(14.8433,2.78252)(14.9267,2.80592)(15.01,2.82796)(15.0933,2.84854)(15.1767,2.86758)(15.26,2.88503)(15.3433,2.90091)(15.4267,2.91524)(15.51,2.92807)(15.5933,2.93946)(15.6767,2.94948)(15.76,2.95817)(15.8433,2.96555)(15.9267,2.9716)(16.01,2.97629)(16.0933,2.97956)(16.1767,2.98134)(16.26,2.98155)(16.3433,2.98013)(16.4267,2.97703)(16.51,2.97226)(16.5933,2.96583)(16.6767,2.9578)(16.76,2.94828)(16.8433,2.93737)(16.9267,2.9252)(17.01,2.91189)(17.0933,2.89754)(17.1767,2.88225)(17.26,2.86607)(17.3433,2.84904)(17.4267,2.83115)(17.51,2.81241)(17.5933,2.7928)(17.6767,2.77231)(17.76,2.75096)(17.8433,2.72877)(17.9267,2.70581)(18.01,2.68218)(18.0933,2.658)(18.1767,2.63342)(18.26,2.60857)(18.3433,2.5836)(18.4267,2.55866)(18.51,2.53384)(18.5933,2.50924)(18.6767,2.4849)(18.76,2.46085)(18.8433,2.4371)(18.9267,2.41365)(19.01,2.39048)(19.0933,2.36762)(19.1767,2.34509)(19.26,2.32294)(19.3433,2.30126)(19.4267,2.28013)(19.51,2.25967)(19.5933,2.24001)(19.6767,2.22127)(19.76,2.20355)(19.8433,2.18693)(19.9267,2.17147)(20.01,2.15718)(20.0933,2.14408)(20.1767,2.13212)(20.26,2.12127)(20.3433,2.11149)(20.4267,2.10275)(20.51,2.09501)(20.5933,2.08829)(20.6767,2.0826)(20.76,2.07798)(20.8433,2.07449)(20.9267,2.0722)(21.01,2.07116)(21.0933,2.07142)(21.1767,2.07301)(21.26,2.07593)(21.3433,2.08015)(21.4267,2.08562)(21.51,2.09227)(21.5933,2.09999)(21.6767,2.10871)(21.76,2.11833)(21.8433,2.12877)(21.9267,2.13998)(22.01,2.15193)(22.0933,2.16459)(22.1767,2.17797)(22.26,2.19209)(22.3433,2.20695)(22.4267,2.22256)(22.51,2.23891)(22.5933,2.25595)(22.6767,2.27362)(22.76,2.29185)(22.8433,2.31053)(22.9267,2.32955)(23.01,2.34878)(23.0933,2.36811)(23.1767,2.38744)(23.26,2.4067)(23.3433,2.42583)(23.4267,2.44478)(23.51,2.46355)(23.5933,2.48214)(23.6767,2.50054)(23.76,2.51875)(23.8433,2.53677)(23.9267,2.55454)(24.01,2.57203)(24.0933,2.58917)(24.1767,2.60585)(24.26,2.62199)(24.3433,2.63749)(24.4267,2.65224)(24.51,2.66617)(24.5933,2.67922)(24.6767,2.69135)(24.76,2.70256)(24.8433,2.71284)(24.9267,2.72222)(25.01,2.73074)(25.0933,2.73842)(25.1767,2.74528)(25.26,2.75133)(25.3433,2.75657)(25.4267,2.76097)(25.51,2.76447)(25.5933,2.76704)(25.6767,2.76862)(25.76,2.76915)(25.8433,2.76861)(25.9267,2.76697)(26.01,2.76425)(26.0933,2.76046)(26.1767,2.75567)(26.26,2.74994)(26.3433,2.74333)(26.4267,2.73593)(26.51,2.72781)(26.5933,2.71901)(26.6767,2.70957)(26.76,2.69952)(26.8433,2.68886)(26.9267,2.67757)(27.01,2.66566)(27.0933,2.65311)(27.1767,2.63991)(27.26,2.6261)(27.3433,2.61171)(27.4267,2.59679)(27.51,2.58142)(27.5933,2.5657)(27.6767,2.54971)(27.76,2.53356)(27.8433,2.51733)(27.9267,2.50111)(28.01,2.48495)(28.0933,2.46889)(28.1767,2.45295)(28.26,2.43714)(28.3433,2.42147)(28.4267,2.40592)(28.51,2.39051)(28.5933,2.37526)(28.6767,2.3602)(28.76,2.34538)(28.8433,2.33087)(28.9267,2.31675)(29.01,2.30311)(29.0933,2.29003)(29.1767,2.27758)(29.26,2.26583)(29.3433,2.25482)(29.4267,2.24458)(29.51,2.23511)(29.5933,2.22639)(29.6767,2.21841)(29.76,2.21114)(29.8433,2.20456)(29.9267,2.19865)(30.01,2.19343)
 };

\legend{Truncated Mie series,Van de Hulst approximation,Evans--Fournier approximation,Semigroup asymptotics};

\node[fill=white,anchor=south,draw=black] at
(axis cs:27.5,.25) {\tiny$ n=4/3$};
\end{axis}
\end{tikzpicture}
\vspace{-.75em}
\begin{center}\begin{tiny}\quad\,(b)\end{tiny}\end{center}\end{minipage}\begin{minipage}{.25\textwidth}\caption{Comparison of various  approximate solutions to the total scattering cross-section of    \emph{(a)}~a glass bead $n=3/2 $  and \emph{(b)}~a water droplet $n=4/3 $.  Following  \cite[p.~177, Fig.~32]{vandeHulst}, the horizontal axes are lined up so that the values of relative phase shift $2(n-1)kR $ match in both panels. For numerical computations, the Mie series \eqref{eq:Mie_sc_series} are truncated after the first  100 terms. \label{fig:5-2}}\end{minipage}\end{figure}
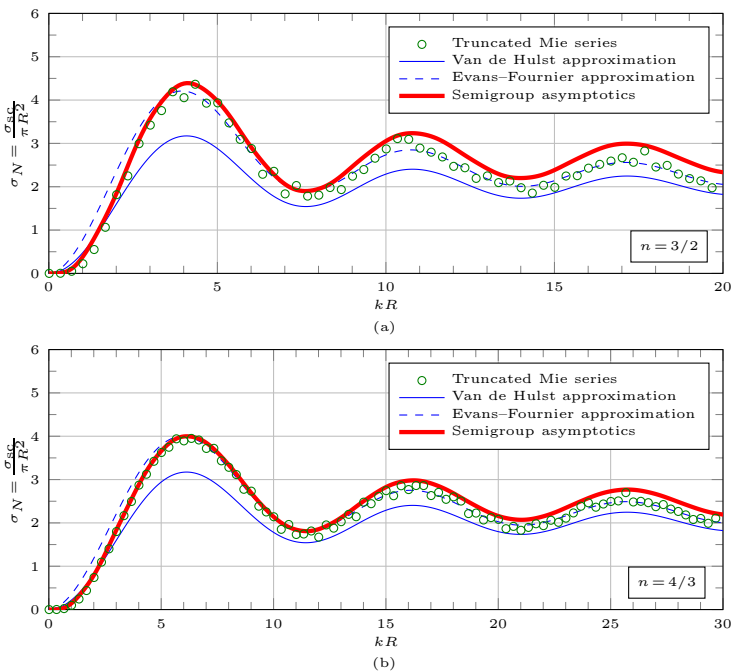
 Admittedly, our semigroup approach is not the first reported attempt to amend the deficiency of Born approximation in the non-perturbative regime. In \cite[\S11.22]{vandeHulst}, H. C. van de Hulst proposed an asymptotic formula \begin{align}\frac{\sigma_\text{sc}}{\pi R^2}\sim 2-\frac{4}{\rho}\sin\rho+\frac{4}{\rho^2}(1-\cos\rho),\quad \rho=2(n-1)kR\end{align}based on scalar wave approximations of anomalous diffraction. As we can check in the non-perturbative refractive indices for glass ($n=3/2$) and water ($n=4/3$) (see Figure~\ref{fig:5-2}), the van de Hulst approximation gives a less accurate approximation than the semigroup approach in the regime of  $0\leq2(n-1)kR\leq 20 $. A possible explanation is that the transversality constraint is honored in the semigroup asymptotic analysis,\footnote{We note that transversality is essential to our semigroup analysis in at least two ways. In theory, the strong stability $ \lim_{\tau\to+\infty}\Vert\exp(-i\tau\hat{\mathscr G})\bm F\Vert_{L^2(V;\mathbb C^3)}=0,\forall \bm F\in\Phi(V;\mathbb C^3)$ has to draw on the discreteness of the physical spectrum $ \sigma^{\Phi}(\hat{\mathscr G})$, while there is no warrant for discreteness in the non-physical spectrum  $ \sigma(\hat{\mathscr G})$ for $ \hat{\mathscr G}:L^2(V;\mathbb C^3)\longrightarrow L^2(V;\mathbb C^3)$ without the transversality constraint $ \nabla\cdot\bm F=\mathbf 0$ \cite{QualEM}. In practice, our approximate formula \eqref{eq:bulkapprox} respects the transversality condition and the long-term behavior of $\exp(-i\tau\hat{\mathscr G})$.} but is absent from the scalar wave approximation \cite{vandeHulst}. As we formerly proved in \cite[\S2]{QualEM}, ignoring the transversality constraint may cause non-robust solutions to the light scattering problem.

While our heuristics in \S\ref{subsec:nonpert_approx} did not, \textit{per se}, accommodate to very large values of $kR$, its output in \S\ref{subsec:nonpert_Mie}  performed  (as seen in Figure~\ref{fig:5-2}) only slightly worse  than the Evans--Fournier approximation \cite{Evans1990} in the  large $ k R$ regime. Here, the Evans--Fournier approximation for $ 1.01\leq n\leq 2.00$ is the van de Hulst approximation times an empirical factor of $ 2-\exp(-(kR)^{-2/3})$, which has been  designed to fit the asymptotic behavior of Mie scattering cross-section  as $ kR\to+\infty$.

Undoubtedly, the numerical accuracy of  our semigroup-based approximation does deteriorate for  either  large $  kR$   (Figure~\ref{fig:5-2})  or  large $n-1$ (Figure~\ref{fig:5-3}). This probably represents an  intrinsic limitation in  our approximate treatment of the Schr\"odinger semigroup associated with the light scattering problem, and the resulting ``quasistatic response'' approximation \eqref{eq:quasistat_resp} at the dielectric boundary. Perhaps a better understanding of the boundary corrections to the free-space propagator \eqref{eq:freeprop} is required to improve the accuracy of the semigroup approach for  large values of phase shifts $2(n-1)kR $.

We also note that A. Ya.\ Perelman \cite{Perelman1991} has derived non-perturbative formulae for Mie scattering that took  a functional form similar to \eqref{eq:nonperturb} and exhibited similar level of approximation accuracy. However, Perelman's approach was based on an asymptotic summation of the Mie series that draws heavily on some identities involving special functions. We take leave to think that the semigroup approach provides a clearer physical picture and is adaptable to more complicated scattering geometry where  analytic solutions (analogs of Mie series) are unavailable.

 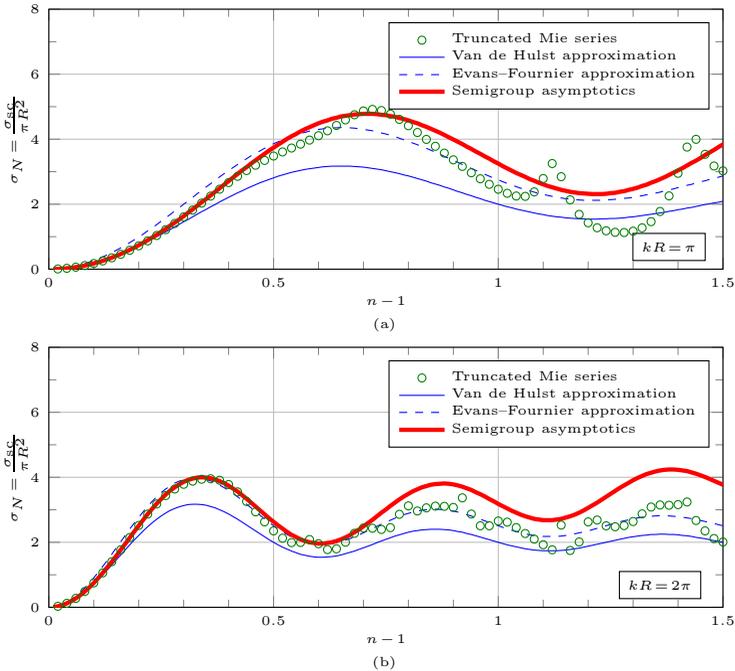
\begin{figure}[!t]
\begin{minipage}{.75\textwidth}\begin{tikzpicture}\pgfplotsset{xlabel style={yshift=.3cm}, ylabel style={yshift=-.8cm},width=10.5cm,height=5cm,xmajorgrids,ymajorgrids, tick label style={font=\tiny},label style={font=\tiny}}
\begin{axis}[xlabel={$n-1$},ylabel={$ \sigma_N=\frac{\sigma_{\text{sc}}}{\pi R^2}$},ymin=0,ymax=8,xmin=0,xmax=1.5,enlargelimits=false,  minor y tick num =1, minor x tick num =4,ytick={0,2,...,10},xtick={0,.5,...,2},legend style={
legend columns=1,
cells={anchor=west},at={(.98,.95)},
font=\tiny,legend style={row sep=-2.5pt},
},
   restrict y to domain=0:10]\addplot [only marks,
draw=green!50!black,mark=o,thin,mark size=1.5pt
] coordinates{(0.0200,0.00660654)(0.0400,0.0269357)(0.0600,0.061622)(0.0800,0.111097)(0.100,0.175574)(0.120,0.255069)(0.140,0.349445)(0.160,0.458489)(0.180,0.581972)(0.200,0.719687)(0.220,0.871438)(0.240,1.03695)(0.260,1.2157)(0.280,1.40676)(0.300,1.60848)(0.320,1.81841)(0.340,2.03322)(0.360,2.24882)(0.380,2.46078)(0.400,2.6648)(0.420,2.85724)(0.440,3.03558)(0.460,3.19871)(0.480,3.34698)(0.500,3.48224)(0.520,3.60763)(0.540,3.72747)(0.560,3.84698)(0.580,3.97187)(0.600,4.10767)(0.620,4.25822)(0.640,4.42314)(0.660,4.59417)(0.680,4.75217)(0.700,4.86878)(0.720,4.91636)(0.740,4.88176)(0.760,4.77234)(0.780,4.60908)(0.800,4.41472)(0.820,4.20626)(0.840,3.99361)(0.860,3.78141)(0.880,3.57144)(0.900,3.36456)(0.920,3.16211)(0.940,2.9665)(0.960,2.78131)(0.980,2.61078)(1.00,2.45975)(1.02,2.33525)(1.04,2.25137)(1.06,2.24061)(1.08,2.37804)(1.10,2.78814)(1.12,3.25073)(1.14,2.84582)(1.16,2.13413)(1.18,1.68569)(1.20,1.42895)(1.22,1.2749)(1.24,1.18212)(1.26,1.13416)(1.28,1.12763)(1.30,1.16839)(1.32,1.27145)(1.34,1.46184)(1.36,1.7748)(1.38,2.25526)(1.40,2.9494)(1.42,3.75725)(1.44,3.99585)(1.46,3.53321)(1.48,3.17509)(1.50,3.02645)
};
\addplot [
draw=blue,thin
] coordinates{(0.010,0.00197)(0.0200,0.00789)(0.0300,0.0177)(0.0400,0.0315)(0.0500,0.0491)(0.0600,0.0705)(0.0700,0.0957)(0.0800,0.125)(0.0900,0.157)(0.100,0.193)(0.110,0.233)(0.120,0.275)(0.130,0.321)(0.140,0.371)(0.150,0.423)(0.160,0.478)(0.170,0.535)(0.180,0.596)(0.190,0.658)(0.200,0.723)(0.210,0.790)(0.220,0.859)(0.230,0.929)(0.240,1.00)(0.250,1.07)(0.260,1.15)(0.270,1.22)(0.280,1.30)(0.290,1.38)(0.300,1.46)(0.310,1.53)(0.320,1.61)(0.330,1.69)(0.340,1.77)(0.350,1.84)(0.360,1.92)(0.370,1.99)(0.380,2.07)(0.390,2.14)(0.400,2.21)(0.410,2.28)(0.420,2.35)(0.430,2.41)(0.440,2.48)(0.450,2.54)(0.460,2.60)(0.470,2.66)(0.480,2.71)(0.490,2.76)(0.500,2.81)(0.510,2.86)(0.520,2.90)(0.530,2.94)(0.540,2.98)(0.550,3.01)(0.560,3.04)(0.570,3.07)(0.580,3.09)(0.590,3.11)(0.600,3.13)(0.610,3.15)(0.620,3.16)(0.630,3.17)(0.640,3.17)(0.650,3.17)(0.660,3.17)(0.670,3.17)(0.680,3.16)(0.690,3.15)(0.700,3.14)(0.710,3.12)(0.720,3.10)(0.730,3.08)(0.740,3.06)(0.750,3.03)(0.760,3.00)(0.770,2.97)(0.780,2.94)(0.790,2.90)(0.800,2.87)(0.810,2.83)(0.820,2.79)(0.830,2.75)(0.840,2.71)(0.850,2.66)(0.860,2.62)(0.870,2.58)(0.880,2.53)(0.890,2.49)(0.900,2.44)(0.910,2.39)(0.920,2.35)(0.930,2.30)(0.940,2.26)(0.950,2.21)(0.960,2.17)(0.970,2.12)(0.980,2.08)(0.990,2.04)(1.00,2.00)(1.01,1.96)(1.02,1.92)(1.03,1.89)(1.04,1.85)(1.05,1.82)(1.06,1.79)(1.07,1.76)(1.08,1.73)(1.09,1.70)(1.10,1.68)(1.11,1.65)(1.12,1.63)(1.13,1.61)(1.14,1.60)(1.15,1.58)(1.16,1.57)(1.17,1.56)(1.18,1.55)(1.19,1.55)(1.20,1.54)(1.21,1.54)(1.22,1.54)(1.23,1.55)(1.24,1.55)(1.25,1.56)(1.26,1.56)(1.27,1.57)(1.28,1.58)(1.29,1.60)(1.30,1.61)(1.31,1.63)(1.32,1.65)(1.33,1.67)(1.34,1.69)(1.35,1.71)(1.36,1.73)(1.37,1.75)(1.38,1.78)(1.39,1.80)(1.40,1.83)(1.41,1.85)(1.42,1.88)(1.43,1.90)(1.44,1.93)(1.45,1.96)(1.46,1.99)(1.47,2.01)(1.48,2.04)(1.49,2.06)(1.50,2.09)(1.51,2.12)(1.52,2.14)(1.53,2.16)(1.54,2.19)(1.55,2.21)(1.56,2.23)(1.57,2.25)(1.58,2.27)(1.59,2.29)(1.60,2.31)(1.61,2.32)(1.62,2.34)(1.63,2.35)(1.64,2.36)(1.65,2.37)(1.66,2.38)(1.67,2.39)(1.68,2.39)(1.69,2.40)(1.70,2.40)(1.71,2.40)(1.72,2.40)(1.73,2.40)(1.74,2.40)(1.75,2.40)(1.76,2.39)(1.77,2.39)(1.78,2.38)(1.79,2.37)(1.80,2.36)(1.81,2.35)(1.82,2.33)(1.83,2.32)(1.84,2.31)(1.85,2.29)(1.86,2.27)(1.87,2.26)(1.88,2.24)(1.89,2.22)(1.90,2.20)(1.91,2.18)(1.92,2.16)(1.93,2.14)(1.94,2.12)(1.95,2.10)(1.96,2.08)(1.97,2.06)(1.98,2.04)(1.99,2.02)(2.00,2.00)(2.01,1.98)(2.02,1.96)(2.03,1.94)(2.04,1.92)(2.05,1.91)(2.06,1.89)(2.07,1.87)(2.08,1.86)(2.09,1.84)(2.10,1.83)(2.11,1.81)(2.12,1.80)(2.13,1.79)(2.14,1.78)(2.15,1.77)(2.16,1.76)(2.17,1.75)(2.18,1.75)(2.19,1.74)(2.20,1.74)(2.21,1.74)(2.22,1.74)(2.23,1.73)(2.24,1.74)(2.25,1.74)(2.26,1.74)(2.27,1.74)(2.28,1.75)(2.29,1.75)(2.30,1.76)(2.31,1.77)(2.32,1.78)(2.33,1.79)(2.34,1.80)(2.35,1.81)(2.36,1.82)(2.37,1.83)(2.38,1.85)(2.39,1.86)(2.40,1.88)(2.41,1.89)(2.42,1.91)(2.43,1.92)(2.44,1.94)(2.45,1.95)(2.46,1.97)(2.47,1.98)(2.48,2.00)(2.49,2.02)(2.50,2.03)(2.51,2.05)(2.52,2.06)(2.53,2.08)(2.54,2.09)(2.55,2.11)(2.56,2.12)(2.57,2.13)(2.58,2.15)(2.59,2.16)(2.60,2.17)(2.61,2.18)(2.62,2.19)(2.63,2.20)(2.64,2.21)(2.65,2.22)(2.66,2.22)(2.67,2.23)(2.68,2.24)(2.69,2.24)(2.70,2.24)(2.71,2.24)(2.72,2.25)(2.73,2.25)(2.74,2.25)(2.75,2.24)(2.76,2.24)(2.77,2.24)(2.78,2.24)(2.79,2.23)(2.80,2.23)(2.81,2.22)(2.82,2.21)(2.83,2.20)(2.84,2.20)(2.85,2.19)(2.86,2.18)(2.87,2.17)(2.88,2.15)(2.89,2.14)(2.90,2.13)(2.91,2.12)(2.92,2.11)(2.93,2.09)(2.94,2.08)(2.95,2.07)(2.96,2.05)(2.97,2.04)(2.98,2.03)(2.99,2.01)(3.00,2.00)
};\addplot [
draw=blue,thin,dashed
] coordinates{(0.010,0.00271)(0.0200,0.0108)(0.0300,0.0243)(0.0400,0.0432)(0.0500,0.0674)(0.0600,0.0968)(0.0700,0.131)(0.0800,0.171)(0.0900,0.216)(0.100,0.265)(0.110,0.319)(0.120,0.378)(0.130,0.441)(0.140,0.509)(0.150,0.580)(0.160,0.656)(0.170,0.735)(0.180,0.817)(0.190,0.903)(0.200,0.992)(0.210,1.08)(0.220,1.18)(0.230,1.28)(0.240,1.37)(0.250,1.48)(0.260,1.58)(0.270,1.68)(0.280,1.79)(0.290,1.89)(0.300,2.00)(0.310,2.10)(0.320,2.21)(0.330,2.32)(0.340,2.42)(0.350,2.53)(0.360,2.63)(0.370,2.73)(0.380,2.84)(0.390,2.94)(0.400,3.03)(0.410,3.13)(0.420,3.22)(0.430,3.31)(0.440,3.40)(0.450,3.49)(0.460,3.57)(0.470,3.64)(0.480,3.72)(0.490,3.79)(0.500,3.86)(0.510,3.92)(0.520,3.98)(0.530,4.04)(0.540,4.09)(0.550,4.13)(0.560,4.18)(0.570,4.21)(0.580,4.25)(0.590,4.28)(0.600,4.30)(0.610,4.32)(0.620,4.34)(0.630,4.35)(0.640,4.35)(0.650,4.36)(0.660,4.35)(0.670,4.35)(0.680,4.34)(0.690,4.32)(0.700,4.30)(0.710,4.28)(0.720,4.26)(0.730,4.23)(0.740,4.19)(0.750,4.16)(0.760,4.12)(0.770,4.08)(0.780,4.03)(0.790,3.98)(0.800,3.93)(0.810,3.88)(0.820,3.83)(0.830,3.77)(0.840,3.72)(0.850,3.66)(0.860,3.60)(0.870,3.54)(0.880,3.47)(0.890,3.41)(0.900,3.35)(0.910,3.29)(0.920,3.22)(0.930,3.16)(0.940,3.10)(0.950,3.04)(0.960,2.98)(0.970,2.92)(0.980,2.86)(0.990,2.80)(1.00,2.75)(1.01,2.69)(1.02,2.64)(1.03,2.59)(1.04,2.54)(1.05,2.49)(1.06,2.45)(1.07,2.41)(1.08,2.37)(1.09,2.33)(1.10,2.30)(1.11,2.27)(1.12,2.24)(1.13,2.22)(1.14,2.19)(1.15,2.17)(1.16,2.16)(1.17,2.14)(1.18,2.13)(1.19,2.12)(1.20,2.12)(1.21,2.12)(1.22,2.12)(1.23,2.12)(1.24,2.13)(1.25,2.14)(1.26,2.15)(1.27,2.16)(1.28,2.18)(1.29,2.19)(1.30,2.21)(1.31,2.24)(1.32,2.26)(1.33,2.29)(1.34,2.31)(1.35,2.34)(1.36,2.37)(1.37,2.41)(1.38,2.44)(1.39,2.47)(1.40,2.51)(1.41,2.54)(1.42,2.58)(1.43,2.61)(1.44,2.65)(1.45,2.69)(1.46,2.72)(1.47,2.76)(1.48,2.80)(1.49,2.83)(1.50,2.87)(1.51,2.90)(1.52,2.94)(1.53,2.97)(1.54,3.00)(1.55,3.03)(1.56,3.06)(1.57,3.09)(1.58,3.12)(1.59,3.14)(1.60,3.16)(1.61,3.19)(1.62,3.21)(1.63,3.22)(1.64,3.24)(1.65,3.25)(1.66,3.27)(1.67,3.28)(1.68,3.29)(1.69,3.29)(1.70,3.30)(1.71,3.30)(1.72,3.30)(1.73,3.30)(1.74,3.30)(1.75,3.29)(1.76,3.28)(1.77,3.27)(1.78,3.26)(1.79,3.25)(1.80,3.24)(1.81,3.22)(1.82,3.20)(1.83,3.19)(1.84,3.17)(1.85,3.14)(1.86,3.12)(1.87,3.10)(1.88,3.07)(1.89,3.05)(1.90,3.02)(1.91,3.00)(1.92,2.97)(1.93,2.94)(1.94,2.91)(1.95,2.89)(1.96,2.86)(1.97,2.83)(1.98,2.80)(1.99,2.77)(2.00,2.75)(2.01,2.72)(2.02,2.69)(2.03,2.67)(2.04,2.64)(2.05,2.62)(2.06,2.59)(2.07,2.57)(2.08,2.55)(2.09,2.53)(2.10,2.51)(2.11,2.49)(2.12,2.47)(2.13,2.46)(2.14,2.44)(2.15,2.43)(2.16,2.42)(2.17,2.41)(2.18,2.40)(2.19,2.39)(2.20,2.39)(2.21,2.38)(2.22,2.38)(2.23,2.38)(2.24,2.38)(2.25,2.38)(2.26,2.39)(2.27,2.39)(2.28,2.40)(2.29,2.41)(2.30,2.42)(2.31,2.43)(2.32,2.44)(2.33,2.45)(2.34,2.47)(2.35,2.48)(2.36,2.50)(2.37,2.52)(2.38,2.54)(2.39,2.56)(2.40,2.57)(2.41,2.60)(2.42,2.62)(2.43,2.64)(2.44,2.66)(2.45,2.68)(2.46,2.70)(2.47,2.72)(2.48,2.75)(2.49,2.77)(2.50,2.79)(2.51,2.81)(2.52,2.83)(2.53,2.85)(2.54,2.87)(2.55,2.89)(2.56,2.91)(2.57,2.93)(2.58,2.95)(2.59,2.96)(2.60,2.98)(2.61,2.99)(2.62,3.01)(2.63,3.02)(2.64,3.03)(2.65,3.04)(2.66,3.05)(2.67,3.06)(2.68,3.07)(2.69,3.07)(2.70,3.08)(2.71,3.08)(2.72,3.08)(2.73,3.08)(2.74,3.08)(2.75,3.08)(2.76,3.08)(2.77,3.07)(2.78,3.07)(2.79,3.06)(2.80,3.05)(2.81,3.05)(2.82,3.04)(2.83,3.02)(2.84,3.01)(2.85,3.00)(2.86,2.99)(2.87,2.97)(2.88,2.96)(2.89,2.94)(2.90,2.93)(2.91,2.91)(2.92,2.89)(2.93,2.87)(2.94,2.86)(2.95,2.84)(2.96,2.82)(2.97,2.80)(2.98,2.78)(2.99,2.76)(3.00,2.75)
};
\addplot [
draw=red,ultra thick
] coordinates{(0.010,0.00164)(0.0200,0.00661)(0.0300,0.0150)(0.0400,0.0270)(0.0500,0.0426)(0.0600,0.0619)(0.0700,0.0849)(0.0800,0.112)(0.0900,0.143)(0.100,0.177)(0.110,0.216)(0.120,0.258)(0.130,0.305)(0.140,0.355)(0.150,0.409)(0.160,0.467)(0.170,0.529)(0.180,0.595)(0.190,0.664)(0.200,0.737)(0.210,0.813)(0.220,0.893)(0.230,0.976)(0.240,1.06)(0.250,1.15)(0.260,1.24)(0.270,1.34)(0.280,1.43)(0.290,1.53)(0.300,1.63)(0.310,1.74)(0.320,1.84)(0.330,1.95)(0.340,2.06)(0.350,2.16)(0.360,2.27)(0.370,2.38)(0.380,2.49)(0.390,2.60)(0.400,2.71)(0.410,2.82)(0.420,2.93)(0.430,3.04)(0.440,3.14)(0.450,3.25)(0.460,3.35)(0.470,3.45)(0.480,3.55)(0.490,3.65)(0.500,3.74)(0.510,3.83)(0.520,3.92)(0.530,4.01)(0.540,4.09)(0.550,4.16)(0.560,4.24)(0.570,4.31)(0.580,4.37)(0.590,4.43)(0.600,4.49)(0.610,4.54)(0.620,4.58)(0.630,4.63)(0.640,4.66)(0.650,4.70)(0.660,4.72)(0.670,4.74)(0.680,4.76)(0.690,4.77)(0.700,4.78)(0.710,4.78)(0.720,4.78)(0.730,4.77)(0.740,4.75)(0.750,4.74)(0.760,4.71)(0.770,4.69)(0.780,4.65)(0.790,4.62)(0.800,4.58)(0.810,4.53)(0.820,4.49)(0.830,4.43)(0.840,4.38)(0.850,4.32)(0.860,4.26)(0.870,4.20)(0.880,4.13)(0.890,4.06)(0.900,3.99)(0.910,3.92)(0.920,3.85)(0.930,3.77)(0.940,3.70)(0.950,3.62)(0.960,3.55)(0.970,3.47)(0.980,3.40)(0.990,3.32)(1.00,3.25)(1.01,3.17)(1.02,3.10)(1.03,3.03)(1.04,2.96)(1.05,2.90)(1.06,2.84)(1.07,2.77)(1.08,2.72)(1.09,2.66)(1.10,2.61)(1.11,2.56)(1.12,2.52)(1.13,2.48)(1.14,2.44)(1.15,2.41)(1.16,2.38)(1.17,2.36)(1.18,2.34)(1.19,2.33)(1.20,2.32)(1.21,2.31)(1.22,2.31)(1.23,2.31)(1.24,2.32)(1.25,2.34)(1.26,2.35)(1.27,2.38)(1.28,2.40)(1.29,2.44)(1.30,2.47)(1.31,2.51)(1.32,2.56)(1.33,2.60)(1.34,2.66)(1.35,2.71)(1.36,2.77)(1.37,2.83)(1.38,2.90)(1.39,2.97)(1.40,3.04)(1.41,3.11)(1.42,3.19)(1.43,3.27)(1.44,3.34)(1.45,3.42)(1.46,3.51)(1.47,3.59)(1.48,3.67)(1.49,3.76)(1.50,3.84)(1.51,3.92)(1.52,4.00)(1.53,4.09)(1.54,4.17)(1.55,4.25)(1.56,4.33)(1.57,4.40)(1.58,4.48)(1.59,4.55)(1.60,4.62)(1.61,4.69)(1.62,4.75)(1.63,4.81)(1.64,4.87)(1.65,4.92)(1.66,4.98)(1.67,5.02)(1.68,5.07)(1.69,5.11)(1.70,5.14)(1.71,5.17)(1.72,5.20)(1.73,5.22)(1.74,5.24)(1.75,5.26)(1.76,5.27)(1.77,5.27)(1.78,5.28)(1.79,5.27)(1.80,5.27)(1.81,5.26)(1.82,5.24)(1.83,5.22)(1.84,5.20)(1.85,5.18)(1.86,5.15)(1.87,5.11)(1.88,5.08)(1.89,5.04)(1.90,5.00)(1.91,4.96)(1.92,4.91)(1.93,4.86)(1.94,4.81)(1.95,4.76)(1.96,4.71)(1.97,4.65)(1.98,4.60)(1.99,4.54)(2.00,4.49)(2.01,4.43)(2.02,4.37)(2.03,4.32)(2.04,4.26)(2.05,4.21)(2.06,4.16)(2.07,4.11)(2.08,4.06)(2.09,4.01)(2.10,3.96)(2.11,3.92)(2.12,3.88)(2.13,3.84)(2.14,3.80)(2.15,3.77)(2.16,3.74)(2.17,3.72)(2.18,3.69)(2.19,3.68)(2.20,3.66)(2.21,3.65)(2.22,3.64)(2.23,3.64)(2.24,3.64)(2.25,3.65)(2.26,3.65)(2.27,3.67)(2.28,3.68)(2.29,3.71)(2.30,3.73)(2.31,3.76)(2.32,3.79)(2.33,3.83)(2.34,3.87)(2.35,3.92)(2.36,3.97)(2.37,4.02)(2.38,4.07)(2.39,4.13)(2.40,4.19)(2.41,4.25)(2.42,4.32)(2.43,4.39)(2.44,4.46)(2.45,4.53)(2.46,4.60)(2.47,4.68)(2.48,4.75)(2.49,4.83)(2.50,4.91)(2.51,4.98)(2.52,5.06)(2.53,5.14)(2.54,5.22)(2.55,5.29)(2.56,5.37)(2.57,5.44)(2.58,5.52)(2.59,5.59)(2.60,5.66)(2.61,5.72)(2.62,5.79)(2.63,5.85)(2.64,5.91)(2.65,5.97)(2.66,6.03)(2.67,6.08)(2.68,6.13)(2.69,6.17)(2.70,6.21)(2.71,6.25)(2.72,6.28)(2.73,6.31)(2.74,6.34)(2.75,6.36)(2.76,6.38)(2.77,6.40)(2.78,6.41)(2.79,6.41)(2.80,6.41)(2.81,6.41)(2.82,6.41)(2.83,6.40)(2.84,6.38)(2.85,6.37)(2.86,6.35)(2.87,6.32)(2.88,6.30)(2.89,6.27)(2.90,6.24)(2.91,6.20)(2.92,6.16)(2.93,6.12)(2.94,6.08)(2.95,6.03)(2.96,5.99)(2.97,5.94)(2.98,5.89)(2.99,5.84)(3.00,5.79)
 };

\legend{Truncated Mie series,Van de Hulst approximation,Evans--Fournier approximation,Semigroup asymptotics};

\node[fill=white,anchor=south,draw=black] at
(axis cs:1.38,.25) {\tiny$ kR=\pi$};
\end{axis}
\end{tikzpicture}
\vspace{-.75em}
\begin{center}\begin{tiny}\quad\,(a)\end{tiny}\end{center}\begin{tikzpicture}\pgfplotsset{xlabel style={yshift=.3cm}, ylabel style={yshift=-.8cm},width=10.5cm,height=5cm,xmajorgrids,ymajorgrids, tick label style={font=\tiny},label style={font=\tiny}}
\begin{axis}[xlabel={$n-1$},ylabel={$ \sigma_N=\frac{\sigma_{\text{sc}}}{\pi R^2}$},ymin=0,ymax=8,xmin=0,xmax=1.5,enlargelimits=false,  minor y tick num =1, minor x tick num =4,ytick={0,2,...,10},xtick={0,.5,...,2},legend style={
legend columns=1,
cells={anchor=west},at={(.98,.95)},
font=\tiny,legend style={row sep=-2.5pt},
},
   restrict y to domain=0:6]\addplot [only marks,
draw=green!50!black,mark=o,thin,mark size=1.5pt
] coordinates{(0.0200,0.0301201)(0.0400,0.121823)(0.0600,0.27499)(0.0800,0.486576)(0.100,0.750767)(0.120,1.05942)(0.140,1.40265)(0.160,1.76911)(0.180,2.14585)(0.200,2.51774)(0.220,2.86774)(0.240,3.17885)(0.260,3.43791)(0.280,3.63909)(0.300,3.78468)(0.320,3.88243)(0.340,3.93959)(0.360,3.9545)(0.380,3.90919)(0.400,3.77686)(0.420,3.54985)(0.440,3.25556)(0.460,2.9369)(0.480,2.62817)(0.500,2.35138)(0.520,2.12665)(0.540,1.98881)(0.560,1.99397)(0.580,2.08532)(0.600,1.95531)(0.620,1.77521)(0.640,1.79833)(0.660,2.00233)(0.680,2.27633)(0.700,2.44031)(0.720,2.43221)(0.740,2.39476)(0.760,2.45103)(0.780,2.8606)(0.800,3.1213)(0.820,2.96289)(0.840,3.05939)(0.860,3.11256)(0.880,3.10567)(0.900,3.07261)(0.920,3.36719)(0.940,2.86948)(0.960,2.50834)(0.980,2.50777)(1.00,2.65564)(1.02,2.6191)(1.04,2.4523)(1.06,2.27009)(1.08,2.09758)(1.10,1.9249)(1.12,1.7688)(1.14,2.52901)(1.16,1.74835)(1.18,2.01014)(1.20,2.62795)(1.22,2.68457)(1.24,2.50656)(1.26,2.47776)(1.28,2.5128)(1.30,2.63832)(1.32,2.8735)(1.34,3.08395)(1.36,3.14823)(1.38,3.13941)(1.40,3.15866)(1.42,3.24048)(1.44,2.6694)(1.46,2.35036)(1.48,2.11523)(1.50,2.00891)
};
\addplot [
draw=blue,thin
] coordinates{(0.010,0.00789)(0.0200,0.0315)(0.0300,0.0705)(0.0400,0.125)(0.0500,0.193)(0.0600,0.275)(0.0700,0.371)(0.0800,0.478)(0.0900,0.596)(0.100,0.723)(0.110,0.859)(0.120,1.00)(0.130,1.15)(0.140,1.30)(0.150,1.46)(0.160,1.61)(0.170,1.77)(0.180,1.92)(0.190,2.07)(0.200,2.21)(0.210,2.35)(0.220,2.48)(0.230,2.60)(0.240,2.71)(0.250,2.81)(0.260,2.90)(0.270,2.98)(0.280,3.04)(0.290,3.09)(0.300,3.13)(0.310,3.16)(0.320,3.17)(0.330,3.17)(0.340,3.16)(0.350,3.14)(0.360,3.10)(0.370,3.06)(0.380,3.00)(0.390,2.94)(0.400,2.87)(0.410,2.79)(0.420,2.71)(0.430,2.62)(0.440,2.53)(0.450,2.44)(0.460,2.35)(0.470,2.26)(0.480,2.17)(0.490,2.08)(0.500,2.00)(0.510,1.92)(0.520,1.85)(0.530,1.79)(0.540,1.73)(0.550,1.68)(0.560,1.63)(0.570,1.60)(0.580,1.57)(0.590,1.55)(0.600,1.54)(0.610,1.54)(0.620,1.55)(0.630,1.56)(0.640,1.58)(0.650,1.61)(0.660,1.65)(0.670,1.69)(0.680,1.73)(0.690,1.78)(0.700,1.83)(0.710,1.88)(0.720,1.93)(0.730,1.99)(0.740,2.04)(0.750,2.09)(0.760,2.14)(0.770,2.19)(0.780,2.23)(0.790,2.27)(0.800,2.31)(0.810,2.34)(0.820,2.36)(0.830,2.38)(0.840,2.39)(0.850,2.40)(0.860,2.40)(0.870,2.40)(0.880,2.39)(0.890,2.38)(0.900,2.36)(0.910,2.33)(0.920,2.31)(0.930,2.27)(0.940,2.24)(0.950,2.20)(0.960,2.16)(0.970,2.12)(0.980,2.08)(0.990,2.04)(1.00,2.00)(1.01,1.96)(1.02,1.92)(1.03,1.89)(1.04,1.86)(1.05,1.83)(1.06,1.80)(1.07,1.78)(1.08,1.76)(1.09,1.75)(1.10,1.74)(1.11,1.74)(1.12,1.74)(1.13,1.74)(1.14,1.75)(1.15,1.76)(1.16,1.78)(1.17,1.80)(1.18,1.82)(1.19,1.85)(1.20,1.88)(1.21,1.91)(1.22,1.94)(1.23,1.97)(1.24,2.00)(1.25,2.03)(1.26,2.06)(1.27,2.09)(1.28,2.12)(1.29,2.15)(1.30,2.17)(1.31,2.19)(1.32,2.21)(1.33,2.22)(1.34,2.24)(1.35,2.24)(1.36,2.25)(1.37,2.25)(1.38,2.24)(1.39,2.24)(1.40,2.23)(1.41,2.21)(1.42,2.20)(1.43,2.18)(1.44,2.15)(1.45,2.13)(1.46,2.11)(1.47,2.08)(1.48,2.05)(1.49,2.03)(1.50,2.00)(1.51,1.97)(1.52,1.95)(1.53,1.92)(1.54,1.90)(1.55,1.88)(1.56,1.86)(1.57,1.85)(1.58,1.83)(1.59,1.82)(1.60,1.82)(1.61,1.81)(1.62,1.81)(1.63,1.82)(1.64,1.82)(1.65,1.83)(1.66,1.84)(1.67,1.85)(1.68,1.87)(1.69,1.89)(1.70,1.91)(1.71,1.93)(1.72,1.95)(1.73,1.97)(1.74,1.99)(1.75,2.02)(1.76,2.04)(1.77,2.06)(1.78,2.08)(1.79,2.10)(1.80,2.12)(1.81,2.13)(1.82,2.15)(1.83,2.16)(1.84,2.17)(1.85,2.17)(1.86,2.18)(1.87,2.18)(1.88,2.18)(1.89,2.17)(1.90,2.16)(1.91,2.15)(1.92,2.14)(1.93,2.13)(1.94,2.11)(1.95,2.10)(1.96,2.08)(1.97,2.06)(1.98,2.04)(1.99,2.02)(2.00,2.00)(2.01,1.98)(2.02,1.96)(2.03,1.94)(2.04,1.93)(2.05,1.91)(2.06,1.90)(2.07,1.88)(2.08,1.87)(2.09,1.87)(2.10,1.86)(2.11,1.86)(2.12,1.86)(2.13,1.86)(2.14,1.86)(2.15,1.87)(2.16,1.87)(2.17,1.88)(2.18,1.90)(2.19,1.91)(2.20,1.92)(2.21,1.94)(2.22,1.96)(2.23,1.97)(2.24,1.99)(2.25,2.01)(2.26,2.03)(2.27,2.04)(2.28,2.06)(2.29,2.08)(2.30,2.09)(2.31,2.10)(2.32,2.11)(2.33,2.12)(2.34,2.13)(2.35,2.13)(2.36,2.14)(2.37,2.14)(2.38,2.14)(2.39,2.13)(2.40,2.13)(2.41,2.12)(2.42,2.11)(2.43,2.10)(2.44,2.09)(2.45,2.08)(2.46,2.06)(2.47,2.05)(2.48,2.03)(2.49,2.02)(2.50,2.00)(2.51,1.98)(2.52,1.97)(2.53,1.95)(2.54,1.94)(2.55,1.93)(2.56,1.92)(2.57,1.91)(2.58,1.90)(2.59,1.89)(2.60,1.89)(2.61,1.88)(2.62,1.88)(2.63,1.88)(2.64,1.89)(2.65,1.89)(2.66,1.90)(2.67,1.90)(2.68,1.91)(2.69,1.93)(2.70,1.94)(2.71,1.95)(2.72,1.96)(2.73,1.98)(2.74,1.99)(2.75,2.01)(2.76,2.02)(2.77,2.04)(2.78,2.05)(2.79,2.06)(2.80,2.07)(2.81,2.08)(2.82,2.09)(2.83,2.10)(2.84,2.11)(2.85,2.11)(2.86,2.11)(2.87,2.11)(2.88,2.11)(2.89,2.11)(2.90,2.11)(2.91,2.10)(2.92,2.09)(2.93,2.08)(2.94,2.07)(2.95,2.06)(2.96,2.05)(2.97,2.04)(2.98,2.03)(2.99,2.01)(3.00,2.00)
};\addplot [
draw=blue,thin,dashed
] coordinates{(0.010,0.00990)(0.0200,0.0395)(0.0300,0.0884)(0.0400,0.156)(0.0500,0.242)(0.0600,0.345)(0.0700,0.465)(0.0800,0.599)(0.0900,0.747)(0.100,0.907)(0.110,1.08)(0.120,1.26)(0.130,1.44)(0.140,1.63)(0.150,1.83)(0.160,2.02)(0.170,2.21)(0.180,2.41)(0.190,2.59)(0.200,2.77)(0.210,2.94)(0.220,3.11)(0.230,3.26)(0.240,3.40)(0.250,3.53)(0.260,3.64)(0.270,3.73)(0.280,3.82)(0.290,3.88)(0.300,3.93)(0.310,3.96)(0.320,3.98)(0.330,3.98)(0.340,3.96)(0.350,3.93)(0.360,3.89)(0.370,3.83)(0.380,3.76)(0.390,3.68)(0.400,3.60)(0.410,3.50)(0.420,3.40)(0.430,3.29)(0.440,3.17)(0.450,3.06)(0.460,2.95)(0.470,2.83)(0.480,2.72)(0.490,2.61)(0.500,2.51)(0.510,2.41)(0.520,2.32)(0.530,2.24)(0.540,2.17)(0.550,2.10)(0.560,2.05)(0.570,2.00)(0.580,1.97)(0.590,1.95)(0.600,1.94)(0.610,1.94)(0.620,1.94)(0.630,1.96)(0.640,1.99)(0.650,2.02)(0.660,2.07)(0.670,2.11)(0.680,2.17)(0.690,2.23)(0.700,2.29)(0.710,2.36)(0.720,2.42)(0.730,2.49)(0.740,2.56)(0.750,2.62)(0.760,2.68)(0.770,2.74)(0.780,2.80)(0.790,2.85)(0.800,2.89)(0.810,2.93)(0.820,2.96)(0.830,2.99)(0.840,3.00)(0.850,3.01)(0.860,3.02)(0.870,3.01)(0.880,3.00)(0.890,2.98)(0.900,2.96)(0.910,2.93)(0.920,2.89)(0.930,2.85)(0.940,2.81)(0.950,2.76)(0.960,2.71)(0.970,2.66)(0.980,2.61)(0.990,2.56)(1.00,2.51)(1.01,2.46)(1.02,2.41)(1.03,2.37)(1.04,2.33)(1.05,2.29)(1.06,2.26)(1.07,2.23)(1.08,2.21)(1.09,2.19)(1.10,2.18)(1.11,2.18)(1.12,2.18)(1.13,2.18)(1.14,2.19)(1.15,2.21)(1.16,2.23)(1.17,2.26)(1.18,2.29)(1.19,2.32)(1.20,2.35)(1.21,2.39)(1.22,2.43)(1.23,2.47)(1.24,2.51)(1.25,2.55)(1.26,2.59)(1.27,2.63)(1.28,2.66)(1.29,2.69)(1.30,2.72)(1.31,2.75)(1.32,2.77)(1.33,2.79)(1.34,2.80)(1.35,2.81)(1.36,2.82)(1.37,2.82)(1.38,2.81)(1.39,2.80)(1.40,2.79)(1.41,2.77)(1.42,2.75)(1.43,2.73)(1.44,2.70)(1.45,2.67)(1.46,2.64)(1.47,2.61)(1.48,2.58)(1.49,2.54)(1.50,2.51)(1.51,2.48)(1.52,2.44)(1.53,2.41)(1.54,2.39)(1.55,2.36)(1.56,2.34)(1.57,2.32)(1.58,2.30)(1.59,2.29)(1.60,2.28)(1.61,2.28)(1.62,2.27)(1.63,2.28)(1.64,2.28)(1.65,2.29)(1.66,2.31)(1.67,2.32)(1.68,2.34)(1.69,2.37)(1.70,2.39)(1.71,2.42)(1.72,2.44)(1.73,2.47)(1.74,2.50)(1.75,2.53)(1.76,2.56)(1.77,2.59)(1.78,2.61)(1.79,2.64)(1.80,2.66)(1.81,2.68)(1.82,2.69)(1.83,2.71)(1.84,2.72)(1.85,2.73)(1.86,2.73)(1.87,2.73)(1.88,2.73)(1.89,2.72)(1.90,2.71)(1.91,2.70)(1.92,2.69)(1.93,2.67)(1.94,2.65)(1.95,2.63)(1.96,2.61)(1.97,2.58)(1.98,2.56)(1.99,2.53)(2.00,2.51)(2.01,2.48)(2.02,2.46)(2.03,2.44)(2.04,2.42)(2.05,2.40)(2.06,2.38)(2.07,2.36)(2.08,2.35)(2.09,2.34)(2.10,2.33)(2.11,2.33)(2.12,2.33)(2.13,2.33)(2.14,2.33)(2.15,2.34)(2.16,2.35)(2.17,2.36)(2.18,2.38)(2.19,2.40)(2.20,2.41)(2.21,2.43)(2.22,2.46)(2.23,2.48)(2.24,2.50)(2.25,2.52)(2.26,2.54)(2.27,2.56)(2.28,2.59)(2.29,2.60)(2.30,2.62)(2.31,2.64)(2.32,2.65)(2.33,2.66)(2.34,2.67)(2.35,2.68)(2.36,2.68)(2.37,2.68)(2.38,2.68)(2.39,2.68)(2.40,2.67)(2.41,2.66)(2.42,2.65)(2.43,2.64)(2.44,2.62)(2.45,2.61)(2.46,2.59)(2.47,2.57)(2.48,2.55)(2.49,2.53)(2.50,2.51)(2.51,2.49)(2.52,2.47)(2.53,2.45)(2.54,2.43)(2.55,2.42)(2.56,2.40)(2.57,2.39)(2.58,2.38)(2.59,2.37)(2.60,2.37)(2.61,2.36)(2.62,2.36)(2.63,2.36)(2.64,2.37)(2.65,2.37)(2.66,2.38)(2.67,2.39)(2.68,2.40)(2.69,2.41)(2.70,2.43)(2.71,2.45)(2.72,2.46)(2.73,2.48)(2.74,2.50)(2.75,2.52)(2.76,2.54)(2.77,2.55)(2.78,2.57)(2.79,2.59)(2.80,2.60)(2.81,2.61)(2.82,2.62)(2.83,2.63)(2.84,2.64)(2.85,2.65)(2.86,2.65)(2.87,2.65)(2.88,2.65)(2.89,2.65)(2.90,2.64)(2.91,2.64)(2.92,2.63)(2.93,2.62)(2.94,2.60)(2.95,2.59)(2.96,2.57)(2.97,2.56)(2.98,2.54)(2.99,2.53)(3.00,2.51)
};
\addplot [
draw=red,ultra thick
] coordinates{(0.010,0.00746)(0.0200,0.0301)(0.0300,0.0681)(0.0400,0.122)(0.0500,0.190)(0.0600,0.274)(0.0700,0.373)(0.0800,0.485)(0.0900,0.611)(0.100,0.749)(0.110,0.899)(0.120,1.06)(0.130,1.23)(0.140,1.40)(0.150,1.58)(0.160,1.77)(0.170,1.95)(0.180,2.14)(0.190,2.33)(0.200,2.51)(0.210,2.69)(0.220,2.87)(0.230,3.03)(0.240,3.19)(0.250,3.33)(0.260,3.47)(0.270,3.59)(0.280,3.69)(0.290,3.78)(0.300,3.86)(0.310,3.92)(0.320,3.96)(0.330,3.99)(0.340,4.00)(0.350,3.99)(0.360,3.97)(0.370,3.93)(0.380,3.88)(0.390,3.82)(0.400,3.74)(0.410,3.65)(0.420,3.56)(0.430,3.45)(0.440,3.34)(0.450,3.22)(0.460,3.10)(0.470,2.98)(0.480,2.86)(0.490,2.74)(0.500,2.63)(0.510,2.52)(0.520,2.42)(0.530,2.32)(0.540,2.23)(0.550,2.16)(0.560,2.09)(0.570,2.04)(0.580,2.00)(0.590,1.98)(0.600,1.96)(0.610,1.96)(0.620,1.98)(0.630,2.00)(0.640,2.04)(0.650,2.09)(0.660,2.15)(0.670,2.22)(0.680,2.31)(0.690,2.39)(0.700,2.49)(0.710,2.59)(0.720,2.69)(0.730,2.80)(0.740,2.91)(0.750,3.01)(0.760,3.12)(0.770,3.22)(0.780,3.32)(0.790,3.41)(0.800,3.49)(0.810,3.57)(0.820,3.63)(0.830,3.69)(0.840,3.73)(0.850,3.77)(0.860,3.79)(0.870,3.80)(0.880,3.81)(0.890,3.80)(0.900,3.78)(0.910,3.75)(0.920,3.71)(0.930,3.66)(0.940,3.61)(0.950,3.54)(0.960,3.48)(0.970,3.41)(0.980,3.33)(0.990,3.26)(1.00,3.18)(1.01,3.11)(1.02,3.04)(1.03,2.97)(1.04,2.91)(1.05,2.85)(1.06,2.80)(1.07,2.76)(1.08,2.72)(1.09,2.70)(1.10,2.68)(1.11,2.68)(1.12,2.68)(1.13,2.70)(1.14,2.72)(1.15,2.76)(1.16,2.80)(1.17,2.85)(1.18,2.91)(1.19,2.98)(1.20,3.06)(1.21,3.14)(1.22,3.22)(1.23,3.31)(1.24,3.40)(1.25,3.49)(1.26,3.58)(1.27,3.67)(1.28,3.75)(1.29,3.83)(1.30,3.91)(1.31,3.98)(1.32,4.04)(1.33,4.10)(1.34,4.14)(1.35,4.18)(1.36,4.21)(1.37,4.23)(1.38,4.24)(1.39,4.24)(1.40,4.23)(1.41,4.21)(1.42,4.19)(1.43,4.15)(1.44,4.11)(1.45,4.06)(1.46,4.01)(1.47,3.95)(1.48,3.89)(1.49,3.83)(1.50,3.77)(1.51,3.71)(1.52,3.65)(1.53,3.59)(1.54,3.53)(1.55,3.48)(1.56,3.44)(1.57,3.40)(1.58,3.37)(1.59,3.34)(1.60,3.33)(1.61,3.32)(1.62,3.33)(1.63,3.34)(1.64,3.36)(1.65,3.39)(1.66,3.43)(1.67,3.47)(1.68,3.53)(1.69,3.59)(1.70,3.66)(1.71,3.73)(1.72,3.80)(1.73,3.88)(1.74,3.97)(1.75,4.05)(1.76,4.13)(1.77,4.22)(1.78,4.30)(1.79,4.37)(1.80,4.44)(1.81,4.51)(1.82,4.57)(1.83,4.63)(1.84,4.67)(1.85,4.71)(1.86,4.74)(1.87,4.77)(1.88,4.78)(1.89,4.78)(1.90,4.78)(1.91,4.77)(1.92,4.75)(1.93,4.72)(1.94,4.68)(1.95,4.64)(1.96,4.60)(1.97,4.55)(1.98,4.49)(1.99,4.44)(2.00,4.38)(2.01,4.32)(2.02,4.27)(2.03,4.21)(2.04,4.16)(2.05,4.11)(2.06,4.07)(2.07,4.04)(2.08,4.01)(2.09,3.98)(2.10,3.97)(2.11,3.96)(2.12,3.96)(2.13,3.97)(2.14,3.99)(2.15,4.02)(2.16,4.05)(2.17,4.10)(2.18,4.15)(2.19,4.21)(2.20,4.27)(2.21,4.34)(2.22,4.41)(2.23,4.49)(2.24,4.56)(2.25,4.64)(2.26,4.72)(2.27,4.80)(2.28,4.88)(2.29,4.96)(2.30,5.03)(2.31,5.09)(2.32,5.15)(2.33,5.21)(2.34,5.25)(2.35,5.29)(2.36,5.33)(2.37,5.35)(2.38,5.37)(2.39,5.37)(2.40,5.37)(2.41,5.36)(2.42,5.34)(2.43,5.32)(2.44,5.29)(2.45,5.25)(2.46,5.21)(2.47,5.16)(2.48,5.11)(2.49,5.06)(2.50,5.01)(2.51,4.95)(2.52,4.90)(2.53,4.85)(2.54,4.80)(2.55,4.75)(2.56,4.71)(2.57,4.68)(2.58,4.65)(2.59,4.63)(2.60,4.61)(2.61,4.60)(2.62,4.60)(2.63,4.61)(2.64,4.63)(2.65,4.65)(2.66,4.69)(2.67,4.73)(2.68,4.78)(2.69,4.83)(2.70,4.89)(2.71,4.96)(2.72,5.03)(2.73,5.10)(2.74,5.18)(2.75,5.26)(2.76,5.34)(2.77,5.41)(2.78,5.49)(2.79,5.56)(2.80,5.63)(2.81,5.70)(2.82,5.76)(2.83,5.81)(2.84,5.86)(2.85,5.90)(2.86,5.93)(2.87,5.96)(2.88,5.98)(2.89,5.99)(2.90,5.99)(2.91,5.98)(2.92,5.96)(2.93,5.94)(2.94,5.91)(2.95,5.88)(2.96,5.84)(2.97,5.79)(2.98,5.75)(2.99,5.70)(3.00,5.64)
 };

\legend{Truncated Mie series,Van de Hulst approximation,Evans--Fournier approximation,Semigroup asymptotics};

\node[fill=white,anchor=south,draw=black] at
(axis cs:1.36,.25) {\tiny$ kR=2\pi$};
\end{axis}
\end{tikzpicture}
\vspace{-.75em}
\begin{center}\begin{tiny}\quad\,(b)\end{tiny}\end{center}\end{minipage}\begin{minipage}{.25\textwidth}\caption{Comparison of various  approximate solutions to the total scattering cross-section  of spherical scatterers (whose refractive indices range from that of  air $n=1$ to that of  diamond $n=5/2$),  with radii  \emph{(a)}~ $R=\pi/k $  and \emph{(b)}~$R=2\pi/k $.   For numerical computations, the Mie series \eqref{eq:Mie_sc_series} are truncated after the first  50 terms. \label{fig:5-3}}\end{minipage}\end{figure}

In addition to the observations above, the integral formulation of the asymptotic solution \eqref{eq:bulkapprox} also explicitly demonstrates the robustness of the light scattering problem against shape distortions. In practice, due to manufacture defects, dielectric geometries are almost never perfectly spherical in the problem of glass bead scattering for optical imaging in aqueous medium \cite{RICM}. The integral formula \eqref{eq:bulkapprox} helps us understand quantitatively how surface ruggedness may affect the scattering pattern, at least asymptotically. This advantage of the integral equation approach is also absent in analogs of Mie theory, which may hang on high symmetry for separation of variables. In view of the possible evaluation of certain volume and surface integrals in closed functional forms, we hope that the integral formula  \eqref{eq:bulkapprox}, which yields approximations in the non-perturbative regime, may find applications in both direct and inverse light scattering problems.

\section{Discussions}

In Papers I and II of this series, we have performed error analysis on perturbative solutions to the Born equation, and have developed a non-perturbative approach to solving light scattering problems of practical interest. There is an  underlying thread of thought that bridges  the perturbative method to the  non-perturbative counterpart. This key idea lies in the evolution semigroup of light scattering, and its asymptotic behavior.  The compactness of certain linear operators not only ensures a robust solution to light scattering problems in principle \cite{QualEM}, it also gives rise to, in practice, useful properties such as the strong stability of the evolution semigroup, which  allows asymptotic expansions beyond the perturbative regime of light scattering.

In prospect, one may develop some spectral methods for further enhancement of the non-perturbative solutions to electromagnetic scattering  in future research.
For example, for the Hilbert--Schmidt operator $ \hat {\mathscr G}(\hat I+2\hat{\mathscr G})^2$,
its functional determinant   is well defined \cite{Simon2005}. This may allow us   to construct analytic approximations to the light scattering problem that apply to an even wider range of dielectric susceptibilities. The presence of a functional determinant in the denominator of a non-perturbative solution would capture the optical resonance modes in the complex $ \chi$-plane where the Born equation ceases to be well-posed. This functional determinant method could probably outshine the semigroup approach presented in the current work, just in the way that the Pad\'e approximants beat polynomial approximations.
\subsection*{Acknowledgments}This research was supported in part  by the Applied Mathematics Program within the Department of Energy
(DOE) Office of Advanced Scientific Computing Research (ASCR) as part of the Collaboratory on
Mathematics for Mesoscopic Modeling of Materials (CM4).

The author thanks Prof.~Xiaowei Zhuang (Harvard University) for her  questions on nanophotonics   in 2006, which inspired the current work. Part of this research formed   Chapter 5 and Appendix D in the  author's  PhD thesis \cite{ZhouThesis} completed in January 2010 under her supervision. The author dedicates this paper to her 50th birthday.

The author is grateful to an anonymous referee for suggestions on improving the presentation of the current work.


\end{document}